\documentclass[12pt, reqno]{amsart}

\usepackage{amssymb,amsmath,amsthm,amsfonts,amscd,bbm}
\usepackage{enumitem}
\usepackage{pdflscape}% rotating the page
\usepackage[paper=portrait,pagesize]{typearea}
\usepackage{caption}
\usepackage{bm}
\usepackage{stmaryrd}

\usepackage{geometry} % change margin

\usepackage{ifpdf}
\ifpdf
\usepackage[pdftex]{graphicx}
\else
\usepackage[dvips]{graphicx}
\fi
\usepackage[all]{xy}
\usepackage{hyperref}
\usepackage[dvipsnames,table]{xcolor}
\usepackage{graphicx}
\usepackage[outdir=./]{epstopdf}
\usepackage{enumitem}
\usepackage{tensor}
\usepackage{tikz-cd}
\usepackage{multicol}
\usepackage{mathtools}
\usepackage{tocvsec2}
\usepackage{bbm}
\usepackage{longtable}
\usepackage{fullpage}

\newcommand{\doi}[1]{\href{https://dx.doi.org/#1}{{\small DOI:#1}}}

\newcommand{\googlebooks}[1]{(preview at \href{https://books.google.com/books?id=#1}{google books})}

\newcommand{\numdam}[1]{}
\usepackage{mathrsfs}

\usepackage{mathrsfs}
\DeclareMathAlphabet{\mathpzc}{OT1}{pzc}{m}{it}

\usepackage{imakeidx}
\makeindex[options=-s dotted, title=Index of Symbols, columns=1]

\def\semicolon{;}
\def\applytolist#1{
    \expandafter\def\csname multi#1\endcsname##1{
        \def\multiack{##1}\ifx\multiack\semicolon
            \def\next{\relax}
        \else
            \csname #1\endcsname{##1}
            \def\next{\csname multi#1\endcsname}
        \fi
        \next}
    \csname multi#1\endcsname}

\def\calc#1{\expandafter\def\csname c#1\endcsname{{\mathcal #1}}}
\applytolist{calc}QWERTYUIOPLKJHGFDSAZXCVBNM;
\def\bbc#1{\expandafter\def\csname bb#1\endcsname{{\mathbb #1}}}
\applytolist{bbc}QWERTYUIOPLKJHGFDSAZXCVBNM;
\def\bfc#1{\expandafter\def\csname bf#1\endcsname{{\mathbf #1}}}
\applytolist{bfc}QWERTYUIOPLKJHGFDSAZXCVBNM;
\def\sfc#1{\expandafter\def\csname s#1\endcsname{{\sf #1}}}
\applytolist{sfc}QWERTYUIOPLKJHGFDSAZXCVBNM;
\def\fc#1{\expandafter\def\csname f#1\endcsname{{\mathfrak #1}}}
\applytolist{fc}QWERTYUIOPLKJHGFDSAZXCVBNM;
\def\rmc#1{\expandafter\def\csname rm#1\endcsname{{\mathrm #1}}}
\applytolist{rmc}QWERTYUIOPLKJHGFDSAZXCVBNM;

\usepackage{tikz}
\usepackage{tikz-cd}

\usetikzlibrary{arrows,backgrounds,patterns.meta}
\usetikzlibrary{positioning,shadings,cd}
\usetikzlibrary{shapes}
\usetikzlibrary{backgrounds}
\usetikzlibrary{decorations,decorations.pathreplacing,decorations.markings,decorations.pathmorphing}
\usetikzlibrary{fit,calc,through}
\usetikzlibrary{external}
\usetikzlibrary{arrows}
\tikzset{vertex/.style = {shape=circle,draw,fill=black,inner sep=0pt,minimum size=5pt}}
\tikzset{edge/.style = {->,> = latex', bend right}}
\tikzset{
	super thick/.style={line width=3pt}
}
\tikzset{
    quadruple/.style args={[#1] in [#2] in [#3] in [#4]}{
        #1,preaction={preaction={preaction={draw,#4},draw,#3}, draw,#2}
    }
}
\tikzstyle{snake}=[decorate, decoration={snake, segment length=1mm, amplitude=.5mm}]
\tikzset{squiggly/.style={decorate, decoration=snake}}
\tikzstyle{knot}=[preaction={super thick, white, draw}]
\tikzstyle{shaded}=[fill=red!10!blue!20!gray!30!white]
\tikzstyle{unshaded}=[fill=white]
\tikzstyle{empty box}=[circle, draw, thick, fill=white, opaque, inner sep=2mm]
\tikzstyle{annular}=[scale=.7, inner sep=1mm, baseline]
\tikzstyle{rectangular}=[scale=.75, inner sep=1mm, baseline=-.1cm]
\tikzstyle{mid>}=[decoration={markings, mark=at position 0.5 with {\arrow{>}}}, postaction={decorate}]
\tikzstyle{far>}=[decoration={markings, mark=at position 0.65 with {\arrow{>}}}, postaction={decorate}]
\tikzstyle{mid<}=[decoration={markings, mark=at position 0.5 with {\arrow{<}}}, postaction={decorate}]
\tikzstyle{over}=[double, draw=white, super thick, double=]
\tikzstyle{snake}=[decorate, decoration={snake, segment length=1mm, amplitude=.3mm}]
\tikzstyle{saw}=[decorate, decoration={saw, segment length=.7mm, amplitude=.25mm}]

\tikzstyle{coupon}=[draw, very thick, rectangle, rounded corners=5pt]
\tikzset{Rightarrow/.style={double equal sign distance,>={Implies},->},
triplecd/.style={-,preaction={draw,Rightarrow}},
quadruplecd/.style={preaction={draw,Rightarrow,
shorten >=0pt
},
shorten >=1pt,
-,double,double
distance=0.2pt}}
\tikzset{
    tripleline/.style args={[#1] in [#2] in [#3]}{
        #1,preaction={preaction={draw,#3},draw,#2}
    }
}
\tikzstyle{triple}=[tripleline={[line width=.15mm,black] in
      [line width=.7mm,white] in
      [line width=1mm,black]}] 
\tikzset{
    quadrupleline/.style args={[#1] in [#2] in [#3] in [#4]}{
        #1,preaction={preaction={preaction={draw,#4},draw,#3}, draw,#2}
    }
}
\tikzstyle{quadruple}=[quadrupleline={[line width=.3mm,white] in
      [line width=.6mm,black] in
      [line width=1.2mm,white] in
      [line width=1.5mm,black]}]

\newcommand{\roundNbox}[6]{
	\draw[rounded corners=5pt, very thick, #1] ($#2+(-#3,-#3)+(-#4,0)$) rectangle ($#2+(#3,#3)+(#5,0)$);
	\coordinate (ZZa) at ($#2+(-#4,0)$);
	\coordinate (ZZb) at ($#2+(#5,0)$);
	\node at ($1/2*(ZZa)+1/2*(ZZb)$) {#6};
}

\newcommand{\levinHex}[5][]{
\coordinate (center) at (#2, #3);
\ifthenelse{\equal{#1}{}}{}{
\coordinate (pointD) at (canvas polar cs:angle=-30,radius=.866*#4cm);
\coordinate (pointE) at (canvas polar cs:angle=-30,radius=.288*#4cm);
\draw[#1] ($(center)+(pointD)$) -- +($#4*(-.333, .333)$);
}
\foreach \hexSideCounter in {1,2,3,4,5,6} {
\coordinate (pointA) at (canvas polar cs:angle={60*(\hexSideCounter - 1)},radius=#4cm);
\coordinate (pointB) at (canvas polar cs:angle={60*\hexSideCounter},radius=#4cm);
\coordinate (pointC) at (canvas polar cs:angle={60*(\hexSideCounter - 1)},radius=1.5*#4cm);
\draw[#5] ($(center)+(pointA)$) -- ($(center)+(pointB)$); \draw[#5] ($(center)+(pointA)$) -- ($(center)+(pointC)$);
}
}
\newcommand{\levinHexRow}[6][]{
\foreach \hexColumnCounter in {1,...,#4}
{
\levinHex[#1]{{#2 + (1.5 * #5)*(\hexColumnCounter -1)}}{{#3 + (.866 * #5)*(\hexColumnCounter -1)}}{#5}{#6}
}
}
\newcommand{\levinHexGrid}[7][]{
\foreach \a in {1,...,#4}
{
\levinHexRow[#1]{#2 + (1.5 * #6)*(\a - 1)}{#3 - (.866 * #6)*(\a - 1)}{#5}{#6}{#7}
}
}

\newcommand{\tikzmath}[2][]
     {\vcenter{\hbox{\begin{tikzpicture}[#1]#2
                     \end{tikzpicture}}}
     }

\newcommand{\levinHexOpen}[6][]{
\coordinate (center) at (#2, #3);
\ifthenelse{\equal{#1}{}}{}{
\coordinate (pointD) at (canvas polar cs:angle=-30,radius=.866*#4cm);
\coordinate (pointE) at (canvas polar cs:angle=-30,radius=.288*#4cm);
\draw[#1] ($(center)+(pointD)$) -- +($#4*(-.333, .333)$);
}
\foreach \hexSideCounter in {1,2,3,4,5,6} {
\coordinate (pointA) at (canvas polar cs:angle={60*(\hexSideCounter - 1)},radius=#4cm);
\coordinate (pointB) at (canvas polar cs:angle={60*\hexSideCounter},radius=#4cm);
\coordinate (pointC) at (canvas polar cs:angle={60*(\hexSideCounter - 1)},radius=1.5*#4cm);
\ifthenelse{#6=\hexSideCounter}{}{\draw[#5] ($(center)+(pointA)$) -- ($(center)+(pointB)$);}
\draw[#5] ($(center)+(pointA)$) -- ($(center)+(pointC)$);
}
}

\theoremstyle{plain}
\newtheorem{thm}{Theorem}[section]
\newtheorem*{thm*}{Theorem}
\newtheorem{thmalpha}{Theorem}

\newtheorem{cor}[thm]{Corollary}

\newtheorem*{cor*}{Corollary}

\newtheorem*{conj*}{Conjecture}
\newtheorem{lem}[thm]{Lemma}
\newtheorem*{lem*}{Lemma}

\newtheorem*{quest*}{Question}
\newtheorem*{claim*}{Claim}

\theoremstyle{definition}
\newtheorem{defn}[thm]{Definition}
\newtheorem{ansatz}[thmalpha]{Ansatz}

\newtheorem{ex}[thm]{Example}
\newtheorem{sub-ex}[thm]{Sub-Example}
\newtheorem{counter-ex}[thm]{Counter-Example}
\newtheorem{rem}[thm]{Remark}
\newtheorem*{rem*}{Remark}

\newtheorem{warn}[thm]{Warning}  
     
\usepackage{xcolor}
\definecolor{dark-red}{rgb}{0.7,0.25,0.25}
\definecolor{dark-blue}{rgb}{0.15,0.15,0.55}
\definecolor{medium-blue}{rgb}{0,0,.8}
\definecolor{DarkGreen}{RGB}{0,150,0}
\definecolor{rho}{named}{red}
\definecolor{OIorange}{HTML}{E69F00}
\definecolor{OIskyblue}{HTML}{56B4E9}
\definecolor{OIbluishgreen}{HTML}{009E73}
\definecolor{OIyellow}{HTML}{F0E442}
\definecolor{OIblue}{HTML}{0072B2}
\definecolor{OIvermillion}{HTML}{D55E00}
\definecolor{OIreddishpurple}{HTML}{CC79A7}
\definecolor{OIblack}{HTML}{000000}
\hypersetup{
   colorlinks, linkcolor={purple},
   citecolor={medium-blue}, urlcolor={medium-blue}
}
\definecolor{lightgray}{gray}{0.8}
\newcommand{\coev}{\operatorname{coev}}
\newcommand{\Dome}{\operatorname{Dome}}
\newcommand{\End}{\operatorname{End}}

\newcommand{\Forget}{\operatorname{Forget}}

\newcommand{\Hom}{\operatorname{Hom}}

\newcommand{\id}{\operatorname{id}}
\newcommand{\im}{\operatorname{im}}
\newcommand{\Irr}{\operatorname{Irr}}

\newcommand{\rev}{\operatorname{rev}}

\newcommand{\tr}{\operatorname{tr}}
\newcommand{\Tr}{\operatorname{Tr}}
\newcommand{\Tube}{\operatorname{Tube}}

\newcommand{\Stab}{\operatorname{Stab}}

\newcommand{\Bim}{\mathsf{Bim}}
\newcommand{\Hilb}{\mathsf{Hilb}}
\newcommand{\Mod}{\mathsf{Mod}}

\newcommand{\Rep}{\mathsf{Rep}}
\newcommand{\UFC}{\mathsf{UFC}}
\newcommand{\UmFC}{\mathsf{UmFC}}

\DeclareMathOperator{\loc}{loc}

\renewcommand{\MR}[1]{}

\newcommand{\set}[2]{\left\{#1 \middle| #2\right\}}
\newcommand{\ket}[1]{\left|#1\right\rangle}

\title{Non-invertible symmetry enriched string net topological orders}
\author{Luisa Eck, Peter Huston, Kyle Kawagoe, and David Penneys}
\date{\today}

\begin{document}

\begin{abstract}
We propose a definition of a non-invertible symmetry enriched topological order (NI-SETO), and we implement our definition for string net models.
We do so in two ways, using full inclusions of unitary fusion categories (UFCs), as well as anyon condensation.
In both cases, the NI-SETO is a relative center of UFCs.
All NI-SETOs can be realized in either model, where we can use enriched UFCs to get chiral examples on the boundary of a 3D Walker-Wang model representing the anomaly.
We describe several examples of NI-SETOs and compute the qualitative symmetry action on anyons and symmetry defects using tube algebra techniques.
\end{abstract}

\maketitle
\tableofcontents

%%%%%%%%%%%%%%%%%%%%%%%%%%%%%%%%%%%%%%%%%%%%%%%%%%%%
%%%%%%%%%%%%%%%%%%%%%%%%%%%%%%%%%%%%%%%%%%%%%%%%%%%%
%%%%%%%%%%%%%%%%%%%%%%%%%%%%%%%%%%%%%%%%%%%%%%%%%%%%
\section{Introduction}

Symmetry is an important tool in the study of phases of matter \cite{10.1016/B978-0-08-010586-4.50034-1}.
Typically, the notion of symmetry is represented by the mathematical notion of a group, consisting of invertible self-maps.
At this time, it is well accepted that various quantum mathematical objects, like topological phases of matter, naturally live in higher categories, which afford richer notions of \emph{quantum symmetry}.
That is, the self-maps, or \emph{endomorphisms}, of an object in a higher category naturally form a higher monoidal category, which can act as the symmetries of the original system.

Topological orders with a group-like global symmetry are called \emph{symmetry enriched topological orders} (SETOs) \cite{1410.4540}.
If the underlying topological order is characterized by the unitary modular tensor category (UMTC) $\cC$, the SETO is represented by a \emph{$G$-crossed braided fusion category} whose trivially graded component is $\cC$.
Gauging the global on-site symmetry yields a new UMTC, and our original SETO can be viewed as a `half-way point' between the underlying UMTC $\cC$ and the gauged theory.

Since topological orders admit non-invertible quantum symmetries, it is natural to ask: {\textbf{What is a non-invertible version of an SETO?}

To propose an answer to this question, we recall that UMTCs/anyon theories $\cD,\cC$ separated by a gapped domain wall are Witt equivalent; indeed, by the folding trick, such a gapped domain wall corresponds to a gapped boundary from $\cD\boxtimes \overline{\cC}$ to the vacuum, and thus determines a Witt equivalence.
Such Witt equivalences, and how they implement phase transitions, were recently worked out via the theory of generalized orbifold data \cite{MR4537311,2101.02482,MR4798129,MR4804348} and generalized condensations \cite{1905.09566,MR4444089}.
Starting with the UMTC/anyon theory $\cC$, one can implement a Witt equivalence by condensing a nice algebra object in the fusion 2-category $\Mod(\cC)$.
In this article, we posit that the wall excitations on such a Witt-equivalence may be viewed as a \emph{non-invertible symmetry enriched topological order} (NI-SETO).

\begin{ansatz}
\label{ansatz:niseto}
A NI-SETO with trivially graded component $\cC$ corresponds to a rigid, separable, unital $E_1$-algebra $\cS\in\Mod(\cC)$ whose unit map $\cC\to \cS$ is a \emph{fully faithful 1-morphism}.\footnote{The notion of a fully faithful 1-morphism in a 2-category was introduced in \cite[\S1.2.2]{1812.11933}: 
the unit $U: \cC\to \cS$ is called fully faithful if for every $\cM\in \Mod(\cC)$, the postcomposition functor $U_*:\Hom(\cM\to \cC)\to \Hom(\cM\to \cS)$ is fully faithful.
The notion of a rigid $E_1$-algebra $\cS$ was discussed in detail in \cite[\S6.1]{2410.05120}, and basically corresponds to the rigidity of $|\cS|$.
}
\end{ansatz}

Similar to how twist defects give a $G$-crossed braided category for an on-site $G$-symmetry \cite{1410.4540}, we get a UFC of \emph{twist defects}
$$
|\cS| := \Hom(\cC\to \cS)
$$
using the graphical calculus construction of \cite{MR4535015}, which comes equipped with a canonical braided central functor $\cC\to Z(|\cS|)$.\footnote{More should be true: $|\cdot|$ should be an equivalence of 3-categories
$$
\{\text{nice }E_1-\text{algebras in }\Mod(\cC)\}\cong \UmFC^\cC,
$$
the 3-category of $\cC$-\emph{enriched unitary multifusion categories} \cite{MR4640433}; see \cite[Lem.~2.23]{MR4654609} and  \cite[Thm.~5.1.2]{MR4600461}.
}
Since $\cC$ is modular, we have a factorization $Z(|\cS|)\cong \cC\boxtimes \overline{\cD}$ for some UMTC $\cD$ \cite{MR1990929}, and thus $|\cS|$ implements a Witt equivalence to a $\cD$ anyon theory (see also \cite[Thm.~2.18]{MR4498161}).

\[
\tikzmath{
\fill[blue!20] (-1,-1) rectangle (0,1);
\fill[red!20] (1,-1) rectangle (0,1);
\draw[thick] (0,-1) node[below]{\scriptsize Wall excitations $|\cS|$}-- (0,1);
\node[blue] at (-.5,0) {$\scriptstyle\cC$};
\node[red] at (.5,0) {$\scriptstyle\cD$};
}
\qquad\qquad
\tikzmath{
\draw[snake,thick,->] (0,0) --node[above]{\scriptsize Folding trick} (2,0);
}
\qquad\qquad
\tikzmath{
\fill[blue!20] (-1,-1) rectangle (1,1);
\filldraw[fill=red!20, thick] (-.2,1) -- (-.2,0) arc(-180:0:.2cm) -- (.2,1);
%\draw[thick] (0,-1) node[below]{$\scriptstyle |\cS|$}-- (0,1);
\node[blue] at (-.5,0) {$\scriptstyle\cC$};
\node[red] at (0,.2) {$\scriptstyle\cD$};
\node at (0,-1.25) {\scriptsize Twist defects $|\cS|$};
}
\]

Morally speaking, $|\cS|$ is a generalization of the notion of center of a bimodule category explored in \cite{MR1151906,MR2587410}, so $|\cS|$ behaves like a \emph{relative center}. 
We make this precise in our two implementations of Ansatz~\ref{ansatz:niseto} below.
The condition that the unit of $\cS$ is fully faithful is exactly the condition that the composite
\begin{equation}
\label{eq:FullyFaithfulComposite}
\cC\to Z(|\cS|) \to |\cS|
\end{equation}
is fully faithful,
so that the anyon theory is a full subcategory of $|\cS|$.\footnote{\label{footnote:localizeFF} If the unit $U: \cC\to \cS$ is fully faithful, then 
$$
\cC=\Hom(\cC\to \cC) \xrightarrow{U_*} \Hom(\cC\to \cS)=|\cS|
$$
is fully faithful, and this functor is exactly the composite \eqref{eq:FullyFaithfulComposite}.
Conversely, $U_*$ is fully faithful at $\cC$ if and only if the evaluation $\epsilon_\cC :U^L_*U_*\to \id_{\End(\cC)}$ is an isomorphism.
But since $\epsilon: (U^LU)_*\to 1_{\Hom(-\to\cC)}=(\epsilon_{\cC})_*$, by the Yoneda Lemma, $\epsilon: U^LU\to 1_\cC$ is an isomorphism in $\Mod(\cC)$, whence $U$ is fully faithful.
} 
This ensures that the `trivial' NI-symmetry acts as the identity on $\cC$ itself.

We implement Ansatz~\ref{ansatz:niseto} in two ways:
\begin{itemize}
\item 
Given a full inclusion $\cX\subset \cY$ of unitary fusion categories (UFCs), we may view $\cY\in \Bim(\cX)$ as such an $E_1$-algebra.
Under the equivalence $\Bim(\cX)\cong\Mod(Z(\cX))$, $|\cY|=Z_\cX(\cY)$, the \emph{relative center}.\footnote{
Recall that $Z(Z_\cX(\cY))\cong Z(\cX)\boxtimes \overline{Z(\cY)}$, and thus $Z_\cX(\cY)$ witnesses the desired Witt equivalence between $Z(\cX)$ and $Z(\cY)$.
}
When $\cY$ is $G$-graded, this relative center is a $G$-crossed braided extension of $Z(\cX)$ \cite{MR2587410} (see also \cite{MR4281262}), recovering the usual SETO story \cite{1410.4540}.

In \S\ref{sec:Vanilla} below, we give a commuting projector lattice model to realize this NI-SETO, obtained from the usual string net model for $\cX$ by adding the simple objects in $\Irr(\cY)\setminus\Irr(\cX)$ as new edge labels, along with describing how one fuses twist defects. 
We compute how the NI-symmetry is seen by sweeping a domain wall in \S\ref{subsec:domainwallwrap}, and we calculate the NI-symmetry action on anyons using tube algebra techniques generalizing \cite{1711.07982} for the \emph{relative tube algebra} $\Tube_\cX(\cY)$ of our full inclusion.
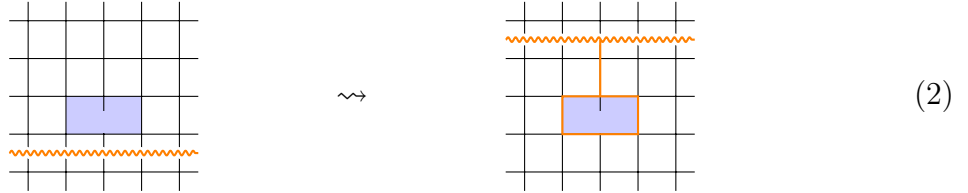
\begin{equation}
\label{eq:SweepDomainWall}
\tikzmath{
\fill[blue!20] (.5,.5) rectangle (1.5,1);
\draw[step=.5] (-.25,-.25) grid (2.25,2.25);
\fill[blue!20] (.98,.51) rectangle (1.02,.8);
\draw[thick, orange, knot, snake] (-.25,.25) -- (2.25,.25);
}
\qquad\qquad
\rightsquigarrow
\qquad\qquad
\tikzmath{
\fill[blue!20] (.5,.5) rectangle (1.5,1);
\draw[step=.5] (-.25,-.25) grid (2.25,2.25);
\draw[thick, orange] (.5,.5) rectangle (1.5,1);
\fill[blue!20] (.98,.51) rectangle (1.02,.8);
\draw[thick, orange, knot, snake] (-.25,1.75) -- (2.25,1.75);
\draw[thick, orange] (1,1) rectangle (1,1.75);
}
\end{equation}

This construction also works on the boundary of an invertible 4D TQFT, which we implement by making our 2D string-net the boundary of a 3D Walker-Wang model associated to a UMTC $\cB$ representing the anomaly as in \cite{MR4640433,2305.14068}.
This allows us to realize NI-SETOs in which $\cC$ is chiral.

\item 
Given a UFC $\cY$ and a condensable algebra $A\in Z(\cY)$, we consider the UFC $\cY_A$ of $A$-modules in $\cY$.\footnote{Endowing $\cY_A$ with a monoidal structure requires us to fix an equivalence ${}_A\cY\cong \cY_A$, i.e., a way to promote every right $A$-module to an $A-A$ bimodule by using the half-braiding with $\cY$.
There are two canonical choices; see \cite[\S3.3]{MR3039775}.}
We may view ${}_A\cY_A$, the $A-A$ bimodules in $\cY$, as such an $E_1$-algebra in $\Bim(\cY_A)$.
In this case, we have a full inclusion $\cY_A\hookrightarrow {}_A\cY_A$, and we may identify
$$
|{}_A\cY_A| 
\cong
\Hom_{\cY_A-\cY_A}(\cY_A\to {}_A\cY \boxtimes_\cY \cY_A)
\cong
\End_{\cY_A-\cY}(\cY_A)
\underset{\text{(\cite[Ex.~III.2]{MR4640433})}}{\cong}
Z(\cY)_A,
$$
the $A$-modules in $Z(\cY)$, which includes $Z(\cY_A)=Z(\cY)_A^{\loc}$ as a full subcategory \cite[Thm.~3.20]{MR3039775}.
In other words, we have $Z_{\cY_A}({}_A\cY_A)\cong Z(\cY)_A$.

From the tensor category point of view, the inclusion $\cY_A\subseteq{}_A\cY_A$ is just an example of a full inclusion.
However, we can also associate such an inclusion with a different string net model construction which captures anyon condensation within the string net model associated to $\cY$ \cite{MR4642306}, involving an $A$-\emph{tube algebra} to measure local modules.
In \S\ref{sec:Condensation} below, we describe how this model realizes a NI-SETO, including a description of how one sees the NI-symmetry action via sweeping domain walls.\footnote{The fact that $\cY_A\subseteq{}_A\cY_A$ is a full inclusion which implements the condensation of $A\in Z(\cY)\cong Z({}_A\cY_A)$ was used in \cite{10.1007/JHEP05_2025_156} in order to study anyon condensation via a lattice model similar to, but simpler than, the one introduced in \cite{MR4642306}. 
It would be interesting to translate the details of the non-invertible symmetry actions in \S\ref{sec:Condensation} into the fluxon condensation framework of that work.}
This construction generalizes existing work in which similar anyon condensation models were used to study invertible SETO \cite{PhysRevB.94.235136}.
Again, our construction also works for a $\cB$-enriched UFC $\cY$ on a 2D boundary of a 3D Walker-Wang model, provided the condensable algebra $A$ lies in $Z^\cB(\cY)$, the M\"uger centralizer/enriched center \cite{MR3725882}.
\end{itemize}

Above, we see two situations that are morally dual to each other:
\begin{itemize}
\item 
A fully faithful tensor functor $\cX\hookrightarrow \cY$ (full inclusion), and
\item 
A dominant tensor functor $\cY \twoheadrightarrow \cY_A$ (condensation).
\end{itemize}
In \S\ref{sec:Implementing} below, we show that every $E_1$ algebra satisfying the requirements of Ansatz~\ref{ansatz:niseto} is equivalent to algebras of both of these forms, so that all NI-SETOs enjoy both kinds of lattice model implementation.
In particular, these constructions show us how to realize every $\cS$ satisfying Ansatz~\ref{ansatz:niseto} as a split condensation algebra.

%%%%%%%%%%%%%%%%%%%%%%%%%%%%%%%%%%%%%%%%%%%%%%
\subsection{Novel features of NI-SETOs}

One distinguishing feature of NI-SETOs is that the symmetry action can transform a single anyon into a direct sum of excitations, including anyons and twist defects. 
A qualitatively similar phenomenon occurs in 1+1D CFTs, where non-invertible duality lines can transform a local field into a nonlocal one
\cite{PhysRevB.3.3918}, see also \cite[Fig.~4]{Shao:2023gho}.
This raises the following immediate concerns (a.k.a.~features!):
\begin{enumerate}
    \item The energy cost of a domain wall scales with length, whereas the energy cost of an anyon is subextensive, so the symmetry cannot commute with the Hamiltonian.
    \item 
    Anyon states have no signature of the location of the string operators used to create them, so a NI-symmetry action must somehow pick where to place the domain wall.
    \item 
    The trivial anyon $1$ can be mapped to a direct sum which includes nontrivial anyons, e.g.~$1 \mapsto 1+e$.
\end{enumerate}

The first concern may be alleviated by realizing that there is not just one NI-symmetry action and Hamiltonian, but an entire \emph{family} of NI-symmetry actions and Hamiltonians parameterized by the potential locations of domain walls.
Although the dependence of the symmetry action on domain wall locations is new, this family of Hamiltonians was already considered in the case of group symmetry.
Indeed, in the standard group SETO setting \cite{1410.4540}, $g$-symmetry defect states are obtained by twisting the Hamiltonian by $g$ along a finite path and considering excitations localized at one end of the resulting domain wall.
Applying the $h$-symmetry to this state results in an $hgh^{-1}$ defect which we model with an $hgh^{-1}$ twisted Hamiltonian instead.
In other words, it is already standard practice to describe these excitations with a family of Hamiltonians which are not strictly invariant under the symmetry.
Our proposal for these non-invertible symmetries simply extends this principle to actions that can transform anyons in $Z(\cX)$ to defects in $Z_\cX(\cY)$. 

Since the symmetry actions now depend on a choice of potential domain wall locations, we have also resolved concern (2).

Resolving concern (3) is a bit trickier, since it relies on the definition of a localized excitation, which violates only a few terms of our local Hamiltonian.
If we view the anyon 1 as the ground state, which does not violate any term in the Hamiltonian, then the ground state is symmetric under a NI-symmetry action.
This can be seen from the fact that the domain wall implementing the NI-symmetry action passes freely along the lattice.
However, if we allow for a location which potentially violates the Hamiltonian as in \eqref{eq:SweepDomainWall} above, we see that as tube algebra representations, wrapping a domain wall can take the trivial representation corresponding to the anyon 1 to nontrivial anyons.
In particular, the tensor product of $\Tube_\cX(\cY)$ representations is not the ordinary Hilbert space tensor product cf.~\cite{MR4808260}, and this tensor product generally does not admit a braiding.
Thus, the symmetry action on a many-body state must depend \emph{a priori} on some additional choice, such as domain wall locations.

In other words, while in the group symmetry case, the untwisted Hamiltonian was invariant under the symmetry action, for NI-symmetries, the untwisted Hamiltonian is not invariant under the symmetry action.
The NI-symmetry action will generally not commute with the untwisted Hamiltonian at the endpoints of potential domain walls, resulting in nontrivial excitations appearing from the vacuum.

%%%%%%%%%%%%%%%%%%%%%%%%%%%%%%%%%%%%%%%%%%%%%%%%%
\subsection{Implementing the ansatz by full inclusions and condensations}
\label{sec:Implementing}

In this section, we explain why Ansatz \ref{ansatz:niseto} can always be implemented both by fully faithful (enriched) inclusions and by (enriched) condensations.

\subsubsection{Realization by full inclusions}

Every $\cS\in \Mod(\cC)$ can be realized by a $\cC$-enriched UFC $|\cS|:= \Hom_{\cC-\cC}(\cC\to \cS)$ such that the composite
$$
\cC\to Z^\cC(|\cS|) \to |\cS|
$$
is fully faithful.
Thus, every NI-SETO $\cS$ comes from a full inclusion, namely $\cC\subseteq|\cS|$ of $\overline{\cC}$-enriched fusion categories.

More is true: for every UMTC $\cB$, $\cB$-enriched fusion category $\cX$, and NI-SETO $\cS\in\Mod(Z^{\cB}(\cX))$, there is a full inclusion $\cX\subseteq\cY$ with $Z^{\cB}_{\cX}(\cY)=|\cS|$.
Indeed, we define $\cY$ to be the image of $\cS$ under the equivalence
$\Mod(Z^{\cB}(\cX))\cong\Bim^{\cB}(\cX)$ from \cite[Prop.~II.13]{MR4640433} so that
\begin{equation}
\label{eq:RealizationAsHoms}
|\cS|=\Hom^\cB_{\cX-\cX}(\cX\to \cY)=Z_\cX^\cB(\cY).
\end{equation}

\subsubsection{From full inclusions to condensations}

Given a fully faithful inclusion $\cX\subset \cY$ of $\cB$-enriched UFCs, by dualizability in $\UmFC^\cB$ applied to \eqref{eq:RealizationAsHoms},
$$
|\cS|\cong\End^\cB_{\cX-\cY}(\cY).
$$
By \cite[Const.~II.8]{MR4640433}, 
\[
Z(\End^{\cB}_{\cX-\cY}(\cY))
\cong 
Z^{\cB}(\cX)\boxtimes Z^{\cB}(\cY)^{\rev}.
\]
Now consider the canonical Lagrangian algebra $L$ in $Z^{\cB}(\cX)\boxtimes Z^{\cB}(\cY)^{\rev}$ corresponding to the forgetful functor $Z(\End^{\cB}_{\cX-\cY}(\cY))\to \End^{\cB}_{\cX-\cY}(\cY)$.
By \cite[\S3]{MR3022755}, $L$ corresponds to a triple $(A,B,\Phi)$, where $A\boxtimes 1\in Z^{\cB}(\cX)$ and $1\boxtimes B\in Z^{\cB}(\cY)$ are the largest condensable subalgebras of $L$ in $Z^{\cB}(\cX)\boxtimes\Hilb$ and $\Hilb\boxtimes Z^{\cB}(\cY)$ respectively, and $\Phi:Z^{\cB}(\cX)_A^{\loc}\to Z^{\cB}(\cY)_B^{\loc}$ is a braided equivalence.
In fact, $A=U^LU$, so by Footnote~\ref{footnote:localizeFF}, $A=1$, and we simply have a braided monoidal equivalence $\Phi:Z^{\cB}(\cX)\to Z^{\cB}(\cX)_B^{\loc}$.
Consequently, the anyon condensation $\cY_B\subseteq\cY$ produces the same $\cS\in\Mod(Z^{\cB}(\cY_B))\cong\Mod(Z^{\cB}(\cX))$.

Indeed, the equivalence $\Mod(Z^{\cB}(\cY_B))\cong\Bim^{\cB}(\cY_B)$ sends $\cS$ to ${}_B\cY_B$, the multifusion category of $B-B$ bimodules in $\cY$.
Since $\cY_B$ is a full subcategory of ${}_B\cY_B$,\footnote{
Here, there are two ways to include $\cY_B$ into ${}_B\cY_B$ corresponding to $\alpha$-induction \cite{MR1815993}. Indeed, the tensor product on $\cY_B$ is defined by restricting the monoidal product $\otimes_B$ of ${}_B\cY_B$ along one of these inclusions, which then becomes the monoidal inclusion of $\cY_B$ as a full subcategory.
}
this shows that anyon condensations are actually a special case of full inclusions.

\begin{rem}
One consequence of this construction
is that the forgetful functor $Z^{\cB}(\cY)\to Z^{\cB}_{\cX}(\cY)$ is always dominant.
Indeed,
\[Z^{\cB}_{\cX}(\cY)\cong\End^{\cB}_{\cX-\cY}(\cY)\cong(Z^{\cB}(\cX)\boxtimes Z^{\cB}(\cY))_L\cong Z^{\cB}(\cY)_B\]
where the last equality relies on the fact that $U^LU\cong 1$.    
Thus every twist defect comes from an anyon in the gauged theory.
\end{rem}

%%%%%%%%%%%%%%%%%%%%%%%%%%%%%%%%%%%%%%%%%%%%%%%%%%%%
\subsection{Conclusions and outlook}

We have given Ansatz~\ref{ansatz:niseto} as a candidate definition of non-invertible symmetry enrichment of a (2+1)D topological order.
However, our approach leaves open important questions about the formalization of the notion of non-invertible symmetry enrichment.
In particular, our definition is incomplete, as it does not include a structure which abstractly captures the collection of symmetries which are acting ({e.g.}~a group, hypergroup \cite{Bischoff:2016jmy}, hopf algebra, hopf monad \cite{MR3905557}, etc.), independent of their realization as symmetries of a particular UMTC.
It is unsatisfying to say that the algebra object $\cS\in\Mod(\cB)$ is acting, because this does not let us compare actions of the same abstract collection of non-invertible symmetries on different UMTCs.
In contrast, for a group $G$, there are already definitions of $G$-action which make sense for any UMTC.
One is that of a $G$-crossed braided extension; another is that of a functor from $\Rep(G)$.
Such definitions allow for the formulation of natural notions like $G$-symmetric domain wall between different $G$-symmetry enriched topological orders, something Ansatz~\ref{ansatz:niseto} does not appear to support.
We aim to show that Ansatz~\ref{ansatz:niseto} is the shadow of natural notions of non-invertible symmetry action, such as that of a monoidal functor $\mathfrak{G}\to \Mod(\cB)$, in future work.

%%%%%%%%%%%%%%%%%%%%%%%%%%%%%%%%%%%%%%%%%%%%%%%%%%%%
\subsection{Acknowledgments}

The authors would like to thank
Maissam Barkeshli,
Clement Delcamp,
Rajath Radhakrishnan,
Corey Jones,
Yuan-Ming Lu,
Milo Moses, 
Shu-Heng Shao,
Dom Williamson,
and Xinping Yang
for helpful conversations.
During the preparation of this manuscript, the authors learned of an independent treatment of some of the concepts in this article in joint work of Clement Delcamp and Dom Williamson, which should appear close to this work.

This research project began at the 2025 program on Generalized Symmetries: High-Energy, Condensed Matter and Mathematics at the Kavli Institute for Theoretical Physics (KITP), supported by NSF PHY-2309135.
The authors would like to thank KITP and the co-organizers of the program for their hospitality.
DP was supported by NSF DMS 2154389.
PH was supported by EPSRC Programme Grant {No.}~EP/W007509/1. 
LE was supported by
the Walter Burke Institute for Theoretical Physics at Caltech.
% : ``Combinatorial Representation Theory: Discovering the Interfaces of Algebra with Geometry and Topology''.

%%%%%%%%%%%%%%%%%%%%%%%%%%%%%%%%%%%%%%%%%%%%%%%%%%%%
%%%%%%%%%%%%%%%%%%%%%%%%%%%%%%%%%%%%%%%%%%%%%%%%%%%%
%%%%%%%%%%%%%%%%%%%%%%%%%%%%%%%%%%%%%%%%%%%%%%%%%%%%
\section{String nets and defects for non-invertible SETO}
\label{sec:Vanilla}

We begin by reviewing topological defects for string nets from \cite{MR2942952}.
We use the skein theoretic unitary tensor category approach from \cite{MR3204497,2305.14068} as opposed to the 6j-symbol approach from \cite{PhysRevB.71.045110,MR2726654,PhysRevB.103.195155}.
Given a UFC $\cX$ and an $\cX-\cX$ bimodule $\cM$, we introduce a topological defect, and we describe how the \emph{relative tube algebra} $\Tube_\cX(\cM)$ can be used to identify the twist defects at the end of the topological defect, as $\Mod(\Tube_\cX(\cM))\cong \Hom_{\cX-\cX}(\cX\to \cM)$.
By the folding trick (see \cite[p19]{2305.14068}), the relative tube algebra is Morita equivalent to the weak Hopf $\rmC^*$ strip algebra used in \cite{MR2942952} to study boundary excitations when $\cM = \cN\boxtimes \cN^{\rm op}$ where $\cN$ is an $\cX$-module category corresponding to a gapped boundary from $Z(\cX)$ to $\Hilb$.
As a sanity check, we note that dualizability in $\UFC$ implies that
$$
\End({}_\cX\cN) \cong 
\Hom_{\cX-\cX}(\cX\to \cN\boxtimes \cN^{\rm op}) 
=
\Hom_{\cX-\cX}(\cX\to \cM). 
$$

The novel part of this section is our analysis of the case in which $\cM=\cY$ is a fusion category containing $\cX$ as a full subcategory.
We discuss how to fuse twist defects using the graphical calculus of strings on tubes from \cite{MR3578212} and the Day convolution on $\Mod(\Tube_\cX(\cY))$ generalizing the discussion from \cite[\S{III.B}]{MR4808260}.
In this case,
$$
\Mod(\Tube_\cX(\cY))\cong \Hom_{\cX-\cX}(\cX\to \cY)\cong Z_\cX(\cY),
$$
the relative center.
In this sense, we say that $Z_\cX(\cY)$ is a 
\emph{non-invertible symmetry enriched} topological order.
Indeed, when $\cY$ is a $G$-graded extension of $\cX$, $Z_\cX(\cY)$ is an honest $G$-crossed braided extension of $Z(\cX)$ by \cite{MR2587410}.
In the event that $\cY = \cN\boxtimes \cN^{\rm op}$ as above, we recover the same fusion of topological wall defects from \cite{MR2942952} by the folding trick. 

%%%%%%%%%%%%%%%%%%%%%%%%%%%%%%%%%%%%%%%%%%%%%%%%%%%%
\subsection{The unitary tensor category string net model}
\label{ssec:LWReview}

We rapidly recall the unitary tensor category string net model for a UFC $\cX$ from \cite{MR3204497,2305.14068}.
The local Hilbert space assigned to a vertex in a $\bbZ^2$ lattice is
$$
\cH_v:=\cX(X\otimes X\to X\otimes X)
=
\left\{\tikzmath{
\draw (-.7,0) -- (-.3,0);
\draw (.7,0) -- (.3,0);
\draw (0,-.7) -- (0,-.3);
\draw (0,.7) -- (0,.3);
\draw[thin, cyan] (-.7,.7) -- (.7,-.7);
\draw[->, cyan, thin] (-.6,.4) -- (-.4,.6);
\draw[<-, cyan, thin] (.6,-.4) -- (.4,-.6);
\roundNbox{fill=white}{(0,0)}{.3}{0}{0}{$f$}
}
\right\}
\qquad\qquad
X:=\bigoplus_{x\in\Irr(\cX)} x.
$$
The space $\cH_v$ carries the \emph{skein module inner product} given by
$$
\langle f |g\rangle := \sum_{w,x,y,z\in\Irr(\cX)} \frac{1}{\sqrt{d_wd_xd_yd_z}}\tr_\cX((p_w\otimes p_x)f^\dag(p_y\otimes p_z)g).
$$
The projections are inserted with the quantum dimension scalars in order to make gluing isometric.
For each edge $\ell$ between neighboring sites $u,v$ in $\bbZ^2$, we have an edge projector $A_\ell$ which projects onto the subspace of composite morphisms which are composed along the edge $\ell$.
As gluing is an isometry, $A_\ell$ is an orthogonal projection.
For each square plaquette and $x\in \Irr(\cX)$, we have an operator $B_p^x$ given by
$$
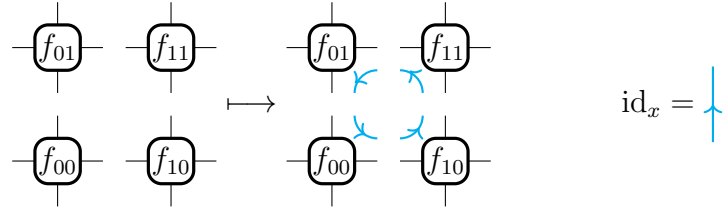
\tikzmath{
\foreach \x in {0,1}{
\foreach \y in {0,1}{
\draw ($ (-.6,0) + 1.5*(\x,\y) $) -- ($ (-.3,0) + 1.5*(\x,\y) $);
\draw ($ (.6,0) + 1.5*(\x,\y) $) -- ($ (.3,0) + 1.5*(\x,\y) $);
\draw ($ (0,-.6) + 1.5*(\x,\y) $) -- ($ (0,-.3) + 1.5*(\x,\y) $);
\draw ($ (0,.6) + 1.5*(\x,\y) $) -- ($ (0,.3) + 1.5*(\x,\y) $);
\roundNbox{fill=white}{($ 1.5*(\x,\y)$)}{.3}{0}{0}{$f_{\x\y}$}
}}
}
\longmapsto
\tikzmath{
\foreach \x in {0,1}{
\foreach \y in {0,1}{
\draw ($ (-.6,0) + 1.5*(\x,\y) $) -- ($ (-.3,0) + 1.5*(\x,\y) $);
\draw ($ (.6,0) + 1.5*(\x,\y) $) -- ($ (.3,0) + 1.5*(\x,\y) $);
\draw ($ (0,-.6) + 1.5*(\x,\y) $) -- ($ (0,-.3) + 1.5*(\x,\y) $);
\draw ($ (0,.6) + 1.5*(\x,\y) $) -- ($ (0,.3) + 1.5*(\x,\y) $);
\roundNbox{fill=white}{($ 1.5*(\x,\y)$)}{.3}{0}{0}{$f_{\x\y}$}
}}
\draw[thick, cyan, mid<] (.3,.9) arc (180:90:.3cm);
\draw[thick, cyan, far>] (.3,.6) arc (-180:-90:.3cm);
\draw[thick, cyan, mid<] (.9,1.2) arc (90:0:.3cm);
\draw[thick, cyan, far>] (.9,.3) arc (-90:0:.3cm);
}
\qquad\qquad
\id_x =
\tikzmath{
\draw[thick,cyan,mid>] (0,-.5) -- (0,.5);
}
$$
The output of this operation must be resolved back into the local Hilbert spaces using the semisimplicity relation in $\cX$ by choosing ONBs of trivalent vertices for the spaces
$$
\cX(X\otimes x\to X)
\qquad\qquad\text{and}\qquad\qquad
\cX(x\otimes X\to X).
$$
The plaquette operator $B_p$ is then given by $\frac{1}{D_\cX}\sum_{x\in\Irr(\cX)} d_x B_p^x$, where $D_\cX=\sum_{x\in\Irr(\cX)} d_x^2$ is the \emph{global dimension} of $\cX$.

The local Hamiltonian of the system is given by
$$
H:= -\sum A_\ell - \sum B_p.
$$
The anyonic localized topological excitations of this model are described by simple objects in the Drinfeld center $Z(\cX)$.
Given an anyonic localized excitation, one can measure/determine its type by looking at the action of the tube algebra $\Tube(\cX)$.
That is, if our excited state $|\psi\rangle$ is localized in two adjacent plaquettes corresponding to violations of two $B_p$ and one $A_\ell$ term of the  Hamiltonian, 
$$
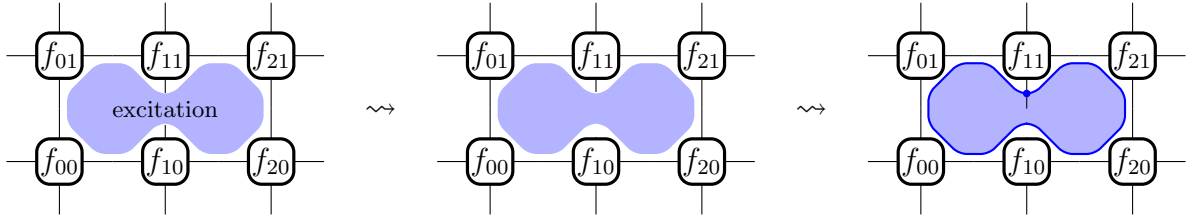
\tikzmath{
\foreach \x in {0,1,2}{
\foreach \y in {0,1}{
\draw ($ (-.7,0) + 1.4*(\x,\y) $) -- ($ (-.3,0) + 1.4*(\x,\y) $);
\draw ($ (.7,0) + 1.4*(\x,\y) $) -- ($ (.3,0) + 1.4*(\x,\y) $);
\draw ($ (0,-.7) + 1.4*(\x,\y) $) -- ($ (0,-.3) + 1.4*(\x,\y) $);
\draw ($ (0,.7) + 1.4*(\x,\y) $) -- ($ (0,.3) + 1.4*(\x,\y) $);
\roundNbox{fill=white}{($ 1.4*(\x,\y)$)}{.3}{0}{0}{$f_{\x\y}$}
}}
\fill[white] ($ 1.4*(.9,.35) $) rectangle ($ 1.4*(1.1,.65) $);
\fill[blue!30, rounded corners=3pt] (.7,.1) -- (.5,.1) -- (.1,.5) -- (.1,.9) -- (.5,1.3) -- (.9,1.3) -- (1.3,.9) -- (1.5,.9) -- (1.9,1.3) -- (2.3,1.3) -- (2.7,.9) -- (2.7,.5) -- (2.3,.1) -- (1.9,.1) -- (1.5,.5) -- (1.3,.5) -- (.9,.1) -- (.7,.1);
\node at (1.4,.7) {\scriptsize{excitation}};
}
\quad\rightsquigarrow\quad
\tikzmath{
\foreach \x in {0,1,2}{
\foreach \y in {0,1}{
\draw ($ (-.7,0) + 1.4*(\x,\y) $) -- ($ (-.3,0) + 1.4*(\x,\y) $);
\draw ($ (.7,0) + 1.4*(\x,\y) $) -- ($ (.3,0) + 1.4*(\x,\y) $);
\draw ($ (0,-.7) + 1.4*(\x,\y) $) -- ($ (0,-.3) + 1.4*(\x,\y) $);
\draw ($ (0,.7) + 1.4*(\x,\y) $) -- ($ (0,.3) + 1.4*(\x,\y) $);
\roundNbox{fill=white}{($ 1.4*(\x,\y)$)}{.3}{0}{0}{$f_{\x\y}$}
}}
\fill[white] ($ 1.4*(.9,.23) $) rectangle ($ 1.4*(1.1,.65) $);
\fill[blue!30, rounded corners=3pt] (.7,.1) -- (.5,.1) -- (.1,.5) -- (.1,.9) -- (.5,1.3) -- (.9,1.3) -- (1.3,.9) -- (1.5,.9) -- (1.9,1.3) -- (2.3,1.3) -- (2.7,.9) -- (2.7,.5) -- (2.3,.1) -- (1.9,.1) -- (1.5,.5) -- (1.3,.5) -- (.9,.1) -- (.7,.1);
%\node at (1.4,.7) {\scriptsize{excitation}};
}
\quad\rightsquigarrow\quad
\tikzmath{
\foreach \x in {0,1,2}{
\foreach \y in {0,1}{
\draw ($ (-.7,0) + 1.4*(\x,\y) $) -- ($ (-.3,0) + 1.4*(\x,\y) $);
\draw ($ (.7,0) + 1.4*(\x,\y) $) -- ($ (.3,0) + 1.4*(\x,\y) $);
\draw ($ (0,-.7) + 1.4*(\x,\y) $) -- ($ (0,-.3) + 1.4*(\x,\y) $);
\draw ($ (0,.7) + 1.4*(\x,\y) $) -- ($ (0,.3) + 1.4*(\x,\y) $);
\roundNbox{fill=white}{($ 1.4*(\x,\y)$)}{.3}{0}{0}{$f_{\x\y}$}
}}
\fill[white] ($ 1.4*(.9,.23) $) rectangle ($ 1.4*(1.1,.65) $);
\filldraw[thick,blue,fill=blue!30, rounded corners=3pt] (.7,.1) -- (.5,.1) -- (.1,.5) -- (.1,.9) -- (.5,1.3) -- (.9,1.3) -- (1.3,.9) -- (1.5,.9) -- (1.9,1.3) -- (2.3,1.3) -- (2.7,.9) -- (2.7,.5) -- (2.3,.1) -- (1.9,.1) -- (1.5,.5) -- (1.3,.5) -- (.9,.1) -- (.7,.1);
\draw (1.4,.9) -- (1.4,.7);
\fill[blue] (1.4,.9) circle (.05cm);
}
$$
we may apply a projection $p_1$ to $|\psi\rangle$ at exactly one of the local Hilbert spaces neighboring $\ell$.\footnote{One can also skip applying $p_1$ and instead act with a weak $\rmC^*$ Hopf algebra Morita equivalent to $\Tube(\cX)$ as in \cite[\S2.4]{2305.14068}.}
This produces a state which lies in a representation (right module, as the action is by precomposition) of the tube algebra, and the anyon type corresponds to the isotypic component of the representation.
We refer the reader to \cite{2305.14068} for more details.

There is an elegant description of $\Tube(\cX)$ as $\End_{Z(\cX)}(\Tr(X))$ by the following adjunction isomorphisms:
\begin{equation}
\label{eq:tubeViaTrace}
\begin{aligned}
\bigoplus_{x,y,c\in\Irr(\cX)}
\cX(c\otimes x\to y\otimes c)
&\cong
\bigoplus_{c\in\Irr(\cX)}
\cX(c\otimes X\otimes \overline{c}\to X)
\cong
\cX(F\Tr(X)\to X)
\\&\cong
Z(\cX)(\Tr(X)\to \Tr(X))
\end{aligned}
\end{equation}
Here, $\Tr:\cX\to Z(\cX)$ is the unitary adjoint \cite{MR4750417} of the forgetful functor $F: Z(\cX)\to \cX$, denoted $F\dashv^\dag \Tr$.
The functor $\Tr$ is often denoted $I$ in the literature \cite{MR1966525}, but it has the structure of a \emph{categorified trace} \cite{MR3578212}, which affords a powerful graphical calculus of strings on tubes \cite{MR4528312}, which are allowed to branch and braid.
In particular, since $F\dashv^\dag \Tr$ and $F$ is monoidal, $\Tr$ comes equipped with a lax monoidal structure denoted $\mu_{\bullet,\bullet}$ represented by a pair of pants;
of more use to us is the oplax structure given by $ \mu^\dag_{\bullet,\bullet}$.
$$
\tikzmath{
\draw[thick, blue] (0,-.8) --node[left]{$\scriptstyle X$} (0,-.3);
\draw[thick, blue] (0,.8) --node[left]{$\scriptstyle X$} (0,.3);
\roundNbox{}{(0,0)}{.3}{0}{0}{$f$}
}
\overset{\Tr}{\longmapsto}
\tikzmath{
\draw[thick] (.5,-.6) -- (.5,1);
\draw[thick] (-.5,-.6) -- (-.5,1);
\draw[thick] (0,1) ellipse (.5 and .2);
\draw[thick] (.5,-.6) arc(0:-180:.5 and .2);
\draw[thick,dotted] (.5,-.6) arc(0:180:.5 and .2);
\draw[thick, blue] (0,-.8) --node[left]{$\scriptstyle X$} (0,-.3);
\draw[thick, blue] (0,.8) --node[left]{$\scriptstyle X$} (0,.3);
\roundNbox{}{(0,0)}{.3}{0}{0}{$f$}
}
\qquad\qquad\qquad
\mu_{X,X}^\dag
=
\tikzmath{
\draw[thick] (0,0) arc (0:-180:.3 and .1);
\draw[thick, dotted] (0,0) arc (0:180:.3 and .1);
\draw[thick] (-.5,1.4) arc (0:-180:.3 and .1);
\draw[thick] (-.5,1.4) arc (0:180:.3 and .1);
\draw[thick] (.5,1.4) arc (0:-180:.3 and .1);
\draw[thick] (.5,1.4) arc (0:180:.3 and .1);
\draw[thick] (-.1,1.4)  -- (-.1,1.2) arc (0:-180:.2cm) -- (-.5,1.4);
\draw[thick, blue] (-.2,-.1) to[out=90,in=-90] (.2,1.3);
\draw[thick, blue] (-.4,-.1) to[out=90,in=-90] (-.8,1.3);
\draw[thick] (0,0) to[out=90,in=-90] (.5,1.4);
\draw[thick] (-.6,0) to[out=90,in=-90] (-1.1,1.4);
}
:
\Tr(X\otimes X) \to \Tr(X)\otimes \Tr(X)
$$
Since $\cX$ is spherical, strings labelled by objects of $\cX$ are allowed to wrap around the back of the tube.
An explicit $*$-isomorphism $\Tube(\cX)\to \End_{Z(\cX)}(\Tr(X))$ was depicted graphically in \cite[\S4.2]{MR3663592}:
\begin{equation}
\label{eq:IsoTube(X)->End(Tr(X))}
\tikzmath{
\draw[thick, blue] (-.2,.3) --node[left]{$\scriptstyle X$} (-.2,.7);
\draw[thick, cyan] (.2,.3) --node[right]{$\scriptstyle c$} (.2,.7);
\draw[thick, cyan] (-.2,-.3) --node[left]{$\scriptstyle c$} (-.2,-.7);
\draw[thick, blue] (.2,-.3) --node[right]{$\scriptstyle X$} (.2,-.7);
\roundNbox{}{(0,0)}{.3}{.1}{.1}{$f$}
}
\longmapsto
\sum_{y,z\in \Irr(\cX)}
\sqrt{d_c}
\tikzmath{
\draw (-.6,-.9) node[below]{$\scriptstyle y$} -- (-.6,.9)  node[above]{$\scriptstyle z$};
\draw (.6,-.9) node[below]{$\scriptstyle\overline{y}$} -- (.6,.9)  node[above]{$\scriptstyle \overline{z}$};
\draw[thick, cyan] (-.6,-.6) -- (-.25,-.25);
\draw[thick, cyan] (.6,.6) -- (.25,.25);
\draw[thick, blue] (0,.3) -- (0,.9) node[above]{$\scriptstyle X$};
\draw[thick, blue] (0,-.3) -- (0,-.9) node[below]{$\scriptstyle X$};
\roundNbox{fill=white}{(0,0)}{.3}{0}{0}{$f$}
\node[cyan] at (.3,.5) {$\scriptstyle c$};
\node[cyan] at (-.3,-.5) {$\scriptstyle c$};
\filldraw[fill=yellow] (.6,.6) circle (.05cm);
\filldraw[fill=yellow] (-.6,-.6) circle (.05cm);
}
=
\,\,
\tikzmath{
\draw[thick] (.5,-.6) -- (.5,1);
\draw[thick] (-.5,-.6) -- (-.5,1);
\draw[thick] (0,1) ellipse (.5 and .2);
\draw[thick] (.5,-.6) arc(0:-180:.5 and .2);
\draw[thick,dotted] (.5,-.6) arc(0:180:.5 and .2);
\draw[thick, blue] (0,-.8) --node[left]{$\scriptstyle X$} (0,-.3);
\draw[thick, blue] (0,.8) --node[left]{$\scriptstyle X$} (0,.3);
\draw[thick, cyan] (.5,.2) arc (0:-180:.5 and .2);
\draw[thick, cyan, dotted] (.5,.2) arc (0:180:.5 and .2);
\roundNbox{fill=white}{(0,0)}{.3}{0}{0}{$f$}
}
\,.
\end{equation}
In the above expression on the right, the shaded nodes correspond to summing over an orthonormal basis (ONB) for the skein module and its adjoint, which is, of course, independent of the choice of ONB.
Of particular importance here is the \emph{fusion relation} in $\cX$ afforded by semisimplicity, depicted below, as well as the \emph{I=H relation} (see \cite[Lem.~2.16]{MR3663592}).
$$
\tikzmath{
\draw (-.3,-.5) --node[left]{$\scriptstyle a$} (-.3,.5);
\draw (.3,-.5) --node[right]{$\scriptstyle b$} (.3,.5);
}
=
\sum_{c\in\Irr(\cX)}
\tikzmath{
\draw (-.3,-.5) --node[left]{$\scriptstyle a$} (0,-.2);
\draw (-.3,.5) --node[left]{$\scriptstyle a$} (0,.2);
\draw (.3,-.5) --node[right]{$\scriptstyle b$} (0,-.2);
\draw (.3,.5) --node[right]{$\scriptstyle b$} (0,.2);
\draw (0,-.2) --node[right]{$\scriptstyle c$} (0,.2);
\filldraw[fill=yellow] (0,.2) circle (.05cm);
\filldraw[fill=yellow] (0,-.2) circle (.05cm);
}
$$

This graphical calculus allows us to describe the tensor product of $\Tube(\cX)$-representations via \emph{Day convolution} using the (unitary) Yoneda Lemma as in \cite{MR4808260}.
That is, for  $w,z\in Z(\cX)$, the fusion of the $\Tube(\cX)$-representations 
$\cK_w:=Z(\cX)(\Tr(X)\to w)$ and $\cK_z:=Z(\cX)(\Tr(X)\to z)$ is given by
$$
\tikzmath{
\draw[thick] (0,0) arc (0:-180:.3 and .1);
\draw[thick, dotted] (0,0) arc (0:180:.3 and .1);
\draw[thick] (-.5,1) arc (0:-180:.3 and .1);
\draw[thick, dotted] (-.5,1) arc (0:180:.3 and .1);
\draw[thick] (.5,1) arc (0:-180:.3 and .1);
\draw[thick, dotted] (.5,1) arc (0:180:.3 and .1);
\draw[thick] (-.1,1) arc (0:-180:.2cm);
\draw[thick] (-.5,1) -- (-.5,1.3) arc (0:180:.3cm) -- (-1.1,1);
\draw[thick] (.5,1) -- (.5,1.3) arc (0:180:.3cm) -- (-.1,1);
\node at (-.8,1.3) {$\scriptstyle\cK_w$};
\node at (.2,1.3) {$\scriptstyle\cK_z$};
\draw[thick, DarkGreen, snake] (-.8,1.6) --node[left]{$\scriptstyle w$} (-.8,2);
\draw[thick, DarkGreen, snake] (.2,1.6) --node[left]{$\scriptstyle z$} (.2,2);
\draw[thick] (0,0) to[out=90,in=-90] (.5,1);
\draw[thick] (-.6,0) to[out=90,in=-90] (-1.1,1);
\draw[thick, blue] (-.3,-.1) -- (-.3,.4);
\draw[thick, blue] (-.3,.4) to[out=45,in=-90] (.2,.9);
\draw[thick, blue] (-.3,.4) to[out=135,in=-90] (-.8,.9);
\filldraw[blue] (-.3,.4) circle (.05cm);
}
\quad
:=
\bigoplus_{a,b\in\Irr(\cX)}
Z(\cX)(\Tr(a)\to w)
\otimes
\cX(X\to a\otimes b)
\otimes
Z(\cX)(\Tr(b)\to z).
$$

In the language of the lattice model, one can visualize this Day convolution product as zooming out on the lattice.
Indeed, consider the space of states $\{|\psi\rangle\}$ which host localized anyonic excitations $a,b$ in distinct non-overlapping rectangles.
This means that the states $|\psi\rangle$ violate terms of the Hamiltonian localized to these rectangles, and that the above procedure produces a vector in an $a,b$ isotypic $\Tube(\cX)$-representation for $a,b\in Z(\cX)$ respectively.
Take a larger rectangle which contains both rectangles where $a,b$ are localized.
We then consider the space of states obtained by projecting all the internal edges of this larger rectangle to $1_\cX$ except for a tree connecting up the rectangles for $a,b$ and one edge at the top of the larger rectangle, as in the right hand side of the following diagram.
$$
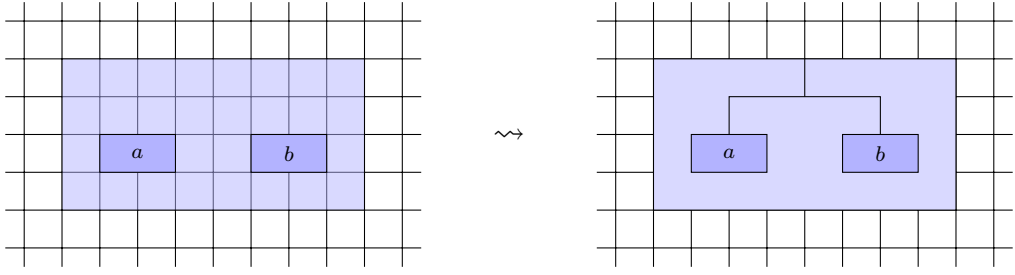
\tikzmath{
\draw[step=.5] (-.25,-.25) grid (5.25,3.25);
\filldraw[fill=blue!30, opacity=.5] (.5,.5) rectangle (4.5,2.5);
\filldraw[fill=blue!30] (1,1) rectangle (2,1.5);
\node at (1.5,1.25) {$\scriptstyle a$};
\filldraw[fill=blue!30] (3,1) rectangle (4,1.5);
\node at (3.5,1.25) {$\scriptstyle b$};
}
\qquad\rightsquigarrow\qquad
\tikzmath{
\draw[step=.5] (-.25,-.25) grid (5.25,3.25);
\filldraw[fill=blue!15] (.5,.5) rectangle (4.5,2.5);
\filldraw[fill=blue!30] (1,1) rectangle (2,1.5);
\node at (1.5,1.25) {$\scriptstyle a$};
\filldraw[fill=blue!30] (3,1) rectangle (4,1.5);
\node at (3.5,1.25) {$\scriptstyle b$};
\draw (1.5,1.5) -- (1.5,2) -- (3.5,2) -- (3.5,1.5);
\draw (2.5,2) -- (2.5,2.5);
}
$$
This space of states is the Day convolution product of the $\Tube(\cX)$-modules corresponding to $a,b$, which is therefore the $\Tube(\cX)$-module corresponding to $a\otimes b$.
Indeed, the $\Tube(\cX)$-action on the right hand side is given by the following diagram, which is exactly the one from \cite[\S{III.B}]{MR4808260}.
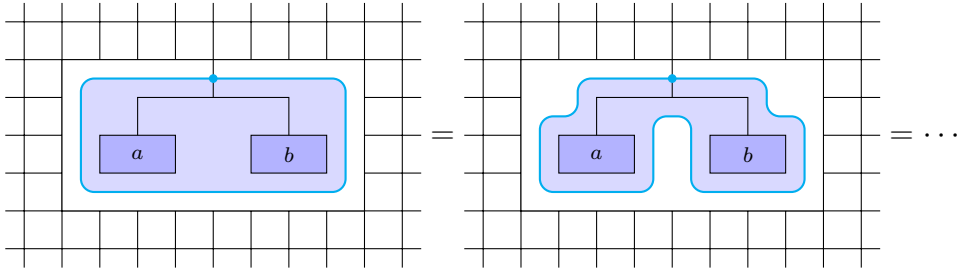
$$
\tikzmath{
\draw[step=.5] (-.25,-.25) grid (5.25,3.25);
\filldraw[fill=white] (.5,.5) rectangle (4.5,2.5);
\filldraw[fill=blue!15, thick, draw=cyan, rounded corners=5pt] (.75,.75) rectangle (4.25,2.25);
\filldraw[fill=blue!30] (1,1) rectangle (2,1.5);
\node at (1.5,1.25) {$\scriptstyle a$};
\filldraw[fill=blue!30] (3,1) rectangle (4,1.5);
\node at (3.5,1.25) {$\scriptstyle b$};
\draw (1.5,1.5) -- (1.5,2) -- (3.5,2) -- (3.5,1.5);
\draw (2.5,2) -- (2.5,2.5);
\filldraw[cyan] (2.5,2.25) circle (.05cm);
}
=
\tikzmath{
\draw[step=.5] (-.25,-.25) grid (5.25,3.25);
\filldraw[fill=white] (.5,.5) rectangle (4.5,2.5);
\filldraw[fill=blue!15, thick, draw=cyan, rounded corners=5pt] (2.5,2.25) -- (1.25,2.25) -- (1.25,1.75) -- (.75,1.75) -- (.75,.75) -- (2.25,.75) -- (2.25,1.75) -- (2.75,1.75) -- (2.75,.75) -- (4.25,.75) -- (4.25,1.75) -- (3.75,1.75) -- (3.75,2.25) -- (2.5,2.25);
\filldraw[fill=blue!30] (1,1) rectangle (2,1.5);
\node at (1.5,1.25) {$\scriptstyle a$};
\filldraw[fill=blue!30] (3,1) rectangle (4,1.5);
\node at (3.5,1.25) {$\scriptstyle b$};
\draw (1.5,1.5) -- (1.5,2) -- (3.5,2) -- (3.5,1.5);
\draw (2.5,2) -- (2.5,2.5);
\filldraw[cyan] (2.5,2.25) circle (.05cm);
}
=
\cdots
$$

%%%%%%%%%%%%%%%%%%%%%%%%%%%%%%%%%%%%%%%%%%%%%%%%%%%%
\subsection{Defects for string nets and twist defects}
\label{sec:GeneralizedTwistDefects}

Following \cite{MR2942952}, given a fully faithful inclusion of UFCs $\cX\subset \cY$, we can form a 1D $\cY$-defect line.
$$
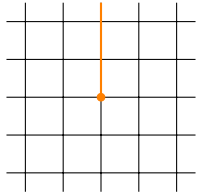
\tikzmath{
\draw[step=.5] (-.25,-.25) grid (2.25,2.25);
\draw[thick, orange] (1,2.25) -- (1,1);
\filldraw[orange] (1,1) circle (.05cm);
}
$$
Along the 1D $\cY$-defect line, the local Hilbert space changes to
$$
\cK_v:=\cY(X\otimes \textcolor{orange}{Y}\to \textcolor{orange}{Y}\otimes X)
=
\left\{\tikzmath{
\draw (-.7,0) -- (-.3,0);
\draw (.7,0) -- (.3,0);
\draw[thick, orange] (0,-.7) -- (0,-.3);
\draw[thick, orange] (0,.7) -- (0,.3);
\draw[thin, cyan] (-.7,.7) -- (.7,-.7);
\draw[->, cyan, thin] (-.6,.4) -- (-.4,.6);
\draw[<-, cyan, thin] (.6,-.4) -- (.4,-.6);
\roundNbox{fill=white}{(0,0)}{.3}{0}{0}{$f$}
}
\right\}
\qquad\qquad
Y:=\bigoplus_{y\in\Irr(\cY)} y.
$$
%and similarly for $\cM$.
At the end of defect, we may insert the local Hilbert space
$$
\cL_v:=
\cY(X\otimes X\to \textcolor{orange}{Y}\otimes X)
=
\left\{\tikzmath{
\draw (-.7,0) -- (-.3,0);
\draw (.7,0) -- (.3,0);
\draw (0,-.7) -- (0,-.3);
\draw[thick, orange] (0,.7) -- (0,.3);
\draw[thin, cyan] (-.7,.7) -- (.7,-.7);
\draw[->, cyan, thin] (-.6,.4) -- (-.4,.6);
\draw[<-, cyan, thin] (.6,-.4) -- (.4,-.6);
\roundNbox{fill=white}{(0,0)}{.3}{0}{0}{$f$}
}
\right\}.
$$

These spaces again carry the skein module inner products, where we must introduce the appropriate projections and normalize by the square roots of their quantum dimensions.
Along the defect line, we have a similar gluing operator $A_\ell$, and our plaquette operators $B_p^x$ are defined similarly for $x\in\Irr(\cX)$ using chosen ONBs of trivalent vertices for the spaces
$$
\cY(x\otimes Y\to Y)
\qquad\qquad
\text{and}
\qquad\qquad
\cY(Y\otimes x\to Y),
$$

By \cite{MR2942952} (see also \cite{MR4640433}), the topological excitations which live at the vertex correspond to the relative center 
$$
\Hom_{\cX-\cX}(\cX\to \cY) \cong Z_\cX(\cY).
$$
Precisely, these vertex excitations are localized on the edge above the vertex and the two neighboring plaquettes, as these correspond to the violated terms of the Hamiltonian.
$$
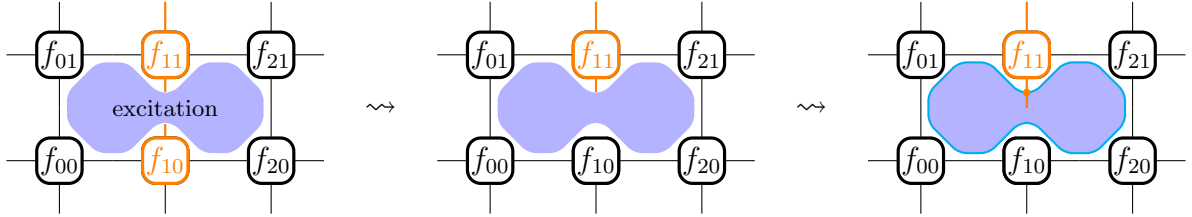
\tikzmath{
\foreach \x in {0,1,2}{
\foreach \y in {0,1}{
\draw ($ (-.7,0) + 1.4*(\x,\y) $) -- ($ (-.3,0) + 1.4*(\x,\y) $);
\draw ($ (.7,0) + 1.4*(\x,\y) $) -- ($ (.3,0) + 1.4*(\x,\y) $);
\draw ($ (0,-.7) + 1.4*(\x,\y) $) -- ($ (0,-.3) + 1.4*(\x,\y) $);
\draw ($ (0,.7) + 1.4*(\x,\y) $) -- ($ (0,.3) + 1.4*(\x,\y) $);
\roundNbox{fill=white}{($ 1.4*(\x,\y)$)}{.3}{0}{0}{$f_{\x\y}$}
}}
\draw[thick, orange] ($ 1.4*(1,0) $) -- ($ (0,.7) + 1.4*(1,1) $);
\fill[white] ($ 1.4*(.9,.35) $) rectangle ($ 1.4*(1.1,.65) $);
\fill[blue!30, rounded corners=3pt] (.7,.1) -- (.5,.1) -- (.1,.5) -- (.1,.9) -- (.5,1.3) -- (.9,1.3) -- (1.3,.9) -- (1.5,.9) -- (1.9,1.3) -- (2.3,1.3) -- (2.7,.9) -- (2.7,.5) -- (2.3,.1) -- (1.9,.1) -- (1.5,.5) -- (1.3,.5) -- (.9,.1) -- (.7,.1);
\node at (1.4,.7) {\scriptsize{excitation}};
\roundNbox{orange, fill=white}{($ 1.4*(1,0)$)}{.3}{0}{0}{\textcolor{orange}{$f_{10}$}}
\roundNbox{orange, fill=white}{($ 1.4*(1,1)$)}{.3}{0}{0}{\textcolor{orange}{$f_{11}$}}
}
\quad\rightsquigarrow\quad
\tikzmath{
\foreach \x in {0,1,2}{
\foreach \y in {0,1}{
\draw ($ (-.7,0) + 1.4*(\x,\y) $) -- ($ (-.3,0) + 1.4*(\x,\y) $);
\draw ($ (.7,0) + 1.4*(\x,\y) $) -- ($ (.3,0) + 1.4*(\x,\y) $);
\draw ($ (0,-.7) + 1.4*(\x,\y) $) -- ($ (0,-.3) + 1.4*(\x,\y) $);
\draw ($ (0,.7) + 1.4*(\x,\y) $) -- ($ (0,.3) + 1.4*(\x,\y) $);
\roundNbox{fill=white}{($ 1.4*(\x,\y)$)}{.3}{0}{0}{$f_{\x\y}$}
}}
\draw[thick, orange] ($ 1.4*(1,0) $) -- ($ (0,.7) + 1.4*(1,1) $);
\fill[white] ($ 1.4*(.9,.23) $) rectangle ($ 1.4*(1.1,.65) $);
\fill[blue!30, rounded corners=3pt] (.7,.1) -- (.5,.1) -- (.1,.5) -- (.1,.9) -- (.5,1.3) -- (.9,1.3) -- (1.3,.9) -- (1.5,.9) -- (1.9,1.3) -- (2.3,1.3) -- (2.7,.9) -- (2.7,.5) -- (2.3,.1) -- (1.9,.1) -- (1.5,.5) -- (1.3,.5) -- (.9,.1) -- (.7,.1);
%\node[orange] at (1.4,.7) {\scriptsize{excitation}};
\roundNbox{fill=white}{($ 1.4*(1,0)$)}{.3}{0}{0}{$f_{10}$}
\roundNbox{orange, fill=white}{($ 1.4*(1,1)$)}{.3}{0}{0}{\textcolor{orange}{$f_{11}$}}
}
\quad\rightsquigarrow\quad
\tikzmath{
\foreach \x in {0,1,2}{
\foreach \y in {0,1}{
\draw ($ (-.7,0) + 1.4*(\x,\y) $) -- ($ (-.3,0) + 1.4*(\x,\y) $);
\draw ($ (.7,0) + 1.4*(\x,\y) $) -- ($ (.3,0) + 1.4*(\x,\y) $);
\draw ($ (0,-.7) + 1.4*(\x,\y) $) -- ($ (0,-.3) + 1.4*(\x,\y) $);
\draw ($ (0,.7) + 1.4*(\x,\y) $) -- ($ (0,.3) + 1.4*(\x,\y) $);
\roundNbox{fill=white}{($ 1.4*(\x,\y)$)}{.3}{0}{0}{$f_{\x\y}$}
}}
\fill[white] ($ 1.4*(.9,.23) $) rectangle ($ 1.4*(1.1,.65) $);
\filldraw[thick,cyan,fill=blue!30, rounded corners=3pt] (.7,.1) -- (.5,.1) -- (.1,.5) -- (.1,.9) -- (.5,1.3) -- (.9,1.3) -- (1.3,.9) -- (1.5,.9) -- (1.9,1.3) -- (2.3,1.3) -- (2.7,.9) -- (2.7,.5) -- (2.3,.1) -- (1.9,.1) -- (1.5,.5) -- (1.3,.5) -- (.9,.1) -- (.7,.1);
\draw (1.4,.9) -- (1.4,.7);
\fill[orange] (1.4,.9) circle (.05cm);
\draw[thick, orange] (1.4,.7) -- ($ (0,.7) + 1.4*(1,1) $);
\roundNbox{fill=white}{($ 1.4*(1,0)$)}{.3}{0}{0}{$f_{10}$}
\roundNbox{orange, fill=white}{($ 1.4*(1,1)$)}{.3}{0}{0}{\textcolor{orange}{$f_{11}$}}
}
$$
As before, we may apply a projection $p_1$ to $|\psi\rangle$ at the local Hilbert space below this edge $\ell$.
Since $\cX\subset \cY$ is a full subcategory, any tensorand 
$$
f_{10}\in \cY(X\otimes X\to 1_\cY\otimes \cX) = \cX(X\otimes X\to X)
$$ 
in the linear combination describing $|\psi\rangle$ lies in $\cX$.
The resulting space of excitations then hosts an action of the  \emph{relative tube algebra}, which is Morita equivalent to the `strip algebra' from \cite{MR2942952} (see also \cite{MR3975865}).
Hence the action of this algebra allows us to measure excitations at the end of the defect line.

\begin{defn}
The \emph{relative tube algebra} $\Tube_\cX(\cY)$ of the inclusion $\cX\subset \cY$ is 
\begin{align*}
\bigoplus_{\substack{x\in\Irr(\cX)
\\ y,z\in\Irr(\cY)}} \cY(x\otimes y \to z\otimes x)
&\cong
\bigoplus_{x\in\Irr(\cX)}
\cY(x\otimes Y\otimes \overline{x}\to Y)
\cong
\cY(F\Tr_\cX(Y)\to Y)
\\&\cong
Z_\cX(\cY)(\Tr_\cX(Y)\to \Tr_\cX(Y)).
\end{align*}
Here, $\Tr_\cX: \cY \to Z_\cX(\cY)$ is the unitary adjoint \cite{MR4750417} of the forgetful functor $F: Z_\cX(\cY)\to\cY$, which admits the structure of a \emph{relative} categorified trace, i.e., there are canonical \emph{relative traciator} unitary natural isomorphisms
$$
\Tr_\cX(x\otimes y) \to \Tr_\cX(y\otimes x)
\qquad\qquad\qquad
\forall x\in\cX,\, y\in\cY.
$$
\end{defn}

\begin{ex}
When $\cY$ is a $G$-graded extension, $Z_\cX(\cY)$ is a $G$-crossed braided extension of $Z(\cX)$ \cite{MR2587410}.
In this case, the $\cY$-defect line should be viewed as the direct sum of all invertible $g$-defect lines in the usual SETO story \cite{1410.4540}.
\end{ex}

\begin{rem}
\label{rem:ConstructionForBimodule}
The construction of this section can be adapted for any $\cX-\cX$ bimodule category $\cM$, where the twist defects would be given by $Z_\cX(\cM)\cong\Hom_{\cX-\cX}(\cX\to \cM)$.
However, we will not have a fusion operation on twist defects without a monoidal structure on $\cM$.
We require that $\cM$ comes equipped with a \emph{unitary trace} in the sense of \cite{MR3019263,MR4598730}, leading to skein module inner products as in \cite[\S{3.1}]{2506.19969}.
$$
\tikzmath{
\draw[step=.5] (-.25,-.25) grid (2.25,2.25);
\draw[thick, orange] (1,2.25) -- (1,.75);
\filldraw[orange] (1,1) circle (.05cm);
\fill[white] (.98,.51) rectangle (1.02,.8);
}
\qquad
\cK_v:=\cM(X\rhd \textcolor{orange}{M}\to \textcolor{orange}{M}\lhd X)
=
\left\{\tikzmath{
\draw (-.7,0) -- (-.3,0);
\draw (.7,0) -- (.3,0);
\draw[thick, orange] (0,-.7) -- (0,-.3);
\draw[thick, orange] (0,.7) -- (0,.3);
\draw[thin, cyan] (-.7,.7) -- (.7,-.7);
\draw[->, cyan, thin] (-.6,.4) -- (-.4,.6);
\draw[<-, cyan, thin] (.6,-.4) -- (.4,-.6);
\roundNbox{fill=white}{(0,0)}{.3}{0}{0}{$f$}
}
\right\}
\qquad
M:=\bigoplus_{m\in\Irr(\cM)} m
$$
For general $\cM$, there is no obvious canonical choice of Hilbert space to put at the end of the defect line.
Thus in our model, the end of the defect line always hosts a point excitation at the defect site.
(In \cite{1912.01760}, the authors work with pointed bimodules, where the pointing identifies a particular twist defect to be considered as the local ground states.)
\end{rem}

%%%%%%%%%%%%%%%%%%%%%%%%%%%%%%%%%%%%%%%%%%%%%%%%%%%%
\subsection{Fusion of twist defects}
\label{sec:FusionOfTwistDefects}

We can use a modified graphical calculus of strings on tubes for $\Tr_\cX: \cY\to Z_\cX(\cY)$ to describe the tensor product of $\Tube_\cX(\cY)$-representations using Day convolution and the unitary Yoneda Lemma.
Again, we draw string diagrams from $\cY$ on tubes, but now only strings from $\cX\subset\cY$ are allowed to travel around the back of the tube.
$$
\Tube_\cX(\cY)
\ni
\tikzmath{
\draw[thick, orange] (-.2,.3) --node[left]{$\scriptstyle Y$} (-.2,.7);
\draw[thick, cyan] (.2,.3) --node[right]{$\scriptstyle x$} (.2,.7);
\draw[thick, cyan] (-.2,-.3) --node[left]{$\scriptstyle x$} (-.2,-.7);
\draw[thick, orange] (.2,-.3) --node[right]{$\scriptstyle Y$} (.2,-.7);
\roundNbox{}{(0,0)}{.3}{.1}{.1}{$f$}
}
\longleftrightarrow
%\overset{\Tr_\cX}{\longmapsto}
\,\,
\tikzmath{
\draw[thick] (.5,-.6) -- (.5,1);
\draw[thick] (-.5,-.6) -- (-.5,1);
\draw[thick] (0,1) ellipse (.5 and .2);
\draw[thick] (.5,-.6) arc(0:-180:.5 and .2);
\draw[thick,dotted] (.5,-.6) arc(0:180:.5 and .2);
\draw[thick, orange] (0,-.8) --node[left]{$\scriptstyle Y$} (0,-.3);
\draw[thick, orange] (0,.8) --node[left]{$\scriptstyle Y$} (0,.3);
\draw[thick, cyan] (.5,.2) arc (0:-180:.5 and .2);
\draw[thick, cyan, dotted] (.5,.2) arc (0:180:.5 and .2);
\roundNbox{fill=white}{(0,0)}{.3}{0}{0}{$f$}
}
\in\End_{Z_\cX(\cY)}(\Tr_\cX(Y))
$$
Given $w,z\in Z_\cX(\cY)$, fusion of the $\Tube_\cX(\cY)$-representations
$\cK_w:=Z_\cX(\cY)(\Tr_\cX(Y)\to w)$ 
and
$\cK_z:=Z_\cX(\cY)(\Tr_\cX(Y)\to z)$ 
is given by
$$
\tikzmath{
\draw[thick] (0,0) arc (0:-180:.3 and .1);
\draw[thick, dotted] (0,0) arc (0:180:.3 and .1);
\draw[thick] (-.5,1) arc (0:-180:.3 and .1);
\draw[thick, dotted] (-.5,1) arc (0:180:.3 and .1);
\draw[thick] (.5,1) arc (0:-180:.3 and .1);
\draw[thick, dotted] (.5,1) arc (0:180:.3 and .1);
\draw[thick] (-.1,1) arc (0:-180:.2cm);
\draw[thick] (-.5,1) -- (-.5,1.3) arc (0:180:.3cm) -- (-1.1,1);
\draw[thick] (.5,1) -- (.5,1.3) arc (0:180:.3cm) -- (-.1,1);
\node at (-.8,1.3) {$\scriptstyle\cK_w$};
\node at (.2,1.3) {$\scriptstyle\cK_z$};
\draw[thick, DarkGreen, snake] (-.8,1.6) --node[left]{$\scriptstyle w$} (-.8,2);
\draw[thick, DarkGreen, snake] (.2,1.6) --node[left]{$\scriptstyle z$} (.2,2);
\draw[thick] (0,0) to[out=90,in=-90] (.5,1);
\draw[thick] (-.6,0) to[out=90,in=-90] (-1.1,1);
\draw[thick, orange] (-.3,-.1) -- (-.3,.4);
\draw[thick, orange] (-.3,.4) to[out=45,in=-90] (.2,.9);
\draw[thick, orange] (-.3,.4) to[out=135,in=-90] (-.8,.9);
\filldraw[orange] (-.3,.4) circle (.05cm);
}
\quad
:=
\bigoplus_{a,b\in\Irr(\cY)}
Z_\cX(\cY)(\Tr_\cX(a)\to w)
\otimes
\cY(Y\to a\otimes b)
\otimes
Z_\cX(\cY)(\Tr_\cX(b)\to z).
$$

On the lattice, we visualize this Day convolution product similarly as before, where now we include our $\cY$-defect line.
$$
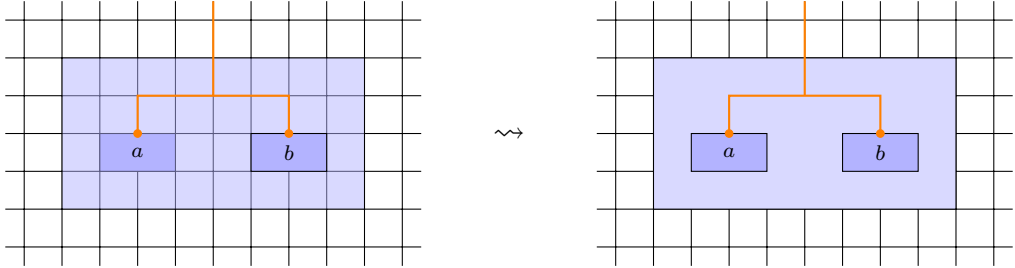
\tikzmath{
\draw[step=.5] (-.25,-.25) grid (5.25,3.25);
\filldraw[fill=blue!30] (1,1) rectangle (2,1.5);
\filldraw[fill=blue!30, opacity=.5] (.5,.5) rectangle (4.5,2.5);
\node at (1.5,1.25) {$\scriptstyle a$};
\filldraw[fill=blue!30] (3,1) rectangle (4,1.5);
\filldraw[orange] (1.5,1.5) circle (.05);
\filldraw[orange] (3.5,1.5) circle (.05);
\node at (3.5,1.25) {$\scriptstyle b$};
\draw[thick,orange] (1.5,1.5) -- (1.5,2) -- (2.5,2) -- (2.5,3.25);
\draw[thick,orange] (3.5,1.5) -- (3.5,2) -- (2.5,2);
}
\qquad\rightsquigarrow\qquad
\tikzmath{
\draw[step=.5] (-.25,-.25) grid (5.25,3.25);
\filldraw[fill=blue!15] (.5,.5) rectangle (4.5,2.5);
\filldraw[fill=blue!30] (1,1) rectangle (2,1.5);
\node at (1.5,1.25) {$\scriptstyle a$};
\filldraw[fill=blue!30] (3,1) rectangle (4,1.5);
\node at (3.5,1.25) {$\scriptstyle b$};
\filldraw[orange] (1.5,1.5) circle (.05);
\filldraw[orange] (3.5,1.5) circle (.05);
\draw[thick,orange] (1.5,1.5) -- (1.5,2) -- (3.5,2) -- (3.5,1.5);
\draw[thick,orange] (2.5,2) -- (2.5,3.25);
}
$$
Note that the edges in black are labelled by objects in $\cX$ (especially on the boundaries of the locations of the $a,b$ anyons!), and the edges in orange are labelled by objects in $\cY$.
The action of the relative tube algebra $\Tube_\cX(\cY)$ on the Day convolution is similar to the previous action.
$$
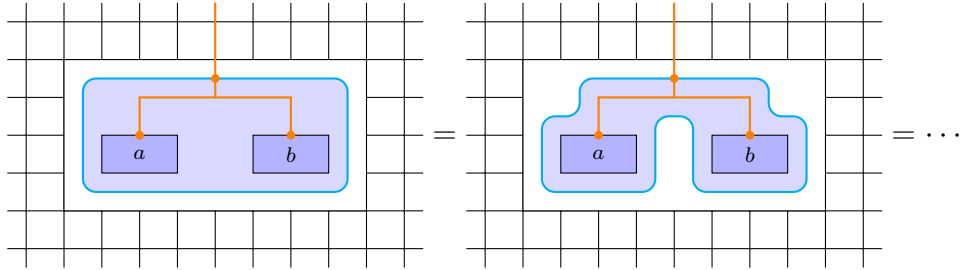
\tikzmath{
\draw[step=.5] (-.25,-.25) grid (5.25,3.25);
\filldraw[fill=white] (.5,.5) rectangle (4.5,2.5);
\filldraw[fill=blue!15, thick, draw=cyan, rounded corners=5pt] (.75,.75) rectangle (4.25,2.25);
\filldraw[fill=blue!30] (1,1) rectangle (2,1.5);
\node at (1.5,1.25) {$\scriptstyle a$};
\filldraw[fill=blue!30] (3,1) rectangle (4,1.5);
\node at (3.5,1.25) {$\scriptstyle b$};
\filldraw[orange] (1.5,1.5) circle (.05);
\filldraw[orange] (3.5,1.5) circle (.05);
\draw[thick,orange] (1.5,1.5) -- (1.5,2) -- (3.5,2) -- (3.5,1.5);
\draw[thick,orange] (2.5,2) -- (2.5,3.25);
\filldraw[orange] (2.5,2.25) circle (.05cm);
}
=
\tikzmath{
\draw[step=.5] (-.25,-.25) grid (5.25,3.25);
\filldraw[fill=white] (.5,.5) rectangle (4.5,2.5);
\filldraw[fill=blue!15, thick, draw=cyan, rounded corners=5pt] (2.5,2.25) -- (1.25,2.25) -- (1.25,1.75) -- (.75,1.75) -- (.75,.75) -- (2.25,.75) -- (2.25,1.75) -- (2.75,1.75) -- (2.75,.75) -- (4.25,.75) -- (4.25,1.75) -- (3.75,1.75) -- (3.75,2.25) -- (2.5,2.25);
\filldraw[fill=blue!30] (1,1) rectangle (2,1.5);
\node at (1.5,1.25) {$\scriptstyle a$};
\filldraw[fill=blue!30] (3,1) rectangle (4,1.5);
\node at (3.5,1.25) {$\scriptstyle b$};
\filldraw[orange] (1.5,1.5) circle (.05);
\filldraw[orange] (3.5,1.5) circle (.05);
\draw[thick,orange] (1.5,1.5) -- (1.5,2) -- (3.5,2) -- (3.5,1.5);
\draw[thick,orange] (2.5,2) -- (2.5,3.25);
\filldraw[orange] (2.5,2.25) circle (.05cm);
}
=
\cdots
$$

In the language of higher categories, we can view the above diagram as the implementation of the equivalence 
$$
\Bim(\cX)\ni \cM \longmapsto \Hom_{\cX-\cX}(\cX\to \cM)\in \Mod(Z(\cX))
$$
from \cite[Thm.~7.14]{MR2678824}:
the fusion category $\cY$ viewed as an $E_1$-algebra object in $\Bim(\cX)$ is mapped to the $E_1$-algebra $Z_\cX(\cY)$ in $\Mod(Z(\cX))$.

%%%%%%%%%%%%%%%%%%%%%%%%%%%%%%%%%%%%%%%%%%%%%%%%%%%%
%%%%%%%%%%%%%%%%%%%%%%%%%%%%%%%%%%%%%%%%%%%%%%%%%%%%
%%%%%%%%%%%%%%%%%%%%%%%%%%%%%%%%%%%%%%%%%%%%%%%%%%%%
\section{Symmetry action from tube algebra}

\subsection{Domain wall wrapping symmetry action}
\label{subsec:domainwallwrap}

As a (possibly non-invertible) symmetry action on an anyon $a\in Z(\cX)$ of the $\cX$ string-net, we consider wrapping the $a$ anyon in a domain wall corresponding to the $\cX-\cX$ bimodule $\cY$.
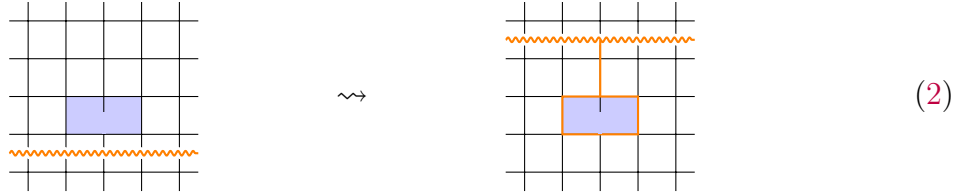
\begin{equation}
\tikzmath{
\fill[blue!20] (.5,.5) rectangle (1.5,1);
\draw[step=.5] (-.25,-.25) grid (2.25,2.25);
\fill[blue!20] (.98,.51) rectangle (1.02,.8);
\draw[thick, orange, knot, snake] (-.25,.25)  -- (2.25,.25);
}
\qquad\qquad
\rightsquigarrow
\qquad\qquad
\tikzmath{
\fill[blue!20] (.5,.5) rectangle (1.5,1);
\draw[step=.5] (-.25,-.25) grid (2.25,2.25);
\draw[thick, orange] (.5,.5) rectangle (1.5,1);
\fill[blue!20] (.98,.51) rectangle (1.02,.8);
\draw[thick, orange, knot, snake] (-.25,1.75) -- (2.25,1.75);
\draw[thick, orange] (1,1) rectangle (1,1.75);
}
\tag{\ref{eq:SweepDomainWall}}
\end{equation}

To determine which anyons and twist defects appear after wrapping the domain wall, we use tube algebra methods.
The initial anyon $a \in Z(\cX)$ corresponds to an idempotent $P_a \in \Tube(\cX)$. 
Applying the full symmetry action means wrapping by the domain wall determined by the $\cX-\cX$ bimodule $\cY$.
We therefore stack $P_a\in\Tube(\cX)$ above $\Tube(\cY)$, and let $\Tube_{\cX}(\cY)$ act on $\Tube(\cY)$ from below to obtain a new $\Tube_{\cX}(\cY)$ representation corresponding to the image of $a$ under the symmetry action.
\begin{align*}
    P_a = \tikzmath[scale=0.85]{
    \draw[thick] (.6,-.6) -- (.6,1);
    \draw[thick] (-.6,-.6) -- (-.6,1);
    \draw[thick] (0,1) ellipse (.6 and .2);
    \draw[thick] (.6,-.6) arc(0:-180:.6 and .2);
    \draw[thick,dotted] (.6,-.6) arc(0:180:.6 and .2);
    \draw[] (0,-.8) node[above left] {$\scriptstyle X$} -- (0,-.25);
    \draw[] (0,.8) node[below left] {$\scriptstyle X$} -- (0,.25); % weirdly (0,.3) isn't low enough, although (0,-.3) is high enough. I guess $f$ is a weird character?
    \draw[thick, cyan] (.6,.2) arc (0:-180:.6 and .2);
    \draw[thick, cyan, dotted] (.6,.2) arc (0:180:.6 and .2);
    \node[blue, right] at (.6,0.2) {\small $\in \Tube(\cX)$};
    \roundNbox{fill=white}{(0,0)}{.25}{0}{0}{$\scriptstyle f$}}
    \mapsto \quad \operatorname{span}\left\{ \; \tikzmath[scale=0.7]{\draw[thick] (.6,-.5) -- (.6,2.3);
    \draw[thick] (-.6,-.5) -- (-.6,2.3);
    \draw[thick] (0,2.3) ellipse (.6 and .2);
    \draw[thick] (.6,0.9) arc(0:-180:.6 and .2);
    \draw[thick,dotted] (.6,0.9) arc(0:180:.6 and .2);
    \draw[thick] (.6,-.5) arc(0:-180:.6 and .2);
    \draw[thick,dotted] (.6,-.5) arc(0:180:.6 and .2);
    \draw[thick, orange] (0,-.7) -- (0,-.2);
    \draw[] (0,1.4) -- (0,2.1);
    \draw[decoration={brace,mirror,raise=8pt},decorate, thick, blue]
    (0.55,1.) -- (0.55,2.3);
    \node[blue, right] at (1.,1.6) {\footnotesize $P_a \in \Tube(\cX)$}; 
    \draw[decoration={brace,mirror,raise=8pt},decorate, thick, orange]
    (0.55,-0.5) -- (0.55,0.8);
    \node[orange, right] at (1.,0.2) {\footnotesize $T \in \Tube(\cY)$}; 
    \draw[] (0,1.3) -- (0,.2);
    \draw[thick, magenta] (.6,.2) arc (0:-180:.6 and .2);
    \draw[thick, magenta, dotted] (.6,.2) arc (0:180:.6 and .2);
    \draw[thick, cyan] (.6,1.6) arc (0:-180:.6 and .2);
    \draw[thick, cyan, dotted] (.6,1.6) arc (0:180:.6 and .2);
    \roundNbox{fill=white}{(0,0)}{.26}{0}{0}{$\scriptstyle g$}
    \roundNbox{fill=white}{(0,1.4)}{.26}{0}{0}{$\scriptstyle f$}}\right\} = P_a\Tube(\cY)
\end{align*}
In other words, we apply the inclusion 
\begin{equation}
    \label{eq:fullInclusionRelativeTensor}
    \Tube(\cX)\hookrightarrow\Tube(\cX)\otimes_{\Tube_{\cX}(\cY)}\Tube(\cY)\otimes_{\Tube_{\cX}(\cY)}\Tube_{\cX}(\cY)\cong\Tube(\cY)1_{\Tube(\cX)}
\end{equation}
to $P_a$.
The final simplification in \eqref{eq:fullInclusionRelativeTensor} is possible since $\Tube(\cX)\subseteq\Tube_{\cX}(\cY)\subseteq\Tube(\cY)$ is a chain of inclusions of $\dag$-subalgebras.
The inclusion $\Tube(\cX)\hookrightarrow\Tube_{\cX}(\cY)$ is nonunital, resulting in the compression $\Tube(\cY)1_{\Tube(\cX)}$, which is just the subspace of $\Tube(\cY)$ in which the labels on the top vertical strings must lie in $\cX$. 

To detect the twist defect and anyon types appearing in the resulting excitation, we can stack from below with an idempotent $P_b \in \Tube_\cX(\cY)$:
\begin{equation}
    \label{eq:fullInclusionTubeStack}
    P_b \Tube(\cY) P_a  = \operatorname{span}\left\{
    \tikzmath[scale=0.7]{\draw[thick] (.6,-1.9) -- (.6,2.3);
    \draw[thick] (-.6,-1.9) -- (-.6,2.3);
    \draw[thick] (0,2.3) ellipse (.6 and .2);
    \draw[thick] (.6,0.9) arc(0:-180:.6 and .2);
    \draw[thick,dotted] (.6,0.9) arc(0:180:.6 and .2);
    \draw[thick] (.6,-.5) arc(0:-180:.6 and .2);
    \draw[thick,dotted] (.6,-.5) arc(0:180:.6 and .2);
    \draw[thick, orange] (0,-1.4) -- (0,-.2);
    \draw[] (0,1.4) -- (0,2.1);
    \draw[decoration={brace,mirror,raise=8pt},decorate, thick, blue]
    (0.55,1.) -- (0.55,2.3);
    \node[blue, right] at (1.,1.6) {\footnotesize $P_a \, \in \, \Tube(\cX)$}; 
    \draw[decoration={brace,mirror,raise=8pt},decorate, thick, orange]
    (0.55,-0.5) -- (0.55,0.8);
    \node[orange, right] at (1.,0.2) {\footnotesize $T \, \in \, \Tube(\cY)$}; 
    \draw[decoration={brace,mirror,raise=8pt},decorate, thick, olive]
    (0.55,-1.9) -- (0.55,-0.6);
    \node[olive, right] at (1.,-1.2) {\footnotesize $P_b \, \in \, \Tube_\cX(\cY)$}; 
    \draw[] (0,1.3) -- (0,.2);
    \draw[thick, magenta] (.6,.2) arc (0:-180:.6 and .2);
    \draw[thick, magenta, dotted] (.6,.2) arc (0:180:.6 and .2);
    \draw[thick, cyan] (.6,1.6) arc (0:-180:.6 and .2);
    \draw[thick, cyan, dotted] (.6,1.6) arc (0:180:.6 and .2);
    \node[magenta, left] at (-0.6,0.2) {$\scriptstyle y \in \cY$};
    \node[olive, left] at (-0.6,-1.3) {$\scriptstyle x \in \cX$};
    %\node[orange, below] at (0,-2.1) {$\scriptstyle z \in \cY$};
    \draw[thick, olive] (0.6,-1.2) arc (0:-180:.6 and .2);
    \draw[thick, olive, dotted] (0.6,-1.2) arc (0:180:.6 and .2);
    \draw[thick] (0.6,-1.9) arc (0:-180:.6 and .2);
    \draw[thick, dotted] (0.6,-1.9) arc (0:180:.6 and .2);
    \draw[thick, orange] (0,-1.4) -- (0,-2.1);
    \roundNbox{fill=white}{(0,0)}{.26}{0}{0}{$\scriptstyle g$}
    \roundNbox{fill=white}{(0,1.4)}{.26}{0}{0}{$\scriptstyle f$}
    \roundNbox{fill=white}{(0,-1.4)}{.26}{0}{0}{$\scriptstyle k$}
    }
    \right\}
\end{equation}
If the result is nonzero, it means the anyon or twist defect $b \in Z_\cX(\cY)$ appears in the wrapped domain wall excitation. For non-invertible symmetries, the multiplicity $N^d_{a b}$ with which $b$ appears can be larger than one. It can be computed from the dimension of the vector space obtained by stacking the three tubes as above. 
The symmetry action can therefore be written as 
\begin{equation}
 \label{eq:symMultCentral}
 a \mapsto \sum_j N_{a b_j} b_j, \quad \text{for } b_j \in Z_\cX(\cY) \; \text{ with } 
 N_{a,b_j} = \frac{\dim\left(P_{b_j} \Tube(\cY) P_a\right)}{\sqrt{\dim(\Tube(\cX)P_a)}\sqrt{\dim(P_{b_j}\Tube_{\cX}(\cY))}}.
\end{equation}
Alternatively, one can select rank $1$ projectors $p_a\in P_a\Tube(\cX)$ and $p_b\in P_b\Tube_{\cX}(\cY)$, in which case the multiplicity becomes\footnote{
A rank $1$ projector in $P_a\Tube(\cX)$ is a projector which selects one of the internal states of an $a$ anyon, {i.e.}~projects onto a $1D$ subspace of the corresponding $\Tube(\cX)$ irreducible representation. Such a $p_a$ is not central (unless $P_a\Tube(\cX)$ is already one-dimensional), since tube algebra elements in $P_a\Tube(\cX)$ can exchange the various internal states.
In particular, $p_aP_a=p_a$.}
\begin{equation}
 \label{eq:symMultRank1}
 N_{a,b_j}=\dim(p_{b_j}\Tube(\cY)p_a).
\end{equation}

\begin{rem}
    As noted above, it is also possible to consider twist defects as inputs to the symmetry action.
    To do so, we need only choose $a\in Z_{\cX}(\cY)$ and consider the corresponding $P_a\in\Tube_{\cX}(\cY)$.
    The procedure for computing the image of $a$ under the symmetry action is otherwise unchanged.
\end{rem}

%%%%%%%%%%%%%%%%%%%%%%%%%%%%%%%%%%%%%%%%%%
\subsubsection{Why this works: induction/restriction}
\label{ssec:indRes}
The above diagrammatic calculation can also be interpreted via induction and restriction functors.
Recall that we have an inclusion $Z(\cX)\hookrightarrow Z_\cX(\cY)$ and a dominant forgetful functor $Z(\cY)\twoheadrightarrow Z_\cX(\cY)$.
Dually, we have a projection $\Tube_\cX(\cY)\twoheadrightarrow \Tube(\cX)$ and an inclusion $\Tube_\cX(\cY)\hookrightarrow \Tube(\cY)$, as taking modules is contravariant.
Thus the above calculation is exactly taking a $\Tube(\cX)$-module (anyon in $Z(\cX)$), thinking of it in $Z_\cX(\cY)$ via the inclusion, inducing it up to $Z(\cY)$ by stacking with $\Tube(\cY)$, and then restricting it again to $Z_\cX(\cY)$ by looking at the action an idempotent in $\Tube_\cX(\cY)$.
In short, we are computing the composition
\begin{equation}
    \label{eq:fullInclusionIndRes}
    Z(\cX)\hookrightarrow Z_{\cX}(\cY)\to Z(\cY)\twoheadrightarrow Z_{\cX}(\cY),
\end{equation}
where the induction functor $Z_{\cX}(\cY)\to Z(\cY)$ is adjoint to the restriction $Z(\cY)\twoheadrightarrow Z_{\cX}(\cY)$.

To see that \eqref{eq:fullInclusionIndRes} is the correct functor to implement the full generalized symmetry action, we interpret the domain wall given by the $\cX-\cX$ bimodule $\cY$ as the stacking of two domain walls, given by the bimodules ${}_\cX \cY \boxtimes_\cY \cY_\cX$.
We know that $Z(\cX)$ and $Z(\cY)$ are the correct anyon theories for $\cX$ and $\cY$ string-net bulk regions, that $Z_{\cX}(\cY):=\End_{\cX-\cY}(\cY)$ is the bimodule multifusion category of wall excitations, and that inclusion and restriction $Z(\cX)\hookrightarrow Z_{\cX}(\cY)\twoheadleftarrow Z(\cY)$ are the correct functors describing the wall excitations obtained when bringing anyons from either bulk to the domain wall by \cite{MR2942952}.
The overall process of tunneling an anyon from the $Z(\cX)$ bulk, through the first $Z(\cX)-Z(\cY)$ domain wall into the $Z(\cY)$ bulk, and then onto the second $Z(\cY)-Z(\cX)$ domain wall is therefore given by \eqref{eq:fullInclusionIndRes}, {cf.}~\cite[Prop.~IV.5]{MR4640433}.

\subsubsection{Computing actions of individual symmetries}
\label{sssec:individualFull}
In the case where $\cY$ is a $G$-graded extension of $\cX$, {i.e.}~when our symmetries are all invertible, wrapping by the $\cX-\cX$ bimodule $\cY$ is the sum of the actions of each $g\in G$.
The simple summands of $\cY$ as an $\cX-\cX$ bimodule are precisely the $g$-graded components $\cY_g$ of our $G$-graded extension, and wrapping by $\cY_g$ implements the $g$-symmetry action.
Therefore, in addition to computing the effect of wrapping by $\cY$, we would like to be able to compute the action of wrapping by a simple $\cX-\cX$ bimodule summand of $\cY$, which we think of as applying an individual non-invertible symmetry generator.

We may decompose 
\[\cY\cong\bigoplus_{j=0}^{n-1}\cY_j,\]
where each $\cY_j$ is an indecomposable $\cX-\cX$ bimodule and $\cY_0=\cX1_{\cY}\cX=\cX$.
Then $\Tube(\cY)$ also splits into $n$ $\Tube_{\cX}(\cY)-\Tube_{\cX}(\cY)$ bimodules
\[\Tube_{\cY_j}(\cY):=\operatorname{span}\set{
\tikzmath[scale=0.85]{
    \draw[thick] (.6,-.6) -- (.6,1);
    \draw[thick] (-.6,-.6) -- (-.6,1);
    \draw[thick] (0,1) ellipse (.6 and .2);
    \draw[thick] (.6,-.6) arc(0:-180:.6 and .2);
    \draw[thick,dotted] (.6,-.6) arc(0:180:.6 and .2);
    \draw[thick,orange] (0,-.8) node[above left] {$\scriptstyle y$} -- (0,-.25);
    \draw[thick,orange] (0,.8) node[below left] {$\scriptstyle z$} -- (0,.25);
    \draw[thick, magenta] (.6,.2) arc (0:-180:.6 and .2);
    \draw[thick, magenta, dotted] (.6,.2) arc (0:180:.6 and .2);
     \node[magenta, left] at (-0.6,0.1) {$\scriptstyle y_j$};
    % \node[blue, right] at (.6,0.2) {\small $\in \Tube(\cX)$};
    \roundNbox{fill=white}{(0,0)}{.25}{0}{0}{$\scriptstyle f$}}
\;}{\, y,z\in\Irr(\cY),\, y_j\in\Irr(\cY_j)}\]

Stacking with $\Tube_{\cY_j}(\cY)$ rather than the full $\Tube(\cY)$ implements the symmetry action of $\cY_j$.

Since $Z_{\cX}(\cY):=\Hom_{\cX-\cX}(\cX\to\cY)$, we also have a grading
\[Z_{\cX}(\cY):=\bigoplus_{j=0}^{n-1}Z_{\cX}(\cY_j)\]
where $Z_{\cX}(\cY_j)$ is the center of a bimodule category as in \cite{MR1151906,MR2587410}.
In the $G$-graded case, this is just $G$-grading on $Z_{\cX}(\cY)$, which is a $G$-crossed braided extension of $Z(\cX)$ \cite{MR2587410}.
As in \eqref{eq:symMultCentral}, the symmetry action of $\cY_j$ on anyons can be written as
\[a\mapsto \sum_k N^j_{a,b_k}b_k,\quad\text{ for }b_k\in Z_{\cX}(\cY_j)\]
with 
\[N^j_{a,b_k}=\frac{\dim(P_{b_k}\Tube_{\cY_j}(\cY)P_a)}{\sqrt{\dim(\Tube(\cX)P_a)}\sqrt{\dim(\Tube_{\cX}(\cY_j)P_{b_k})}}\]
Alternatively, we may pick rank $1$ projectors $p_a$ and $p_{b_k}$ as above and obtain
\begin{equation}
 \label{eq:indivSymMultRank1}
 N^j_{a,b_k}=\dim(p_{b_k}\Tube_{\cY_j}(\cY)p_a)
\end{equation}
analogous to \eqref{eq:symMultRank1}.

\begin{rem}
 Since $\cY_0\cong\cX$ as an $\cX-\cX$ bimodule, $\Tube_{\cY_0}(\cY)=\Tube_{\cX}(\cY)$, and 
 \[P_a\Tube_{\cY_0}(\cY)=P_a\Tube_{\cX}(\cY)=P_a\Tube(\cX)=\Tube_{\cX}(\cY)P_a\] 
 We therefore see that the trivial symmetry sends each anyon to itself, as expected.
\end{rem}

\subsection{\texorpdfstring{$G$}{G}-graded example with \texorpdfstring{$\cY=$}{Y=} Ising}

When $\cY$ is a $G$-graded extension of $\cX$, the above construction essentially reduces to the computation of group symmetry actions in \cite{1711.07982}. They denote the linear combination of domain wall tubes stacked onto the anyon idempotent as 
$$
B^h_a = \left(\sum_i c_i T^h_i \right) P_a,
$$
where the tubes $T^h_i\in\Tube_{\cY_h}(\cY)$ have circumference labels in the graded component $\cY_h$ for some $h \in G$. 
After normalizing so that $(B^h_a)^\dagger B^h_a=P_a$, the composite
\[
    P_{h(a)}:=B^h_a(B^h_a)^\dagger
\]
is an idempotent in $\Tube(\cX)$ and hence identifies the permuted anyon $h(a)\in Z(\cX)$. Equivalently, following our method, among the projectors $P_b$, exactly one has
$P_b B^h_a \neq 0$, identifying $h(a)=b$.

For $\cX=\Hilb(\bbZ/2)$, the tube algebra idempotents for the toric code anyons in $Z(\cX)$ are well-known and given by
\begin{align*}
    P_1 &= \frac{1}{2} \left(
    \tikzmath[scale=0.8]{\draw[thick] (.5,-.6) -- (.5,1);
    \draw[thick] (-.5,-.6) -- (-.5,1);
    \draw[thick] (0,1) ellipse (.5 and .2);
    \draw[thick] (.5,-.6) arc(0:-180:.5 and .2);
    \draw[thick,dotted] (.5,-.6) arc(0:180:.5 and .2);} + 
    \tikzmath[scale=0.8]{\draw[thick] (.5,-.6) -- (.5,1);
    \draw[thick] (-.5,-.6) -- (-.5,1);
    \draw[thick] (0,1) ellipse (.5 and .2);
    \draw[thick] (.5,-.6) arc(0:-180:.5 and .2);
    \draw[thick,dotted] (.5,-.6) arc(0:180:.5 and .2);
    \draw[thick, cyan] (.5,.2) arc (0:-180:.5 and .2);
    \draw[thick, cyan, dotted] (.5,.2) arc (0:180:.5 and .2);
    \node[cyan] at (.0,-.2) {$\scriptstyle \psi$};
}\right), \quad 
    P_e = \frac{1}{2} \left(
    \tikzmath[scale=0.8]{\draw[thick] (.5,-.6) -- (.5,1);
    \draw[thick] (-.5,-.6) -- (-.5,1);
    \draw[thick] (0,1) ellipse (.5 and .2);
    \draw[thick] (.5,-.6) arc(0:-180:.5 and .2);
    \draw[thick,dotted] (.5,-.6) arc(0:180:.5 and .2);} - 
    \tikzmath[scale=0.8]{\draw[thick] (.5,-.6) -- (.5,1);
    \draw[thick] (-.5,-.6) -- (-.5,1);
    \draw[thick] (0,1) ellipse (.5 and .2);
    \draw[thick] (.5,-.6) arc(0:-180:.5 and .2);
    \draw[thick,dotted] (.5,-.6) arc(0:180:.5 and .2);
    \draw[thick, cyan] (.5,.2) arc (0:-180:.5 and .2);
    \draw[thick, cyan, dotted] (.5,.2) arc (0:180:.5 and .2);
    \node[cyan] at (.0,-.2) {$\scriptstyle \psi$};
}\right), \\
    P_m &=\frac{1}{2} \left(
    \tikzmath[scale=0.8]{\draw[thick] (.5,-.6) -- (.5,1);
    \draw[thick] (-.5,-.6) -- (-.5,1);
    \draw[thick] (0,1) ellipse (.5 and .2);
    \draw[thick] (.5,-.6) arc(0:-180:.5 and .2);
    \draw[thick,dotted] (.5,-.6) arc(0:180:.5 and .2);
    \draw[thick, cyan] (0,-0.8) -- (0,0.8);
    \node[cyan, right] at (-0.1,0) {$\scriptstyle \psi$};
    } + 
    \tikzmath[scale=0.8]{\draw[thick] (.5,-.6) -- (.5,1);
    \draw[thick] (-.5,-.6) -- (-.5,1);
    \draw[thick] (0,1) ellipse (.5 and .2);
    \draw[thick] (.5,-.6) arc(0:-180:.5 and .2);
    \draw[thick,dotted] (.5,-.6) arc(0:180:.5 and .2);
    \draw[thick, cyan] (.5,.) arc (0:-90:.5 and .2);
    \draw[thick, cyan, dotted] (.5,0.) arc (0:180:.5 and .2);
    \draw[thick, cyan] (0,-0.2) -- (0,-0.8);
    \draw[thick, cyan] (-.5,0.) arc (180:90:.5 and .4);
    \draw[thick, cyan] (0,0.4) -- (0,0.8);
    %\node[cyan] at (.0,-.2) {$\scriptstyle \psi$};
}\right), \quad 
    P_{f} = \frac{1}{2} \left(
    \tikzmath[scale=0.8]{\draw[thick] (.5,-.6) -- (.5,1);
    \draw[thick] (-.5,-.6) -- (-.5,1);
    \draw[thick] (0,1) ellipse (.5 and .2);
    \draw[thick] (.5,-.6) arc(0:-180:.5 and .2);
    \draw[thick,dotted] (.5,-.6) arc(0:180:.5 and .2);
    \draw[thick, cyan] (0,-0.8) -- (0,0.8);
    \node[cyan, right] at (-0.1,0) {$\scriptstyle \psi$};
    } - 
    \tikzmath[scale=0.8]{\draw[thick] (.5,-.6) -- (.5,1);
    \draw[thick] (-.5,-.6) -- (-.5,1);
    \draw[thick] (0,1) ellipse (.5 and .2);
    \draw[thick] (.5,-.6) arc(0:-180:.5 and .2);
    \draw[thick,dotted] (.5,-.6) arc(0:180:.5 and .2);
    \draw[thick, cyan] (.5,.) arc (0:-90:.5 and .2);
    \draw[thick, cyan, dotted] (.5,0.) arc (0:180:.5 and .2);
    \draw[thick, cyan] (0,-0.2) -- (0,-0.8);
    \draw[thick, cyan] (-.5,0.) arc (180:90:.5 and .4);
    \draw[thick, cyan] (0,0.4) -- (0,0.8);
    %\node[cyan] at (.0,-.2) {$\scriptstyle \psi$};
}\right).
\end{align*}
We calculate the symmetry action for $\cY=\cT\cY(\bbZ/2,\pm)$ a $\bbZ/2$-Tambara-Yamagami fusion category,\footnote{One may think of $\cT\cY(\bbZ/2,\pm)$ as an Ising category $\mathsf{Ising}_\pm$ where the $\pm$ indicates the Frobenius-Schur indicator \cite{MR2381536} of the non-invertible $\sigma$ anyon, but with the braiding forgotten.
Note that both $\cT\cY(\bbZ/2,\pm)$ agree as $\cX-\cX$ bimodule categories, as the sign $\pm$ only appears in $[F_\sigma^{\sigma\sigma\sigma}]$ \cite{MR1659954}.}
which has simple objects $1,\sigma,\psi$ and fusion rules determined by
$$
\psi\otimes \psi \cong 1
\qquad\qquad
\text{ and } 
\qquad\qquad
\sigma \otimes \sigma \cong 1\oplus \psi.
$$
To compute the action of the nontrivial symmetry, we must stack with a linear combination of domain wall tubes with $\sigma \in \cY_{1}$ running around, {i.e.} an element of $\Tube_{\cY_1}(\cY)$.
In the graphical calculus, we have
$$
\sum_i c_i T_i = c_1 \; 
 \tikzmath[scale=0.8]{\draw[thick] (.5,-.6) -- (.5,1);
    \draw[thick] (-.5,-.6) -- (-.5,1);
    \draw[thick] (0,1) ellipse (.5 and .2);
    \draw[thick] (.5,-.6) arc(0:-180:.5 and .2);
    \draw[thick, dotted] (.5,-.6) arc(0:180:.5 and .2);
    \draw[thick, violet] (.5,.2) arc (0:-180:.5 and .2);
    \draw[thick, violet, dotted] (.5,.2) arc (0:180:.5 and .2);
    \node[violet] at (.0,-.2) {$\scriptstyle \sigma$};
}
+ c_2 \;
\tikzmath[scale=0.8]{\draw[thick] (.5,-.6) -- (.5,1);
    \draw[thick] (-.5,-.6) -- (-.5,1);
    \draw[thick] (0,1) ellipse (.5 and .2);
    \draw[thick] (.5,-.6) arc(0:-180:.5 and .2);
    \draw[thick,dotted] (.5,-.6) arc(0:180:.5 and .2);
    \draw[thick, violet] (.5,.) arc (0:-90:.5 and .2);
    \draw[thick, violet, dotted] (.5,0.) arc (0:180:.5 and .2);
    \draw[thick, cyan] (0,-0.2) -- (0,-0.8);
    \draw[thick, violet] (-.5,0.) arc (180:90:.5 and .4);
    \draw[thick, violet] (0,0.4) -- (0,-0.2);
    \draw[thick, cyan] (0,0.4) -- (0,0.8);
    \node[cyan, right] at (-.1,0.6) {$\scriptstyle \psi$};
    \node[violet] at (0.2,0) {$\scriptstyle \sigma$};
} +  c_3 \;
\tikzmath[scale=0.8]{\draw[thick] (.5,-.6) -- (.5,1);
    \draw[thick] (-.5,-.6) -- (-.5,1);
    \draw[thick] (0,1) ellipse (.5 and .2);
    \draw[thick] (.5,-.6) arc(0:-180:.5 and .2);
    \draw[thick,dotted] (.5,-.6) arc(0:180:.5 and .2);
    \draw[thick, cyan] (0,-0.) -- (0,-0.8);
    \draw[thick, violet] (.5,.2) arc (0:-180:.5 and .2);
    \draw[thick, violet, dotted] (.5,.2) arc (0:180:.5 and .2);
    %\node[violet] at (0.2,0) {$\scriptstyle \sigma$};
} + 
 c_4 \;
\tikzmath[scale=0.8]{\draw[thick] (.5,-.6) -- (.5,1);
    \draw[thick] (-.5,-.6) -- (-.5,1);
    \draw[thick] (0,1) ellipse (.5 and .2);
    \draw[thick] (.5,-.6) arc(0:-180:.5 and .2);
    \draw[thick,dotted] (.5,-.6) arc(0:180:.5 and .2);
    \draw[thick, cyan] (0,-0.) -- (0,0.8);
    \draw[thick, violet] (.5,.2) arc (0:-180:.5 and .2);
    \draw[thick, violet, dotted] (.5,.2) arc (0:180:.5 and .2);
    %\node[violet] at (0.2,0) {$\scriptstyle \sigma$};
} + c_5 \; \tikzmath[scale=0.8]{\draw[thick] (.5,-.6) -- (.5,1);
    \draw[thick] (-.5,-.6) -- (-.5,1);
    \draw[thick] (0,1) ellipse (.5 and .2);
    \draw[thick] (.5,-.6) arc(0:-180:.5 and .2);
    \draw[thick,dotted] (.5,-.6) arc(0:180:.5 and .2);
    \draw[thick, violet] (.5,.) arc (0:-90:.5 and .2);
    \draw[thick, violet, dotted] (.5,0.) arc (0:180:.5 and .2);
    \draw[thick, violet] (0,-0.2) -- (0,-0.8);
    \draw[thick, violet] (-.5,0.) arc (180:90:.5 and .4);
    \draw[thick, violet] (0,0.4) -- (0,0.8);
} + c_6 \; \tikzmath[scale=0.8]{\draw[thick] (.5,-.6) -- (.5,1);
    \draw[thick] (-.5,-.6) -- (-.5,1);
    \draw[thick] (0,1) ellipse (.5 and .2);
    \draw[thick] (.5,-.6) arc(0:-180:.5 and .2);
    \draw[thick,dotted] (.5,-.6) arc(0:180:.5 and .2);
    \draw[thick, violet] (.5,.) arc (0:-90:.5 and .2);
    \draw[thick, violet, dotted] (.5,0.) arc (0:180:.5 and .2);
    \draw[thick, violet] (0,-0.2) -- (0,-0.8);
    \draw[thick, violet] (-.5,0.) arc (180:90:.5 and .4);
    \draw[thick, cyan] (0,0.4) -- (0,-0.2);
    \draw[thick, violet] (0,0.4) -- (0,0.8);
} 
$$
We start by stacking these domain wall tubes onto the idempotent for the $e$ anyon: 
\begin{align*}
    &P_e \mapsto \Big(\sum_i c_i T_i \Big) \frac{1}{2} \left(
    \tikzmath[scale=0.65]{\draw[thick] (.5,-.6) -- (.5,1);
    \draw[thick] (-.5,-.6) -- (-.5,1);
    \draw[thick] (0,1) ellipse (.5 and .2);
    \draw[thick] (.5,-.6) arc(0:-180:.5 and .2);
    \draw[thick,dotted] (.5,-.6) arc(0:180:.5 and .2);} - 
    \tikzmath[scale=0.65]{\draw[thick] (.5,-.6) -- (.5,1);
    \draw[thick] (-.5,-.6) -- (-.5,1);
    \draw[thick] (0,1) ellipse (.5 and .2);
    \draw[thick] (.5,-.6) arc(0:-180:.5 and .2);
    \draw[thick,dotted] (.5,-.6) arc(0:180:.5 and .2);
    \draw[thick, cyan] (.5,.2) arc (0:-180:.5 and .2);
    \draw[thick, cyan, dotted] (.5,.2) arc (0:180:.5 and .2);
    \node[cyan] at (.0,-.2) {$\scriptstyle \psi$};
}\right) \hspace{-0.1cm} = \frac{c_1}{2} \left(  \tikzmath[scale=0.65]{\draw[thick] (.5,-.6) -- (.5,1);
    \draw[thick] (-.5,-.6) -- (-.5,1);
    \draw[thick] (0,1) ellipse (.5 and .2);
    \draw[thick] (.5,-.6) arc(0:-180:.5 and .2);
    \draw[thick, dotted] (.5,-.6) arc(0:180:.5 and .2);
    \draw[thick, violet] (.5,.2) arc (0:-180:.5 and .2);
    \draw[thick, violet, dotted] (.5,.2) arc (0:180:.5 and .2);
    \node[violet] at (.0,-.2) {$\scriptstyle \sigma$};
} - \tikzmath[scale=0.65]{\draw[thick] (.5,-.6) -- (.5,1);
    \draw[thick] (-.5,-.6) -- (-.5,1);
    \draw[thick] (0,1) ellipse (.5 and .2);
    \draw[thick] (.5,-.6) arc(0:-180:.5 and .2);
    \draw[thick, dotted] (.5,-.6) arc(0:180:.5 and .2);
    \draw[thick, violet] (.5,-.1) arc (0:-180:.5 and .2);
    \draw[thick, violet, dotted] (.5,-0.1) arc (0:180:.5 and .2);
    \node[violet] at (.0,-.5) {$\scriptstyle \sigma$};
    \draw[thick, cyan] (.5,0.4) arc (0:-180:.5 and .2);
    \draw[thick, cyan, dotted] (.5,0.4) arc (0:180:.5 and .2);
    \node[cyan] at (.0,0.4) {$\scriptstyle \psi$};
} \right) + \frac{c_3}{2} \left( \tikzmath[scale=0.65]{\draw[thick] (.5,-.6) -- (.5,1);
    \draw[thick] (-.5,-.6) -- (-.5,1);
    \draw[thick] (0,1) ellipse (.5 and .2);
    \draw[thick] (.5,-.6) arc(0:-180:.5 and .2);
    \draw[thick,dotted] (.5,-.6) arc(0:180:.5 and .2);
    \draw[thick, cyan] (0,-0.) -- (0,-0.8);
    \draw[thick, violet] (.5,.2) arc (0:-180:.5 and .2);
    \draw[thick, violet, dotted] (.5,.2) arc (0:180:.5 and .2);
    %\node[violet] at (0.2,0) {$\scriptstyle \sigma$};
} - \tikzmath[scale=0.65]{\draw[thick] (.5,-.6) -- (.5,1);
    \draw[thick] (-.5,-.6) -- (-.5,1);
    \draw[thick] (0,1) ellipse (.5 and .2);
    \draw[thick] (.5,-.6) arc(0:-180:.5 and .2);
    \draw[thick,dotted] (.5,-.6) arc(0:180:.5 and .2);
    \draw[thick, cyan] (0,-0.3) -- (0,-0.8);
    \draw[thick, violet] (.5,-.1) arc (0:-180:.5 and .2);
    \draw[thick, violet, dotted] (.5,-.1) arc (0:180:.5 and .2);
    \draw[thick, cyan] (.5,0.4) arc (0:-180:.5 and .2);
    \draw[thick, cyan, dotted] (.5,0.4) arc (0:180:.5 and .2);
    %\node[violet] at (0.2,0) {$\scriptstyle \sigma$};
} \right) = c_3 \;
\tikzmath[scale=0.65]{\draw[thick] (.5,-.6) -- (.5,1);
    \draw[thick] (-.5,-.6) -- (-.5,1);
    \draw[thick] (0,1) ellipse (.5 and .2);
    \draw[thick] (.5,-.6) arc(0:-180:.5 and .2);
    \draw[thick,dotted] (.5,-.6) arc(0:180:.5 and .2);
    \draw[thick, cyan] (0,-0.) -- (0,-0.8);
    \draw[thick, violet] (.5,.2) arc (0:-180:.5 and .2);
    \draw[thick, violet, dotted] (.5,.2) arc (0:180:.5 and .2);
}
\end{align*}
In the last step, we used that the second tube in the $c_3$ term simplifies to the first one, up to a minus sign coming from the F-symbol $[F^{\sigma \psi \sigma}_\psi] = -1$. 
\cite[Def.~3.1]{MR1659954}.

The resulting tube has zero overlap with the anyon idempotents $P_1$, $P_e$ and $P_f$, but nonzero overlap with $P_m$: 
\begin{align*}
&P_m \; \tikzmath[scale=0.7]{\draw[thick] (.5,-.6) -- (.5,1);
    \draw[thick] (-.5,-.6) -- (-.5,1);
    \draw[thick] (0,1) ellipse (.5 and .2);
    \draw[thick] (.5,-.6) arc(0:-180:.5 and .2);
    \draw[thick,dotted] (.5,-.6) arc(0:180:.5 and .2);
    \draw[thick, cyan] (0,-0.) -- (0,-0.8);
    \draw[thick, violet] (.5,.2) arc (0:-180:.5 and .2);
    \draw[thick, violet, dotted] (.5,.2) arc (0:180:.5 and .2);
}  = \frac{1}{2}  \left(
    \tikzmath[scale=0.7]{\draw[thick] (.5,-.6) -- (.5,1);
    \draw[thick] (-.5,-.6) -- (-.5,1);
    \draw[thick] (0,1) ellipse (.5 and .2);
    \draw[thick] (.5,-.6) arc(0:-180:.5 and .2);
    \draw[thick,dotted] (.5,-.6) arc(0:180:.5 and .2);
    \draw[thick, cyan] (0,-0.8) -- (0,0.8);
    \node[cyan, right] at (-0.1,0) {$\scriptstyle \psi$};
    } + 
    \tikzmath[scale=0.7]{\draw[thick] (.5,-.6) -- (.5,1);
    \draw[thick] (-.5,-.6) -- (-.5,1);
    \draw[thick] (0,1) ellipse (.5 and .2);
    \draw[thick] (.5,-.6) arc(0:-180:.5 and .2);
    \draw[thick,dotted] (.5,-.6) arc(0:180:.5 and .2);
    \draw[thick, cyan] (.5,.) arc (0:-90:.5 and .2);
    \draw[thick, cyan, dotted] (.5,0.) arc (0:180:.5 and .2);
    \draw[thick, cyan] (0,-0.2) -- (0,-0.8);
    \draw[thick, cyan] (-.5,0.) arc (180:90:.5 and .4);
    \draw[thick, cyan] (0,0.4) -- (0,0.8);
}\right)\, \tikzmath[scale=0.7]{\draw[thick] (.5,-.6) -- (.5,1);
    \draw[thick] (-.5,-.6) -- (-.5,1);
    \draw[thick] (0,1) ellipse (.5 and .2);
    \draw[thick] (.5,-.6) arc(0:-180:.5 and .2);
    \draw[thick,dotted] (.5,-.6) arc(0:180:.5 and .2);
    \draw[thick, cyan] (0,-0.) -- (0,-0.8);
    \draw[thick, violet] (.5,.2) arc (0:-180:.5 and .2);
    \draw[thick, violet, dotted] (.5,.2) arc (0:180:.5 and .2);
} =  \tikzmath[scale=0.7]{\draw[thick] (.5,-.6) -- (.5,1);
    \draw[thick] (-.5,-.6) -- (-.5,1);
    \draw[thick] (0,1) ellipse (.5 and .2);
    \draw[thick] (.5,-.6) arc(0:-180:.5 and .2);
    \draw[thick,dotted] (.5,-.6) arc(0:180:.5 and .2);
    \draw[thick, cyan] (0,-0.) -- (0,-0.8);
    \draw[thick, violet] (.5,.2) arc (0:-180:.5 and .2);
    \draw[thick, violet, dotted] (.5,.2) arc (0:180:.5 and .2);
} \ne 0, \quad P_{1,e,f} \;\tikzmath[scale=0.7]{\draw[thick] (.5,-.6) -- (.5,1);
    \draw[thick] (-.5,-.6) -- (-.5,1);
    \draw[thick] (0,1) ellipse (.5 and .2);
    \draw[thick] (.5,-.6) arc(0:-180:.5 and .2);
    \draw[thick,dotted] (.5,-.6) arc(0:180:.5 and .2);
    \draw[thick, cyan] (0,-0.) -- (0,-0.8);
    \draw[thick, violet] (.5,.2) arc (0:-180:.5 and .2);
    \draw[thick, violet, dotted] (.5,.2) arc (0:180:.5 and .2);
}  = 0.
\end{align*}
Therefore, we recover the known result that $e$ is mapped to $m$ by the symmetry action of the non-identity element of $\mathbb{Z}_2$. It can also be checked that $1$ and $f$ are invariant under the symmetry.

\subsection{Non-invertible example with \texorpdfstring{$\cY=S_3$}{Y=S3}}
\label{subsec:beyondG}

The method of Section~\ref{subsec:domainwallwrap} does not require $\cY$ to be a
$G$-graded extension of $\cX$. In the more general setting of $\cX \subset \cY$, wrapping a domain wall can
map an anyon to a direct sum of anyons and twist defects, as we will see in this example. 

As input categories, we take $\cX=\Hilb(\mathbb{Z}_2)$ as before, together with $\cY = \Hilb(S_3)$ with trivial F-symbols. We denote the group elements in $S_3$ as follows:
$$
S_3=\langle r, s | r^3=s^2=1, \; sr = r^2s \rangle.
$$ 
Elements in $\Tube(\cY)$ can be written as
\[
     \tikzmath[scale=0.8]{
    \draw[thick] (.5,-.6) -- (.5,1);
    \draw[thick] (-.5,-.6) -- (-.5,1);
    \draw[thick] (0,1) ellipse (.5 and .2);
    \draw[thick] (.5,-.6) arc(0:-180:.5 and .2);
    \draw[thick, dotted] (.5,-.6) arc(0:180:.5 and .2);
    \draw[thick, DarkGreen,
  postaction={decorate},decoration={markings, mark=at position 0.7 with {\arrow{>}}  }] (-.5,.2) arc (-180:-90:.5 and .2);
  \draw[thick, DarkGreen,
  postaction={decorate},decoration={markings, mark=at position 0.7 with {\arrow{>}}  }] (0.,0.) arc (-90:0:.5 and .2);
    \node[DarkGreen, left] at (-0.5,.2) {$\scriptstyle g$};
    \draw[thick, violet, postaction={decorate},decoration={markings, mark=at position 0.8 with {\arrow{>}}}] (0,0) -- (0,0.8);
    \node[violet, right] at (0.5,0.5) {$\scriptstyle g^{-1} h g$};
    \draw[thick, gray,postaction={decorate},decoration={markings, mark=at position 0.4 with {\arrow{>}}}] (0,-0.8) -- (0,-0.);
    \node[gray, right] at (0.5,-0.3) {$\scriptstyle h$};
    \fill[DarkGreen] (0,0) circle (0.8mm);
    \draw[thick, DarkGreen, dotted] (.5,.2) arc (0:180:.5 and .2);
} \in \Tube(\cY) \quad \text{with } g,h \in S_3.
\]

In this case, we may decompose $\cY\cong\cX\oplus\cY_1$, where $\Irr(\cY_1)=S_3\setminus\mathbb{Z}_2=\{r,r^2,rs,r^2s\}$.
The action of the nontrivial symmetry sector is therefore obtained by stacking with
a linear combination of tubes with $r$, $r^2$, $sr$, or $sr^2$ lines wrapping around:
\[
\sum_i c_i T_i = c_1 \,  \tikzmath[scale=0.72]{\draw[thick] (.5,-.6) -- (.5,1);
    \draw[thick] (-.5,-.6) -- (-.5,1);
    \draw[thick] (0,1) ellipse (.5 and .2);
    \draw[thick] (.5,-.6) arc(0:-180:.5 and .2);
    \draw[thick, dotted] (.5,-.6) arc(0:180:.5 and .2);
    \draw[thick, magenta,
  postaction={decorate},decoration={markings, mark=at position 0.7 with {\arrow{>}}  }] (-.5,.2) arc (-180:0:.5 and .2);
    \draw[thick, magenta, dotted] (.5,.2) arc (0:180:.5 and .2);
    \node[magenta] at (.0,-.25) {$\scriptstyle r$};
} \, +  c_2 \,  \tikzmath[scale=0.72]{\draw[thick] (.5,-.6) -- (.5,1);
    \draw[thick] (-.5,-.6) -- (-.5,1);
    \draw[thick] (0,1) ellipse (.5 and .2);
    \draw[thick] (.5,-.6) arc(0:-180:.5 and .2);
    \draw[thick, dotted] (.5,-.6) arc(0:180:.5 and .2);
    \draw[thick, orange,
  postaction={decorate},decoration={markings, mark=at position 0.7 with {\arrow{>}}  }] (-.5,.2) arc (-180:0:.5 and .2);
    \draw[thick, orange, dotted] (.5,.2) arc (0:180:.5 and .2);
    \node[orange] at (.0,-.25) {$\scriptstyle r^2$};} 
\, + c_3 \,  \tikzmath[scale=0.72]{
    \draw[thick] (.5,-.6) -- (.5,1);
    \draw[thick] (-.5,-.6) -- (-.5,1);
    \draw[thick] (0,1) ellipse (.5 and .2);
    \draw[thick] (.5,-.6) arc(0:-180:.5 and .2);
    \draw[thick, dotted] (.5,-.6) arc(0:180:.5 and .2);
    \draw[thick, magenta,
  postaction={decorate},decoration={markings, mark=at position 0.7 with {\arrow{>}}  }] (-.5,.2) arc (-180:-90:.5 and .2);
  \draw[thick, magenta,
  postaction={decorate},decoration={markings, mark=at position 0.7 with {\arrow{>}}  }] (0.,0.) arc (-90:0:.5 and .2);
    \node[magenta, right] at (.5,.1) {$\scriptstyle r$};
    \draw[thick, cyan, postaction={decorate},decoration={markings, mark=at position 0.8 with {\arrow{>}}}] (0,0) -- (0,0.8);
    \node[cyan, right] at (0.5,0.6) {$\scriptstyle s$};
    \draw[thick, teal,postaction={decorate},decoration={markings, mark=at position 0.4 with {\arrow{>}}}] (0,-0.8) -- (0,-0.);
    \node[teal, right] at (0.5,-0.4) {$\scriptstyle s r$};
    \fill[magenta] (0,0) circle (0.8mm);
    \draw[thick, magenta, dotted] (.5,.2) arc (0:180:.5 and .2);
} \hspace{-1mm}  + c_4 \, \tikzmath[scale=0.72]{
    \draw[thick] (.5,-.6) -- (.5,1);
    \draw[thick] (-.5,-.6) -- (-.5,1);
    \draw[thick] (0,1) ellipse (.5 and .2);
    \draw[thick] (.5,-.6) arc(0:-180:.5 and .2);
    \draw[thick, dotted] (.5,-.6) arc(0:180:.5 and .2);
    \draw[thick, orange,
  postaction={decorate},decoration={markings, mark=at position 0.7 with {\arrow{>}}  }] (-.5,.2) arc (-180:-90:.5 and .2);
  \draw[thick, orange,
  postaction={decorate},decoration={markings, mark=at position 0.7 with {\arrow{>}}  }] (0.,0.) arc (-90:0:.5 and .2);
    \node[orange, right] at (.5,.1) {$\scriptstyle r^2$};
    \draw[thick, cyan, postaction={decorate},decoration={markings, mark=at position 0.8 with {\arrow{>}}}] (0,0) -- (0,0.8);
    \node[cyan, right] at (0.5,0.6) {$\scriptstyle s$};
    \draw[thick, olive,postaction={decorate},decoration={markings, mark=at position 0.4 with {\arrow{>}}}] (0,-0.8) -- (0,-0.);
    \node[olive, right] at (0.5,-0.4) {$\scriptstyle s r^2$};
    \fill[orange] (0,0) circle (0.8mm);
    \draw[thick, orange, dotted] (.5,.2) arc (0:180:.5 and .2);
} \hspace{-2mm}
 + c_5 \,  \tikzmath[scale=0.72]{\draw[thick] (.5,-.6) -- (.5,1);
    \draw[thick] (-.5,-.6) -- (-.5,1);
    \draw[thick] (0,1) ellipse (.5 and .2);
    \draw[thick] (.5,-.6) arc(0:-180:.5 and .2);
    \draw[thick, dotted] (.5,-.6) arc(0:180:.5 and .2);
    \draw[thick, teal,
  postaction={decorate},decoration={markings, mark=at position 0.7 with {\arrow{>}}  }] (-.5,.2) arc (-180:0:.5 and .2);
    \draw[thick, teal, dotted] (.5,.2) arc (0:180:.5 and .2);
    \node[teal] at (.0,-.25) {$\scriptstyle sr$};
} \, +  c_6 \,  \tikzmath[scale=0.72]{\draw[thick] (.5,-.6) -- (.5,1);
    \draw[thick] (-.5,-.6) -- (-.5,1);
    \draw[thick] (0,1) ellipse (.5 and .2);
    \draw[thick] (.5,-.6) arc(0:-180:.5 and .2);
    \draw[thick, dotted] (.5,-.6) arc(0:180:.5 and .2);
    \draw[thick, olive,
  postaction={decorate},decoration={markings, mark=at position 0.7 with {\arrow{>}}  }] (-.5,.2) arc (-180:0:.5 and .2);
    \draw[thick, olive, dotted] (.5,.2) arc (0:180:.5 and .2);
    \node[olive] at (.0,-.25) {$\scriptstyle s r^2$};} 
\, + c_7 \,  \tikzmath[scale=0.72]{
    \draw[thick] (.5,-.6) -- (.5,1);
    \draw[thick] (-.5,-.6) -- (-.5,1);
    \draw[thick] (0,1) ellipse (.5 and .2);
    \draw[thick] (.5,-.6) arc(0:-180:.5 and .2);
    \draw[thick, dotted] (.5,-.6) arc(0:180:.5 and .2);
    \draw[thick, teal,
  postaction={decorate},decoration={markings, mark=at position 0.7 with {\arrow{>}}  }] (-.5,.2) arc (-180:-90:.5 and .2);
  \draw[thick, teal,
  postaction={decorate},decoration={markings, mark=at position 0.7 with {\arrow{>}}  }] (0.,0.) arc (-90:0:.5 and .2);
    \node[teal, right] at (.5,.1) {$\scriptstyle sr$};
    \draw[thick, cyan, postaction={decorate},decoration={markings, mark=at position 0.8 with {\arrow{>}}}] (0,0) -- (0,0.8);
    \node[cyan, right] at (0.5,0.6) {$\scriptstyle s$};
    \draw[thick, olive,postaction={decorate},decoration={markings, mark=at position 0.4 with {\arrow{>}}}] (0,-0.8) -- (0,-0.);
    \node[olive, right] at (0.5,-0.4) {$\scriptstyle s r^2$};
    \fill[teal] (0,0) circle (0.8mm);
    \draw[thick, teal, dotted] (.5,.2) arc (0:180:.5 and .2);
} \hspace{-2mm} + c_8 \, \tikzmath[scale=0.72]{
    \draw[thick] (.5,-.6) -- (.5,1);
    \draw[thick] (-.5,-.6) -- (-.5,1);
    \draw[thick] (0,1) ellipse (.5 and .2);
    \draw[thick] (.5,-.6) arc(0:-180:.5 and .2);
    \draw[thick, dotted] (.5,-.6) arc(0:180:.5 and .2);
    \draw[thick, olive,
  postaction={decorate},decoration={markings, mark=at position 0.7 with {\arrow{>}}  }] (-.5,.2) arc (-180:-90:.5 and .2);
  \draw[thick, olive,
  postaction={decorate},decoration={markings, mark=at position 0.7 with {\arrow{>}}  }] (0.,0.) arc (-90:0:.5 and .2);
    \node[olive, right] at (.5,.1) {$\scriptstyle s r^2$};
    \draw[thick, cyan, postaction={decorate},decoration={markings, mark=at position 0.8 with {\arrow{>}}}] (0,0) -- (0,0.8);
    \node[cyan, right] at (0.5,0.6) {$\scriptstyle s$};
    \draw[thick, teal,postaction={decorate},decoration={markings, mark=at position 0.4 with {\arrow{>}}}] (0,-0.8) -- (0,-0.);
    \node[teal, right] at (0.5,-0.4) {$\scriptstyle s r$};
    \fill[olive] (0,0) circle (0.8mm);
    \draw[thick, olive, dotted] (.5,.2) arc (0:180:.5 and .2);
}
\]
Here we only include tubes in the linear combination which have $1$ or $s$ as the top vertical strand so that stacking them with toric code anyons gives a nonzero result. 

Stacking these domain wall tubes beneath $P_1$ and $P_e$, we find 
\begin{align*}
     \left(\sum_i c_i T_i\right) P_a &= \frac{c_1}{2} \;  \tikzmath[scale=0.8]{\draw[thick] (.5,-.6) -- (.5,1);
    \draw[thick] (-.5,-.6) -- (-.5,1);
    \draw[thick] (0,1) ellipse (.5 and .2);
    \draw[thick] (.5,-.6) arc(0:-180:.5 and .2);
    \draw[thick, dotted] (.5,-.6) arc(0:180:.5 and .2);
    \draw[thick, magenta,
  postaction={decorate},decoration={markings, mark=at position 0.7 with {\arrow{>}}  }] (-.5,.2) arc (-180:0:.5 and .2);
    \draw[thick, magenta, dotted] (.5,.2) arc (0:180:.5 and .2);
    \node[magenta] at (.0,-.25) {$\scriptstyle r$};
} +  \frac{c_2}{2} \;  \tikzmath[scale=0.8]{\draw[thick] (.5,-.6) -- (.5,1);
    \draw[thick] (-.5,-.6) -- (-.5,1);
    \draw[thick] (0,1) ellipse (.5 and .2);
    \draw[thick] (.5,-.6) arc(0:-180:.5 and .2);
    \draw[thick, dotted] (.5,-.6) arc(0:180:.5 and .2);
    \draw[thick, orange,
  postaction={decorate},decoration={markings, mark=at position 0.7 with {\arrow{>}}  }] (-.5,.2) arc (-180:0:.5 and .2);
    \draw[thick, orange, dotted] (.5,.2) arc (0:180:.5 and .2);
    \node[orange] at (.0,-.25) {$\scriptstyle r^2$};} +
     \frac{c_5}{2} \;  \tikzmath[scale=0.8]{\draw[thick] (.5,-.6) -- (.5,1);
    \draw[thick] (-.5,-.6) -- (-.5,1);
    \draw[thick] (0,1) ellipse (.5 and .2);
    \draw[thick] (.5,-.6) arc(0:-180:.5 and .2);
    \draw[thick, dotted] (.5,-.6) arc(0:180:.5 and .2);
    \draw[thick, teal,
  postaction={decorate},decoration={markings, mark=at position 0.7 with {\arrow{>}}  }] (-.5,.2) arc (-180:0:.5 and .2);
    \draw[thick, teal, dotted] (.5,.2) arc (0:180:.5 and .2);
    \node[teal] at (.0,-.25) {$\scriptstyle sr$};
} +  \frac{c_6}{2} \;  \tikzmath[scale=0.8]{\draw[thick] (.5,-.6) -- (.5,1);
    \draw[thick] (-.5,-.6) -- (-.5,1);
    \draw[thick] (0,1) ellipse (.5 and .2);
    \draw[thick] (.5,-.6) arc(0:-180:.5 and .2);
    \draw[thick, dotted] (.5,-.6) arc(0:180:.5 and .2);
    \draw[thick, olive,
  postaction={decorate},decoration={markings, mark=at position 0.7 with {\arrow{>}}  }] (-.5,.2) arc (-180:0:.5 and .2);
    \draw[thick, olive, dotted] (.5,.2) arc (0:180:.5 and .2);
    \node[olive] at (.0,-.25) {$\scriptstyle s r^2$};} \\
    &+ \epsilon_a \left( \frac{c_1}{2} \;  \tikzmath[scale=0.8]{\draw[thick] (.5,-.6) -- (.5,1);
    \draw[thick] (-.5,-.6) -- (-.5,1);
    \draw[thick] (0,1) ellipse (.5 and .2);
    \draw[thick] (.5,-.6) arc(0:-180:.5 and .2);
    \draw[thick, dotted] (.5,-.6) arc(0:180:.5 and .2);
    \draw[thick, olive,
  postaction={decorate},decoration={markings, mark=at position 0.7 with {\arrow{>}}  }] (-.5,.2) arc (-180:0:.5 and .2);
    \draw[thick, olive, dotted] (.5,.2) arc (0:180:.5 and .2);
    \node[olive] at (.0,-.3) {$\scriptstyle r s$};
} +  \frac{c_2}{2} \;  \tikzmath[scale=0.8]{\draw[thick] (.5,-.6) -- (.5,1);
    \draw[thick] (-.5,-.6) -- (-.5,1);
    \draw[thick] (0,1) ellipse (.5 and .2);
    \draw[thick] (.5,-.6) arc(0:-180:.5 and .2);
    \draw[thick, dotted] (.5,-.6) arc(0:180:.5 and .2);
    \draw[thick, teal,
  postaction={decorate},decoration={markings, mark=at position 0.7 with {\arrow{>}}  }] (-.5,.2) arc (-180:0:.5 and .2);
    \draw[thick, teal, dotted] (.5,.2) arc (0:180:.5 and .2);
    \node[teal] at (.0,-.3) {$\scriptstyle r^2 s$};} +
    \frac{c_5}{2} \;  \tikzmath[scale=0.8]{\draw[thick] (.5,-.6) -- (.5,1);
    \draw[thick] (-.5,-.6) -- (-.5,1);
    \draw[thick] (0,1) ellipse (.5 and .2);
    \draw[thick] (.5,-.6) arc(0:-180:.5 and .2);
    \draw[thick, dotted] (.5,-.6) arc(0:180:.5 and .2);
    \draw[thick, orange,
  postaction={decorate},decoration={markings, mark=at position 0.7 with {\arrow{>}}  }] (-.5,.2) arc (-180:0:.5 and .2);
    \draw[thick, orange, dotted] (.5,.2) arc (0:180:.5 and .2);
    \node[orange] at (.0,-.3) {$\scriptstyle r^2$};
} +  \frac{c_6}{2} \;  \tikzmath[scale=0.8]{\draw[thick] (.5,-.6) -- (.5,1);
    \draw[thick] (-.5,-.6) -- (-.5,1);
    \draw[thick] (0,1) ellipse (.5 and .2);
    \draw[thick] (.5,-.6) arc(0:-180:.5 and .2);
    \draw[thick, dotted] (.5,-.6) arc(0:180:.5 and .2);
    \draw[thick, magenta,
  postaction={decorate},decoration={markings, mark=at position 0.7 with {\arrow{>}}  }] (-.5,.2) arc (-180:0:.5 and .2);
    \draw[thick, magenta, dotted] (.5,.2) arc (0:180:.5 and .2);
    \node[magenta] at (.0,-.3) {$\scriptstyle r$};}
    \right), \quad a\in\{1,e\},
\end{align*}
where $\epsilon_1=+1$ and $\epsilon_e=-1$.
Similarly, stacking them beneath $P_m$ and $P_f$ gives 
\begin{align*} 
\left(\sum_i c_i T_i\right) P_{a} =& 
\frac{c_3}{2} \;  \tikzmath[scale=0.75]{
    \draw[thick] (.5,-.6) -- (.5,1);
    \draw[thick] (-.5,-.6) -- (-.5,1);
    \draw[thick] (0,1) ellipse (.5 and .2);
    \draw[thick] (.5,-.6) arc(0:-180:.5 and .2);
    \draw[thick, dotted] (.5,-.6) arc(0:180:.5 and .2);
    \draw[thick, magenta,
  postaction={decorate},decoration={markings, mark=at position 0.7 with {\arrow{>}}  }] (-.5,.2) arc (-180:-90:.5 and .2);
  \draw[thick, magenta,
  postaction={decorate},decoration={markings, mark=at position 0.7 with {\arrow{>}}  }] (0.,0.) arc (-90:0:.5 and .2);
    \node[right, magenta] at (.5,.1) {$\scriptstyle r$};
    \draw[thick, cyan, postaction={decorate},decoration={markings, mark=at position 0.8 with {\arrow{>}}}] (0,0) -- (0,0.8);
    \node[cyan, right] at (0.5,0.6) {$\scriptstyle s$};
    \draw[thick, teal,postaction={decorate},decoration={markings, mark=at position 0.4 with {\arrow{>}}}] (0,-0.8) -- (0,-0.);
    \node[teal, right] at (0.5,-0.4) {$\scriptstyle s r$};
    \fill[magenta] (0,0) circle (0.8mm);
    \draw[thick, magenta, dotted] (.5,.2) arc (0:180:.5 and .2);
} \hspace{-1mm}  + \,\frac{c_4}{2}\; \tikzmath[scale=0.75]{
    \draw[thick] (.5,-.6) -- (.5,1);
    \draw[thick] (-.5,-.6) -- (-.5,1);
    \draw[thick] (0,1) ellipse (.5 and .2);
    \draw[thick] (.5,-.6) arc(0:-180:.5 and .2);
    \draw[thick, dotted] (.5,-.6) arc(0:180:.5 and .2);
    \draw[thick, orange,
  postaction={decorate},decoration={markings, mark=at position 0.7 with {\arrow{>}}  }] (-.5,.2) arc (-180:-90:.5 and .2);
  \draw[thick, orange,
  postaction={decorate},decoration={markings, mark=at position 0.7 with {\arrow{>}}  }] (0.,0.) arc (-90:0:.5 and .2);
    \node[orange, right] at (.5,.1) {$\scriptstyle r^2$};
    \draw[thick, cyan, postaction={decorate},decoration={markings, mark=at position 0.8 with {\arrow{>}}}] (0,0) -- (0,0.8);
    \node[cyan, right] at (0.5,0.6) {$\scriptstyle s$};
    \draw[thick, olive,postaction={decorate},decoration={markings, mark=at position 0.4 with {\arrow{>}}}] (0,-0.8) -- (0,-0.);
    \node[olive, right] at (0.5,-0.4) {$\scriptstyle s r^2$};
    \fill[orange] (0,0) circle (0.8mm);
    \draw[thick, orange, dotted] (.5,.2) arc (0:180:.5 and .2);
} \hspace{-1mm} + \frac{c_7}{2} \;  \tikzmath[scale=0.75]{
    \draw[thick] (.5,-.6) -- (.5,1);
    \draw[thick] (-.5,-.6) -- (-.5,1);
    \draw[thick] (0,1) ellipse (.5 and .2);
    \draw[thick] (.5,-.6) arc(0:-180:.5 and .2);
    \draw[thick, dotted] (.5,-.6) arc(0:180:.5 and .2);
    \draw[thick, teal,
  postaction={decorate},decoration={markings, mark=at position 0.7 with {\arrow{>}}  }] (-.5,.2) arc (-180:-90:.5 and .2);
  \draw[thick, teal,
  postaction={decorate},decoration={markings, mark=at position 0.7 with {\arrow{>}}  }] (0.,0.) arc (-90:0:.5 and .2);
    \node[right,teal] at (.5,.1) {$\scriptstyle sr$};
    \draw[thick, cyan, postaction={decorate},decoration={markings, mark=at position 0.8 with {\arrow{>}}}] (0,0) -- (0,0.8);
    \node[cyan, right] at (0.5,0.6) {$\scriptstyle s$};
    \draw[thick, olive,postaction={decorate},decoration={markings, mark=at position 0.4 with {\arrow{>}}}] (0,-0.8) -- (0,-0.);
    \node[olive, right] at (0.5,-0.4) {$\scriptstyle s r^2$};
    \fill[teal] (0,0) circle (0.8mm);
    \draw[thick, teal, dotted] (.5,.2) arc (0:180:.5 and .2);
} \hspace{-1mm}  + \,\frac{c_8}{2}\; \tikzmath[scale=0.75]{
    \draw[thick] (.5,-.6) -- (.5,1);
    \draw[thick] (-.5,-.6) -- (-.5,1);
    \draw[thick] (0,1) ellipse (.5 and .2);
    \draw[thick] (.5,-.6) arc(0:-180:.5 and .2);
    \draw[thick, dotted] (.5,-.6) arc(0:180:.5 and .2);
    \draw[thick, olive,
  postaction={decorate},decoration={markings, mark=at position 0.7 with {\arrow{>}}  }] (-.5,.2) arc (-180:-90:.5 and .2);
  \draw[thick, olive,
  postaction={decorate},decoration={markings, mark=at position 0.7 with {\arrow{>}}  }] (0.,0.) arc (-90:0:.5 and .2);
    \node[olive, right] at (.5,.1) {$\scriptstyle s r^2$};
    \draw[thick, cyan, postaction={decorate},decoration={markings, mark=at position 0.8 with {\arrow{>}}}] (0,0) -- (0,0.8);
    \node[cyan, right] at (0.5,0.6) {$\scriptstyle s$};
    \draw[thick, teal,postaction={decorate},decoration={markings, mark=at position 0.4 with {\arrow{>}}}] (0,-0.8) -- (0,-0.);
    \node[teal, right] at (0.5,-0.4) {$\scriptstyle s r$};
    \fill[olive] (0,0) circle (0.8mm);
    \draw[thick, olive, dotted] (.5,.2) arc (0:180:.5 and .2);
}
\\&+ \epsilon_a \left( \frac{c_3}{2} \;  \tikzmath[scale=0.75]{
    \draw[thick] (.5,-.6) -- (.5,1);
    \draw[thick] (-.5,-.6) -- (-.5,1);
    \draw[thick] (0,1) ellipse (.5 and .2);
    \draw[thick] (.5,-.6) arc(0:-180:.5 and .2);
    \draw[thick, dotted] (.5,-.6) arc(0:180:.5 and .2);
    \draw[thick, olive,
  postaction={decorate},decoration={markings, mark=at position 0.7 with {\arrow{>}}  }] (-.5,.2) arc (-180:-90:.5 and .2);
  \draw[thick, olive,
  postaction={decorate},decoration={markings, mark=at position 0.7 with {\arrow{>}}  }] (0.,0.) arc (-90:0:.5 and .2);
    \node[right, olive] at (.5,0.1) {$\scriptstyle rs$};
    \draw[thick, cyan, postaction={decorate},decoration={markings, mark=at position 0.8 with {\arrow{>}}}] (0,0) -- (0,0.8);
    \node[cyan, right] at (0.5,0.6) {$\scriptstyle s$};
    \draw[thick, teal,postaction={decorate},decoration={markings, mark=at position 0.4 with {\arrow{>}}}] (0,-0.8) -- (0,-0.);
    \node[teal, right] at (0.5,-0.4) {$\scriptstyle sr$};
    \fill[olive] (0,0) circle (0.8mm);
    \draw[thick, olive, dotted] (.5,.2) arc (0:180:.5 and .2);
} \hspace{-1mm} + \frac{c_4}{2} \; \tikzmath[scale=0.75]{
    \draw[thick] (.5,-.6) -- (.5,1);
    \draw[thick] (-.5,-.6) -- (-.5,1);
    \draw[thick] (0,1) ellipse (.5 and .2);
    \draw[thick] (.5,-.6) arc(0:-180:.5 and .2);
    \draw[thick, dotted] (.5,-.6) arc(0:180:.5 and .2);
    \draw[thick, teal,
  postaction={decorate},decoration={markings, mark=at position 0.7 with {\arrow{>}}  }] (-.5,.2) arc (-180:-90:.5 and .2);
  \draw[thick, teal,
  postaction={decorate},decoration={markings, mark=at position 0.7 with {\arrow{>}}  }] (0.,0.) arc (-90:0:.5 and .2);
    \node[teal, right] at (.5,.1) {$\scriptstyle r^2 s$};
    \draw[thick, cyan, postaction={decorate},decoration={markings, mark=at position 0.8 with {\arrow{>}}}] (0,0) -- (0,0.8);
    \node[cyan, right] at (0.5,0.6) {$\scriptstyle s$};
    \draw[thick, olive,postaction={decorate},decoration={markings, mark=at position 0.4 with {\arrow{>}}}] (0,-0.8) -- (0,-0.);
    \node[olive, right] at (0.5,-0.4) {$\scriptstyle sr^2$};
    \fill[teal] (0,0) circle (0.8mm);
    \draw[thick, teal, dotted] (.5,.2) arc (0:180:.5 and .2);
}
\hspace{-1mm} + \frac{c_7}{2} \;  \tikzmath[scale=0.75]{
    \draw[thick] (.5,-.6) -- (.5,1);
    \draw[thick] (-.5,-.6) -- (-.5,1);
    \draw[thick] (0,1) ellipse (.5 and .2);
    \draw[thick] (.5,-.6) arc(0:-180:.5 and .2);
    \draw[thick, dotted] (.5,-.6) arc(0:180:.5 and .2);
    \draw[thick, orange,
  postaction={decorate},decoration={markings, mark=at position 0.7 with {\arrow{>}}  }] (-.5,.2) arc (-180:-90:.5 and .2);
  \draw[thick, orange,
  postaction={decorate},decoration={markings, mark=at position 0.7 with {\arrow{>}}  }] (0.,0.) arc (-90:0:.5 and .2);
    \node[right, orange] at (.5,0.1) {$\scriptstyle r^2$};
    \draw[thick, cyan, postaction={decorate},decoration={markings, mark=at position 0.8 with {\arrow{>}}}] (0,0) -- (0,0.8);
    \node[cyan, right] at (0.5,0.6) {$\scriptstyle s$};
    \draw[thick, olive,postaction={decorate},decoration={markings, mark=at position 0.4 with {\arrow{>}}}] (0,-0.8) -- (0,-0.);
    \node[olive, right] at (0.5,-0.4) {$\scriptstyle sr^2$};
    \fill[orange] (0,0) circle (0.8mm);
    \draw[thick, orange, dotted] (.5,.2) arc (0:180:.5 and .2);
} \hspace{-1mm} + \frac{c_8}{2} \; \tikzmath[scale=0.75]{
    \draw[thick] (.5,-.6) -- (.5,1);
    \draw[thick] (-.5,-.6) -- (-.5,1);
    \draw[thick] (0,1) ellipse (.5 and .2);
    \draw[thick] (.5,-.6) arc(0:-180:.5 and .2);
    \draw[thick, dotted] (.5,-.6) arc(0:180:.5 and .2);
    \draw[thick, magenta,
  postaction={decorate},decoration={markings, mark=at position 0.7 with {\arrow{>}}  }] (-.5,.2) arc (-180:-90:.5 and .2);
  \draw[thick, magenta,
  postaction={decorate},decoration={markings, mark=at position 0.7 with {\arrow{>}}  }] (0.,0.) arc (-90:0:.5 and .2);
    \node[magenta, right] at (.5,.1) {$\scriptstyle r$};
    \draw[thick, cyan, postaction={decorate},decoration={markings, mark=at position 0.8 with {\arrow{>}}}] (0,0) -- (0,0.8);
    \node[cyan, right] at (0.5,0.6) {$\scriptstyle s$};
    \draw[thick, teal,postaction={decorate},decoration={markings, mark=at position 0.4 with {\arrow{>}}}] (0,-0.8) -- (0,-0.);
    \node[teal, right] at (0.5,-0.4) {$\scriptstyle sr$};
    \fill[magenta] (0,0) circle (0.8mm);
    \draw[thick, magenta, dotted] (.5,.2) arc (0:180:.5 and .2);
} \hspace{-2mm}
\right), \; \, a \in \{m,f\},
\end{align*}
where $\epsilon_m=+1$ and $\epsilon_f=-1$.

The excitation created by stacking the domain wall tubes with $P_1$ or $P_e$ has nonzero overlap with both $P_1$ and $P_e$:
$$
    P_b \left(\sum_i c_i T_i \right) P_a = \frac{1}{2} (c_1 + \epsilon_a c_6 + \epsilon_b c_5 + \epsilon_a \epsilon_b c_2) \left(\tikzmath[scale=0.8]{\draw[thick] (.5,-.6) -- (.5,1);
    \draw[thick] (-.5,-.6) -- (-.5,1);
    \draw[thick] (0,1) ellipse (.5 and .2);
    \draw[thick] (.5,-.6) arc(0:-180:.5 and .2);
    \draw[thick, dotted] (.5,-.6) arc(0:180:.5 and .2);
    \draw[thick, magenta,
  postaction={decorate},decoration={markings, mark=at position 0.7 with {\arrow{>}}  }] (-.5,.2) arc (-180:0:.5 and .2);
    \draw[thick, magenta, dotted] (.5,.2) arc (0:180:.5 and .2);
    \node[magenta] at (.0,-.3) {$\scriptstyle r$};} + \epsilon_a \epsilon_b \; \tikzmath[scale=0.8]{\draw[thick] (.5,-.6) -- (.5,1);
    \draw[thick] (-.5,-.6) -- (-.5,1);
    \draw[thick] (0,1) ellipse (.5 and .2);
    \draw[thick] (.5,-.6) arc(0:-180:.5 and .2);
    \draw[thick, dotted] (.5,-.6) arc(0:180:.5 and .2);
    \draw[thick, orange,
  postaction={decorate},decoration={markings, mark=at position 0.7 with {\arrow{>}}  }] (-.5,.2) arc (-180:0:.5 and .2);
    \draw[thick, orange, dotted] (.5,.2) arc (0:180:.5 and .2);
    \node[orange] at (.0,-.3) {$\scriptstyle r^2$};} + \epsilon_a \; \tikzmath[scale=0.8]{\draw[thick] (.5,-.6) -- (.5,1);
    \draw[thick] (-.5,-.6) -- (-.5,1);
    \draw[thick] (0,1) ellipse (.5 and .2);
    \draw[thick] (.5,-.6) arc(0:-180:.5 and .2);
    \draw[thick, dotted] (.5,-.6) arc(0:180:.5 and .2);
    \draw[thick, olive,
  postaction={decorate},decoration={markings, mark=at position 0.7 with {\arrow{>}}  }] (-.5,.2) arc (-180:0:.5 and .2);
    \draw[thick, olive, dotted] (.5,.2) arc (0:180:.5 and .2);
    \node[olive] at (.0,-.3) {$\scriptstyle  sr^2$};} + \epsilon_b \; \tikzmath[scale=0.8]{\draw[thick] (.5,-.6) -- (.5,1);
    \draw[thick] (-.5,-.6) -- (-.5,1);
    \draw[thick] (0,1) ellipse (.5 and .2);
    \draw[thick] (.5,-.6) arc(0:-180:.5 and .2);
    \draw[thick, dotted] (.5,-.6) arc(0:180:.5 and .2);
    \draw[thick, teal,
  postaction={decorate},decoration={markings, mark=at position 0.7 with {\arrow{>}}  }] (-.5,.2) arc (-180:0:.5 and .2);
    \draw[thick, teal, dotted] (.5,.2) arc (0:180:.5 and .2);
    \node[teal] at (.0,-.3) {$\scriptstyle sr$};} \right),
$$
with $a,b \in \{1,e\}$.
The vector spaces spanned by these stacks of tubes are one-dimensional, and so we conclude 
$$
1 \mapsto 1 \oplus e \quad \text{and} \quad e \mapsto 1 \oplus e. 
$$

The remaining two anyons $m$ and $f$ are mapped to a twist defect under the symmetry action.
Recall that an object in the relative center $Z_\cX(\cY)$ can be described as a pair $(Y,\rho)$ with $Y \in \cY$ and $\rho$ a half-braiding.
In this case, since $\cX\subseteq\cY$ comes from the subgroup inclusion $\mathbb{Z}_2\subseteq S_3$, a simple object is determined by a choice of orbit $[g]$ for the conjugacy of $\mathbb{Z}_2$ on $S_3$ and an irreducible representation $\rho\in\Irr(\Rep(\operatorname{stab}_{\mathbb{Z}_2}(g))$.  
In the present example, the relative center contains the four toric code anyons together with two twist defects $(r \oplus r^2,1)$ and $(rs \oplus r^2s,1)$.
Both twist defects have quantum dimension 2.
Their $\Tube_\cX(\cY)$ idempotents are given by 
$$
    P_{r \oplus r^2} = \tikzmath[scale=0.8]{\draw[thick] (.5,-.6) -- (.5,1);
    \draw[thick] (-.5,-.6) -- (-.5,1);
    \draw[thick] (0,1) ellipse (.5 and .2);
    \draw[thick] (.5,-.6) arc(0:-180:.5 and .2);
    \draw[thick, dotted] (.5,-.6) arc(0:180:.5 and .2);
    \draw[thick, magenta, postaction={decorate},decoration={markings, mark=at position 0.7 with {\arrow{>}}}] (0,-0.8) -- (0,0.8);
    \node[magenta, right] at (0,0) {$\scriptstyle r$};} + \tikzmath[scale=0.8]{\draw[thick] (.5,-.6) -- (.5,1);
    \draw[thick] (-.5,-.6) -- (-.5,1);
    \draw[thick] (0,1) ellipse (.5 and .2);
    \draw[thick] (.5,-.6) arc(0:-180:.5 and .2);
    \draw[thick, dotted] (.5,-.6) arc(0:180:.5 and .2);
    \draw[thick, orange, postaction={decorate},decoration={markings, mark=at position 0.7 with {\arrow{>}}}] (0,-0.8) -- (0,0.8);
    \node[orange, right] at (-0.1,0) {$\scriptstyle r^2$};}, \qquad P_{rs \oplus r^2s} = \tikzmath[scale=0.8]{\draw[thick] (.5,-.6) -- (.5,1);
    \draw[thick] (-.5,-.6) -- (-.5,1);
    \draw[thick] (0,1) ellipse (.5 and .2);
    \draw[thick] (.5,-.6) arc(0:-180:.5 and .2);
    \draw[thick, dotted] (.5,-.6) arc(0:180:.5 and .2);
    \draw[thick, olive, postaction={decorate},decoration={markings, mark=at position 0.7 with {\arrow{>}}}] (0,-0.8) -- (0,0.8);
    \node[olive, right] at (0.5,0) {$\scriptstyle rs$};} + \tikzmath[scale=0.8]{\draw[thick] (.5,-.6) -- (.5,1);
    \draw[thick] (-.5,-.6) -- (-.5,1);
    \draw[thick] (0,1) ellipse (.5 and .2);
    \draw[thick] (.5,-.6) arc(0:-180:.5 and .2);
    \draw[thick, dotted] (.5,-.6) arc(0:180:.5 and .2);
    \draw[thick, teal, postaction={decorate},decoration={markings, mark=at position 0.7 with {\arrow{>}}}] (0,-0.8) -- (0,0.8);
    \node[teal, right] at (0.5,0) {$\scriptstyle r^2 s$};}.
$$
Only the overlap with the second twist defect idempotent is nonzero, and gives 
$$
    P_{rs \oplus r^2s} \left(\sum_i c_i T_i\right) P_a = \left(\sum_i c_i T_i\right) P_a, \qquad a \in \{m,f\}.
$$
Therefore, we conclude 
$$
    m \mapsto (rs \oplus r^2s,1) \quad \text{and} \quad f \mapsto (rs \oplus r^2s,1). 
$$
This non-invertible symmetry action matches the one discussed in e.g. \cite{KNBalasubramanian:2025vum, MR3905557, 10.21468/SciPostPhys.17.3.095}. In the language of \cite{KNBalasubramanian:2025vum}, the symmetry is the condensation surface $S_1 \boxplus S_e$ corresponding to condensing the $1 \oplus e$ 1-form symmetry line.

\begin{rem}
We speculate that this symmetry generator together with the trivial generator that sends each anyon to itself should form the hypergroup of double cosets $H \setminus G / H$ with $H = \mathbb{Z}_2$ and $G=S_3$. The trivial generator corresponds to $H = \{1,s\}$ and the nontrivial generator $S$ to $HrH = \{r,r^2,rs,r^2s\}$. Their multiplication follows from multiplying the double cosets, not from iterating their action on anyons, and is given by $S^2 = 2 \cdot 1 + S$.
Indeed, the $\cX-\cX$ bimodule $\cY=\Hilb[S_3]$ splits into two summands corresponding to the two double cosets, and the monoidal product of $\cY$ implements multiplication of double cosets with the correct multiplicities.
The same example was discussed in \cite[Section 6]{10.21468/SciPostPhys.17.3.095} from the perspective of gauging the $\mathbb{Z}_2$ non-normal subgroup of a system with $S_3$ symmetry. They referred to the residual non-invertible symmetry as a coset symmetry, although, as later clarified by \cite{10.21468/SciPostPhysCore.8.4.070}, the relevant symmetry generators are more accurately described by double cosets. The focus of \cite{10.21468/SciPostPhys.17.3.095} is non-invertible symmetry fractionalization in topological orders, which we do not discuss here. 
\end{rem}

%%%%%%%%%%%%%%%%%%%%%%%%%%%%%%%%%%%%%%%%%%%%%%%%%%%%%%
%%%%%%%%%%%%%%%%%%%%%%%%%%%%%%%%%%%%%%%%%%%%%%%%%%%%%%
%%%%%%%%%%%%%%%%%%%%%%%%%%%%%%%%%%%%%%%%%%%%%%%%%%%%%%
\section{Fusion of twist defects for enriched string nets}
\label{sec:Chocolate}

We now expand the results from \S\ref{sec:Vanilla} above to the \emph{enriched setting}, which yields NI-SETOs extending chiral theories.
The main technique is to use the fact that chiral (2+1)D theories live on boundaries of invertible (3+1)D theories \cite{MR4444089}.

Let $\cX\in \UFC^\cB$ be a $\cB$-\emph{enriched fusion category}, which comes equipped with a braided central functor $\Phi^Z: \cB\to Z(\cX)$.
This central functor is exactly the data needed to attach a (2+1)D Levin-Wen model for $\cX$ to the (3+1)D Walker-Wang model for $\cB$ \cite{MR4640433,2305.14068}.
Given a $\cB$-enriched $\cX-\cX$ bimodule $\cN\in \UFC^\cB(\cX\to \cX)$, we get a (1+1)D topological defect on the (2+1)D boundary.
We can again compute the twist defects using a \emph{relative dome algebra} $\Dome^{\textcolor{red}{\cB}}_\cX(\cN)$ generalizing the relative tube algebra in \S\ref{sec:GeneralizedTwistDefects} above and the dome algebra from \cite{2305.14068}, as
$$
\Mod(\Dome^{\textcolor{red}{\cB}}_\cX(\cN))
\cong
\Hom_{\cX-\cX}^\cB(\cX\to \cN).
$$

We then consider the case that $\cN=\cY \in \UFC^\cB(\cX\to \cX)$ is a fusion category \emph{enriched extension} of $\cX$, i.e., a fusion category $\cY\in \UFC^\cB$ which contains $\cX$ as a compatible full subcategory.
This means that we have a monoidal natural isomorphism
\begin{equation}
\label{eq:EnrichedExtension}
\begin{tikzcd}
\cB 
\arrow[r,"\Phi^Z_\cY"]
\arrow[d,"\Phi^Z_\cX"]
&
Z(\cY)
\arrow[d]
\\
Z(\cX)
\arrow[r, hook]
\arrow[ur,phantom,"\overset{\gamma}{\Rightarrow}"]
&
Z_\cX(\cY),
\end{tikzcd}
\end{equation}
which is exactly the action coherence isomorphism from \cite[Def.~3.2]{MR4528312}.
We then discuss how to fuse defects in the $\cB$-enriched relative center $Z_\cX^\cB(\cY):=\Hom_{\cX-\cX}^\cB(\cX\to \cY)$, which we again think of as a NI-SETO.

%%%%%%%%%%%%%%%%%%%%%%%%%%%%%%%%%%%%%%%%%%%%%%%%%%%%
\subsection{The enriched string net model}

We now review the enriched string net model from \cite{MR4640433,2305.14068}.
We give only a cursory introduction and we refer the reader to \cite{2305.14068} for the explicit details.
We consider a cubic lattice where degrees of freedom lie on the vertices as in \cite[p2]{2305.14068}.
Writing $X=\bigoplus_{x\in\Irr(\cX)} x$,  $B=\bigoplus_{b\in\Irr(\cB)} b$, and $\Phi: \cB\to \cX$ as the composite of the central functor $\Phi^Z: \cB\to Z(\cX)$ and the forgetful functor $Z(\cX)\to \cX$, our model is given as follows.
$$
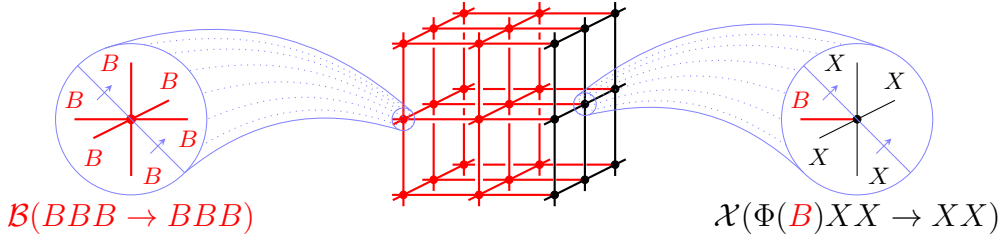
\begin{tikzpicture}
	\pgfmathsetmacro{\lattice}{1};	
	\pgfmathsetmacro{\xoffset}{.4};	
	\pgfmathsetmacro{\yoffset}{.2};	
	\pgfmathsetmacro{\extra}{.15};	
	\pgfmathsetmacro{\zextra}{.35};	
	\pgfmathsetmacro{\zoom}{6};	
	\coordinate (z1) at ($ -.6*(\zoom,0) + (0,\lattice) $);
	\coordinate (z2) at ($ (\zoom,\lattice) $);
	\coordinate (aaa) at (0,0);
	\coordinate (baa) at ($ (aaa) + (\lattice,0) $);
	\coordinate (caa) at ($ (aaa) + 2*(\lattice,0) $);
	\coordinate (aba) at ($ (aaa) + (0,\lattice) $);
	\coordinate (bba) at ($ (aaa) + (\lattice,0) + (0,\lattice) $);
	\coordinate (cba) at ($ (aaa) + 2*(\lattice,0) + (0,\lattice) $);
	\coordinate (aca) at ($ (aaa) + 2*(0,\lattice) $);
	\coordinate (bca) at ($ (aaa) + (\lattice,0) + 2*(0,\lattice) $);
	\coordinate (cca) at ($ (aaa) + 2*(\lattice,0) + 2*(0,\lattice) $);
	\coordinate (aab) at ($ (aaa) + (\xoffset,\yoffset) $);
	\coordinate (bab) at ($ (aaa) + (\lattice,0) + (\xoffset,\yoffset) $);
	\coordinate (cab) at ($ (aaa) + 2*(\lattice,0) + (\xoffset,\yoffset) $);
	\coordinate (abb) at ($ (aaa) + (0,\lattice) + (\xoffset,\yoffset) $);
	\coordinate (bbb) at ($ (aaa) + (\lattice,0) + (0,\lattice) + (\xoffset,\yoffset) $);
	\coordinate (cbb) at ($ (aaa) + 2*(\lattice,0) + (0,\lattice) + (\xoffset,\yoffset) $);
	\coordinate (acb) at ($ (aaa) + 2*(0,\lattice) + (\xoffset,\yoffset) $);
	\coordinate (bcb) at ($ (aaa) + (\lattice,0) + 2*(0,\lattice) + (\xoffset,\yoffset) $);
	\coordinate (ccb) at ($ (aaa) + 2*(\lattice,0) + 2*(0,\lattice) + (\xoffset,\yoffset) $);
	\coordinate (aac) at ($ (aaa) + 2*(\xoffset,\yoffset) $);
	\coordinate (bac) at ($ (aaa) + (\lattice,0) + 2*(\xoffset,\yoffset) $);
	\coordinate (cac) at ($ (aaa) + 2*(\lattice,0) + 2*(\xoffset,\yoffset) $);
	\coordinate (abc) at ($ (aaa) + (0,\lattice) + 2*(\xoffset,\yoffset) $);
	\coordinate (bbc) at ($ (aaa) + (\lattice,0) + (0,\lattice) + 2*(\xoffset,\yoffset) $);
	\coordinate (cbc) at ($ (aaa) + 2*(\lattice,0) + (0,\lattice) + 2*(\xoffset,\yoffset) $);
	\coordinate (acc) at ($ (aaa) + 2*(0,\lattice) + 2*(\xoffset,\yoffset) $);
	\coordinate (bcc) at ($ (aaa) + (\lattice,0) + 2*(0,\lattice) + 2*(\xoffset,\yoffset) $);
	\coordinate (ccc) at ($ (aaa) + 2*(\lattice,0) + 2*(0,\lattice) + 2*(\xoffset,\yoffset) $);
	\draw[thick, red] ($ (aac) - \extra*(\lattice,0) $) -- (cac);
	\draw[thick, red] ($ (abc) - \extra*(\lattice,0) $) -- (cbc);
	\draw[thick, red] ($ (acc) - \extra*(\lattice,0) $) -- (ccc);
	\draw[thick, red] ($ (aac) - \extra*(0,\lattice) $) -- ($ (acc) + \extra*(0,\lattice) $);
	\draw[thick, red] ($ (bac) - \extra*(0,\lattice) $) -- ($ (bcc) + \extra*(0,\lattice) $);
	\draw[thick] ($ (cac) - \extra*(0,\lattice) $) -- ($ (ccc) + \extra*(0,\lattice) $);
	\draw[thick, red, knot] ($ (aab) - \extra*(\lattice,0) $) -- (cab);
	\draw[thick, red, knot] ($ (abb) - \extra*(\lattice,0) $) -- (cbb);
	\draw[thick, red, knot] ($ (acb) - \extra*(\lattice,0) $) -- (ccb);
	\draw[thick, red, knot] ($ (aab) - \extra*(0,\lattice) $) -- ($ (acb) + \extra*(0,\lattice) $);
	\draw[thick, red, knot] ($ (bab) - \extra*(0,\lattice) $) -- ($ (bcb) + \extra*(0,\lattice) $);
	\draw[very thick, white] ($ (cab) - \extra*(0,\lattice) $) -- ($ (ccb) + \extra*(0,\lattice) $);
	\draw[thick] ($ (cab) - \extra*(0,\lattice) $) -- ($ (ccb) + \extra*(0,\lattice) $);
	\draw[thick, red, knot] ($ (aaa) - \extra*(\lattice,0) $) -- (caa);
	\draw[thick, red, knot] ($ (aba) - \extra*(\lattice,0) $) -- (cba);
	\draw[thick, red, knot] ($ (aca) - \extra*(\lattice,0) $) -- (cca);
	\draw[thick, red, knot] ($ (aaa) - \extra*(0,\lattice) $) -- ($ (aca) + \extra*(0,\lattice) $);
	\draw[thick, red, knot] ($ (baa) - \extra*(0,\lattice) $) -- ($ (bca) + \extra*(0,\lattice) $);
	\draw[very thick, white] ($ (caa) - \extra*(0,\lattice) $) -- ($ (cca) + \extra*(0,\lattice) $);
	\draw[thick] ($ (caa) - \extra*(0,\lattice) $) -- ($ (cca) + \extra*(0,\lattice) $);
	\draw[thick, red] ($ (aaa) - \zextra*(\xoffset,\yoffset) $) -- ($ (aac) + \zextra*(\xoffset,\yoffset) $);
	\draw[thick, red] ($ (aba) - \zextra*(\xoffset,\yoffset) $) -- ($ (abc) + \zextra*(\xoffset,\yoffset) $);
	\draw[thick, red] ($ (aca) - \zextra*(\xoffset,\yoffset) $) -- ($ (acc) + \zextra*(\xoffset,\yoffset) $);
	\draw[thick, red] ($ (baa) - \zextra*(\xoffset,\yoffset) $) -- ($ (bac) + \zextra*(\xoffset,\yoffset) $);
	\draw[thick, red] ($ (bba) - \zextra*(\xoffset,\yoffset) $) -- ($ (bbc) + \zextra*(\xoffset,\yoffset) $);
	\draw[thick, red] ($ (bca) - \zextra*(\xoffset,\yoffset) $) -- ($ (bcc) + \zextra*(\xoffset,\yoffset) $);
	\draw[thick] ($ (caa) - \zextra*(\xoffset,\yoffset) $) -- ($ (cac) + \zextra*(\xoffset,\yoffset) $);
	\draw[thick] ($ (cba) - \zextra*(\xoffset,\yoffset) $) -- ($ (cbc) + \zextra*(\xoffset,\yoffset) $);
	\draw[thick] ($ (cca) - \zextra*(\xoffset,\yoffset) $) -- ($ (ccc) + \zextra*(\xoffset,\yoffset) $);
	\filldraw[red] (aaa) circle (.05cm);
	\filldraw[red] (aba) circle (.05cm);
	\filldraw[red] (aca) circle (.05cm);
	\filldraw[red] (aab) circle (.05cm);
	\filldraw[red] (abb) circle (.05cm);
	\filldraw[red] (acb) circle (.05cm);
	\filldraw[red] (aac) circle (.05cm);
	\filldraw[red] (abc) circle (.05cm);
	\filldraw[red] (acc) circle (.05cm);
	\filldraw[red] (baa) circle (.05cm);
	\filldraw[red] (bab) circle (.05cm);
	\filldraw[red] (bac) circle (.05cm);
	\filldraw[red] (bba) circle (.05cm);
	\filldraw[red] (bbb) circle (.05cm);
	\filldraw[red] (bbc) circle (.05cm);
	\filldraw[red] (bca) circle (.05cm);
	\filldraw[red] (bcb) circle (.05cm);
	\filldraw[red] (bcc) circle (.05cm);
	\filldraw (caa) circle (.05cm);
	\filldraw (cba) circle (.05cm);
	\filldraw (cca) circle (.05cm);
	\filldraw (cab) circle (.05cm);
	\filldraw (cbb) circle (.05cm);
	\filldraw (ccb) circle (.05cm);
	\filldraw (cac) circle (.05cm);
	\filldraw (cbc) circle (.05cm);
	\filldraw (ccc) circle (.05cm);
	\draw[blue!50, very thin] (cbb) circle (.15cm);
		\foreach \x/\y/\s in {155/100/24, 175/120/23, 185/140/22, 200/160/23, 220/180/25, 240/200/26} 
		{\draw[dotted, blue!50] ($(cbb)+(\x:\extra)$) to[bend left=\s] ($(z2) + (\y:\lattice)$);}
		\foreach \x/\y/\s in {135/70/30,290/225/28}
		{\draw[blue!50, very thin] ($ (cbb) + (\x:\extra) $) to[bend left=\s] ($ (z2) + (\y:\lattice) $);}
			\draw[blue!50, very thin] (z2) circle (\lattice);
			\draw[thick, red] ($ (z2) - .75*(\lattice,0) $) node[above] {\scriptsize{$B$}} -- (z2);
			\draw ($ (z2) - .75*(0,\lattice) $) node[right] {\scriptsize{$X$}} -- ($ (z2) + .75*(0,\lattice) $) node[left] {\scriptsize{$X$}};
			\draw ($ (z2) - 1.25*(\xoffset,\yoffset) $) node[below] {\scriptsize{$X$}} -- ($ (z2) + 1.25*(\xoffset,\yoffset) $) node[above] {\scriptsize{$X$}};
			\filldraw (z2) circle (.05cm);
			\node at ($ (z2) -1.3*(0,\lattice) $) {$\cX(\Phi(\textcolor{red}{B})XX \to XX)$};
			\draw[blue!50, very thin] ($ (z2) + (135:\lattice) $) -- ($ (z2) +(-45:\lattice) $);
			\foreach \x in {135,-45}
			{\draw[blue!50, very thin, -stealth] ($ (z2) + (\x:.5*\lattice) - (.1,.1)$) to ($ (z2) + (\x:.5*\lattice) + (.1,.1)$);}
	\pgftransformxscale{-1}
	\draw[blue!50, very thin] (aba) circle (.15cm);
		\foreach \x/\y/\s in {155/100/24, 175/120/23, 185/140/22, 200/160/23, 220/180/25, 240/200/26} 
		{\draw[dotted, blue!50] ($(aba)+(\x:\extra)$) to[bend left=\s] ($(z1) + (\y:\lattice)$);}
		\foreach \x/\y/\s in {135/70/30,290/225/28}
		{\draw[blue!50, very thin] ($ (aba) + (\x:\extra) $) to[bend left=\s] ($ (z1) + (\y:\lattice) $);}
			\draw[blue!50, very thin] (z1) circle (\lattice);
			\draw[thick, red] ($ (z1) - .75*(\lattice,0) $) node[below] {\scriptsize{$B$}} -- ($ (z1) + .75*(\lattice,0) $) node[above] {\scriptsize{$B$}};
			\draw[thick, red] ($ (z1) - .75*(0,\lattice) $) node[right] {\scriptsize{$B$}} -- ($ (z1) + .75*(0,\lattice) $) node[left] {\scriptsize{$B$}};
			\draw[thick, red] ($ (z1) - 1.25*(\xoffset,0) + 1.25*(0,\yoffset) $) node[above] {\scriptsize{$B$}} -- ($ (z1) + 1.25*(\xoffset,0) - 1.25*(0,\yoffset) $) node[below] {\scriptsize{$B$}};
			\filldraw[thick, red] (z1) circle (.05cm);
			\node at ($ (z1) -1.3*(0,\lattice) $) {\textcolor{red}{$\cB(BBB \to BBB)$}};
			\draw[blue!50, very thin] ($ (z1) + (225:\lattice) $) -- ($ (z1) + (45:\lattice) $);
			\foreach \x in {225,45}
			{\draw[blue!50, very thin, -stealth] ($ (z1) + (\x:.5*\lattice) + (.1,-.1)$) to ($ (z1) + (\x:.5*\lattice) - (.1,-.1)$);}
\end{tikzpicture}
$$
For each edge, we have projectors which enforce that the edge labels match.
For each plaquette, we have a projector which implements the regular element.

\begin{warn}
\label{warn:Perspective}
When we draw 3D pictures, we use the perspective that the 2D edge is closer to the reader, as that is where the interesting physics lives.
However, when drawing 2D string diagrams, we swap our perspective to look at the 2D boundary \emph{through the 3D bulk}.
This has the effect of reversing the convention for half-braidings.
\end{warn}

$$
\frac{1}{D_\cB}\sum_{r\in \Irr(\cB)}
d_r\cdot
\tikzmath{
\draw[step=1.0,red,thick] (0.5,0.5) grid (2.5,2.5);
\draw[thick, red] (2,1) -- ($ (2,1) + (-.3,.3)$);
\draw[thick, red] (1,1) -- ($ (1,1) + (-.3,.3)$);
\draw[thick, red] (1,2) -- ($ (1,2) + (-.3,.3)$);
\draw[thick, red] (2,2) -- ($ (2,2) + (-.3,.3)$);
\node at (2.3,.8) {$\scriptstyle \xi_{2,1}$};
\node at (.5,2.2) {$\scriptstyle \xi_{1,2}$};
\node at (.7,.8) {$\scriptstyle \xi_{1,1}$};
\node at (2.3,2.2) {$\scriptstyle \xi_{2,2}$};
\node at (1.5,.85) {$\scriptstyle g$};
\node at (2.15,1.5) {$\scriptstyle h$};
\node at (1.5,2.15) {$\scriptstyle i$};
\node at (.85,1.5) {$\scriptstyle j$};
\draw[knot, thick, orange, rounded corners=5pt] (1.15,1.15) rectangle (1.85,1.85);
\fill[gray!60, rounded corners=5pt, opacity=.5] (1.16,1.16) rectangle (1.84,1.84);
\node[orange] at (1.3,1.5) {$\scriptstyle r$};
\draw[thick, knot, red] (2,1) -- ($ (2,1) + (.3,-.3)$);
\draw[thick, knot, red] (1,1) -- ($ (1,1) + (.3,-.3)$);
\draw[thick, knot, red] (1,2) -- ($ (1,2) + (.3,-.3)$);
\draw[thick, knot, red] (2,2) -- ($ (2,2) + (.3,-.3)$);
}
\qquad\qquad\text{and}\qquad\qquad
\frac{1}{D_\cX}\sum_{r\in \Irr(\cX)}
d_r\cdot
\tikzmath{
\draw[step=1.0,black,thin] (0.5,0.5) grid (2.5,2.5);
\draw[thick, red] (2,1) -- ($ (2,1) + (-.3,.3)$);
\draw[thick, red] (1,1) -- ($ (1,1) + (-.3,.3)$);
\draw[thick, red] (1,2) -- ($ (1,2) + (-.3,.3)$);
\draw[thick, red] (2,2) -- ($ (2,2) + (-.3,.3)$);
\node at (2.3,.8) {$\scriptstyle \xi_{2,1}$};
\node at (.5,2.2) {$\scriptstyle \xi_{1,2}$};
\node at (.7,.8) {$\scriptstyle \xi_{1,1}$};
\node at (2.3,2.2) {$\scriptstyle \xi_{2,2}$};
\node at (1.5,.85) {$\scriptstyle g$};
\node at (2.15,1.5) {$\scriptstyle h$};
\node at (1.5,2.15) {$\scriptstyle i$};
\node at (.85,1.5) {$\scriptstyle j$};
\draw[knot, thick, cyan, rounded corners=5pt] (1.15,1.15) rectangle (1.85,1.85);
\fill[gray!60, rounded corners=5pt, opacity=.5] (1.16,1.16) rectangle (1.84,1.84);
\node[cyan] at (1.3,1.5) {$\scriptstyle r$};
}
$$
Observe that the plaquette operator on the left requires the braiding on $\cB$, and the plaquette operator on the right requires the half-braiding for $\cB$ with $\cX$.

Boundary excitations are now localized on two adjacent plaquettes, and using the same procedure as before, we can look at the space of states $|\psi\rangle$ giving a representation of the tube algebra $\Tube(\cX)$.
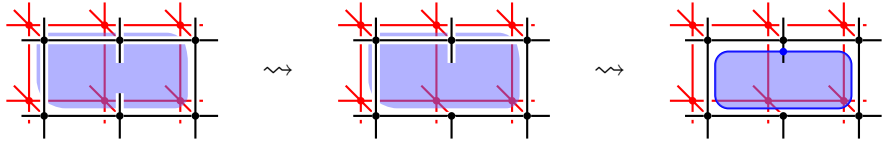
$$
\tikzmath{
\draw[thick, red, knot, xshift=-.2cm, yshift=.2cm] (.7,.7) grid (3.3,2.3);
\foreach \x in {1,2,3}{
\foreach \y in {1,2}{
\draw[thick, red] (\x,\y) -- ($ (\x,\y) + (-.4,.4) $);
}}
\foreach \x in {.8,1.8,2.8}{
\foreach \y in {1.2,2.2}{
\fill[red] (\x,\y) circle (.05cm);
}}
\fill[blue!60, opacity=.5, rounded corners=5pt] (.9,1.5) -- (.9,1.3) -- (1.1,1.1) -- (2.9,1.1) -- (2.9,1.9) -- (2.8,2.1) -- (.9,2.1) -- (.9,1.5);
\draw[thick, knot] (.7,1) -- (3.3,1);
\draw[thick, knot] (.7,2) -- (3.3,2);
\draw[thick, knot] (1,.7) -- (1,2.3);
\draw[thick, knot] (3,.7) -- (3,2.3);
\draw[thick, knot] (2,.7) -- (2,1.3);
\draw[thick, knot] (2,2.3) -- (2,1.7);
\foreach \x in {1,2,3}{
\foreach \y in {1,2}{
\fill (\x,\y) circle (.05cm);
}}
% \node at (2.15,2.15) {$\scriptstyle v$};
% \node at (2.15,1.75) {$\scriptstyle \ell$};
}
\quad\rightsquigarrow\quad
\tikzmath{
\draw[thick, red, knot, xshift=-.2cm, yshift=.2cm] (.7,.7) grid (3.3,2.3);
\foreach \x in {1,2,3}{
\foreach \y in {1,2}{
\draw[thick, red] (\x,\y) -- ($ (\x,\y) + (-.4,.4) $);
}}
\foreach \x in {.8,1.8,2.8}{
\foreach \y in {1.2,2.2}{
\fill[red] (\x,\y) circle (.05cm);
}}
\fill[blue!60, opacity=.5, rounded corners=5pt] (.9,1.5) -- (.9,1.3) -- (1.1,1.1) -- (2.9,1.1) -- (2.9,1.9) -- (2.8,2.1) -- (.9,2.1) -- (.9,1.5);
\draw[thick, knot] (.7,1) -- (3.3,1);
\draw[thick, knot] (.7,2) -- (3.3,2);
\draw[thick, knot] (1,.7) -- (1,2.3);
\draw[thick, knot] (3,.7) -- (3,2.3);
\draw[thick, knot] (2,.7) -- (2,1);
\draw[thick, knot] (2,2.3) -- (2,1.7);
\foreach \x in {1,2,3}{
\foreach \y in {1,2}{
\fill (\x,\y) circle (.05cm);
}}
% \node at (2.15,2.15) {$\scriptstyle v$};
% \node at (2.15,1.75) {$\scriptstyle \ell$};
}
\quad\rightsquigarrow\quad
\tikzmath{
\draw[thick, red, knot, xshift=-.2cm, yshift=.2cm] (.7,.7) grid (3.3,2.3);
\foreach \x in {1,2,3}{
\foreach \y in {1,2}{
\draw[thick, red] (\x,\y) -- ($ (\x,\y) + (-.4,.4) $);
}}
\foreach \x in {.8,1.8,2.8}{
\foreach \y in {1.2,2.2}{
\fill[red] (\x,\y) circle (.05cm);
}}
\draw[rounded corners=5pt, thick, blue] (1.1,1.1) rectangle (2.9,1.85);
\fill[fill=blue!60, opacity=.5, rounded corners=5pt] (1.1,1.1) rectangle (2.9,1.85);
\draw[thick, knot] (.7,1) -- (3.3,1);
\draw[thick, knot] (.7,2) -- (3.3,2);
\draw[thick, knot] (1,.7) -- (1,2.3);
\draw[thick, knot] (3,.7) -- (3,2.3);
\draw[thick, knot] (2,.7) -- (2,1);
\draw[thick, knot] (2,2.3) -- (2,2);
\draw[thick] (2,2) -- (2,1.7);
\foreach \x in {1,2,3}{
\foreach \y in {1,2}{
\fill (\x,\y) circle (.05cm);
}}
\fill[blue] (2,1.85) circle (.05cm);
% \node at (2.15,2.15) {$\scriptstyle v$};
% \node at (2.15,1.75) {$\scriptstyle \ell$};
}
$$
However, we now see that when the $x$-string going around the back of the tube lies in the image of $\cB$ in $\cX$, the existence of the $\cB$-bulk attached to $\cX$ allows us to perform an isotopy which pulls the string back into the bulk.
One may then pull the string across the top $\cB$-plaquette operators.
$$
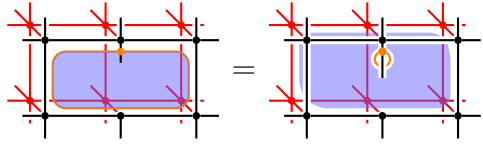
\tikzmath{
\draw[thick, red, knot, xshift=-.2cm, yshift=.2cm] (.7,.7) grid (3.3,2.3);
\foreach \x in {1,2,3}{
\foreach \y in {1,2}{
\draw[thick, red] (\x,\y) -- ($ (\x,\y) + (-.4,.4) $);
}}
\foreach \x in {.8,1.8,2.8}{
\foreach \y in {1.2,2.2}{
\fill[red] (\x,\y) circle (.05cm);
}}
\draw[rounded corners=5pt, thick, orange] (1.1,1.1) rectangle (2.9,1.85);
\fill[fill=blue!60, opacity=.5, rounded corners=5pt] (1.1,1.1) rectangle (2.9,1.85);
\draw[thick, knot] (.7,1) -- (3.3,1);
\draw[thick, knot] (.7,2) -- (3.3,2);
\draw[thick, knot] (1,.7) -- (1,2.3);
\draw[thick, knot] (3,.7) -- (3,2.3);
\draw[thick, knot] (2,.7) -- (2,1);
\draw[thick, knot] (2,2.3) -- (2,2);
\draw[thick] (2,2) -- (2,1.7);
\foreach \x in {1,2,3}{
\foreach \y in {1,2}{
\fill (\x,\y) circle (.05cm);
}}
\fill[orange] (2,1.85) circle (.05cm);
% \node at (2.15,2.15) {$\scriptstyle v$};
% \node at (2.15,1.75) {$\scriptstyle \ell$};
}
=
\tikzmath{
\draw[thick, red, knot, xshift=-.2cm, yshift=.2cm] (.7,.7) grid (3.3,2.3);
\foreach \x in {1,2,3}{
\foreach \y in {1,2}{
\draw[thick, red] (\x,\y) -- ($ (\x,\y) + (-.4,.4) $);
}}
\foreach \x in {.8,1.8,2.8}{
\foreach \y in {1.2,2.2}{
\fill[red] (\x,\y) circle (.05cm);
}}
\fill[blue!60, opacity=.5, rounded corners=5pt] (.9,1.5) -- (.9,1.3) -- (1.1,1.1) -- (2.9,1.1) -- (2.9,1.9) -- (2.8,2.1) -- (.9,2.1) -- (.9,1.5);
\draw[thick, orange, knot] (2,1.75) circle (.1cm);
\draw[thick, knot] (.7,1) -- (3.3,1);
\draw[thick, knot] (.7,2) -- (3.3,2);
\draw[thick, knot] (1,.7) -- (1,2.3);
\draw[thick, knot] (3,.7) -- (3,2.3);
\draw[thick, knot] (2,.7) -- (2,1);
\draw[thick, knot] (2,2.3) -- (2,1.5);
\foreach \x in {1,2,3}{
\foreach \y in {1,2}{
\fill (\x,\y) circle (.05cm);
\fill[orange] (2,1.85) circle (.05cm);
}}
% \node at (2.15,2.15) {$\scriptstyle v$};
% \node at (2.15,1.75) {$\scriptstyle \ell$};
}
$$
That is, the $\Tube(\cX)$ action descends to an action of the \emph{dome algebra} $\Dome^{\textcolor{red}{\cB}}(\cX)$, the quotient of $\Tube(\cX)$ by the ideal
(recall Warning \ref{warn:Perspective}!)
\begin{equation}
\label{eq:DomeIdeal}
\cI=
\left\langle
\tikzmath{
\draw (-.2,.3) --node[left]{$\scriptstyle X$} (-.2,.7);
\draw[thick, red] (.2,.3) --node[right]{$\scriptstyle \Phi b$} (.2,.7);
\draw[thick, red] (-.2,-.3) --node[left]{$\scriptstyle \Phi b$} (-.2,-.7);
\draw (.2,-.3) --node[right]{$\scriptstyle X$} (.2,-.7);
\roundNbox{}{(0,0)}{.3}{.1}{.1}{$f$}
}
-
\tikzmath{
\draw (-.2,.3) --node[left]{$\scriptstyle X$} (-.2,.7);
\draw (.2,-.3) -- (.2,-.7) --node[right]{$\scriptstyle X$} (.2,-1);
\draw[thick, knot, red] (.2,.3) %node[left, yshift=.2cm]{$\scriptstyle \Phi(a)$} 
arc(180:0:.2cm) -- (.6,-.3) arc(0:-180:.4cm) node[left, yshift=-.2cm]{$\scriptstyle \Phi b$};
\roundNbox{}{(0,0)}{.3}{.1}{.1}{$f$}
}
\right\rangle
=
\left\langle
\tikzmath{
\draw[thick] (.5,-.6) -- (.5,1);
\draw[thick] (-.5,-.6) -- (-.5,1);
\draw[thick] (0,1) ellipse (.5 and .2);
\draw[thick] (.5,-.6) arc(0:-180:.5 and .2);
\draw[thick,dotted] (.5,-.6) arc(0:180:.5 and .2);
\draw (0,-.8) --node[left]{$\scriptstyle X$} (0,-.3);
\draw (0,.8) --node[left]{$\scriptstyle X$} (0,.3);
\draw[thick, red] (.5,.2) arc (0:-180:.5 and .2);
\draw[thick, red, dotted] (.5,.2) arc (0:180:.5 and .2);
\roundNbox{fill=white}{(0,0)}{.3}{0}{0}{$f$}
}
-
\tikzmath{
\draw[thick] (.7,-.8) -- (.7,.8);
\draw[thick] (-.7,-.8) -- (-.7,.81);
\draw[thick] (0,.8) ellipse (.7 and .2);
\draw[thick] (.7,-.8) arc(0:-180:.7 and .2);
\draw[thick, dotted] (.7,-.8) arc(0:180:.7 and .2);
\draw (0,-1) --node[left]{$\scriptstyle X$} (0,-.6) -- (0,-.25);
\draw (0,.6) --node[left]{$\scriptstyle X$} (0,.3);
\draw[thick, red, knot] (.3,0) arc(90:-90:.25cm) -- (-.3,-.5) arc(270:90:.25cm);
\roundNbox{fill=white}{(0,0)}{.3}{0}{0}{$f$}
}
\right\rangle.
\end{equation}
Similar to how $\Mod(\Tube(\cX))\cong Z(\cX)$, $\Mod(\Dome^{\textcolor{red}{\cB}}(\cX))\cong Z^{\textcolor{red}{\cB}}(\cX)$ \cite[Thm.~3.9]{2305.14068}, the \emph{enriched center/M\"uger centralizer} $\cB'\subset Z(\cX)$ \cite{MR3725882}.
We have an elegant description of the unitary adjoint $\Tr^{\textcolor{red}{\cB}}: \cX\to Z^{\textcolor{red}{\cB}}(\cX)$ of the forgetful functor $Z^{\textcolor{red}{\cB}}(\cX)\hookrightarrow Z(\cX)\to \cX$
by splitting an idempotent in $\End_{Z(\cX)}(\Tr(x))$ for each $x\in\cX$, analogous to the explicit description of $\Tr_\cX$ given in \cite[\S{4}]{MR3663592} and the idempotent from \cite[Lem.~6.3]{MR4581741}
(compare with the formula from \cite[Rem.~4.18]{2506.19969}).
We have a simple formula for this idempotent
using the language of strings on tubes afforded by $\Tr_\cX$:
$$
\Tr^{\textcolor{red}{\cB}}(x)
=
\im(p^{\textcolor{red}{\cB}}_x)
\qquad
\text{where}
\qquad
p^{\textcolor{red}{\cB}}_x:=
\frac{1}{D_\cB}
\sum_{
b\in\Irr(\cB)
}
d_b\cdot
\tikzmath{
\draw[thick] (.5,-.6) -- (.5,1);
\draw[thick] (-.5,-.6) -- (-.5,1);
\draw (0,-.8) --node[left]{$\scriptstyle x$} (0,.4) -- (0,.8);
\draw[thick, red, knot] (.5,.2) arc (0:-180:.5 and .2);
\draw[thick, red, dotted] (.5,.2) arc (0:180:.5 and .2);
\node[red] at (.2,.2) {$\scriptstyle \Phi b$};
\draw[thick] (0,1) ellipse (.5 and .2);
\draw[thick] (.5,-.6) arc(0:-180:.5 and .2);
\draw[thick,dotted] (.5,-.6) arc(0:180:.5 and .2);
}
\in \End_{Z(\cX)}(\Tr(X)).
$$
Using the usual argument that the plaquette operator in the Levin-Wen model implements isotopy in the skein module, we see that $p^\cB_X$ implements the defining relation of the ideal $\cI\subset \Tube(\cX)$:
$$
\frac{1}{D_\cB}
\sum_{b\in\Irr(\cB)}
d_b\cdot
\tikzmath{
\draw[thick] (.8,-1) -- (.8,1);
\draw[thick] (-.8,-1) -- (-.8,1);
\draw[thick] (0,1) ellipse (.8 and .3);
\draw[thick] (.8,-1) arc(0:-180:.8 and .3);
\draw (0,-1.3) --node[left]{$\scriptstyle X$} (0,-.7) -- (0,-.3);
\draw (0,.7) --node[left]{$\scriptstyle X$} (0,.3);
\draw[thick, dotted] (.8,-1) arc(0:180:.8 and .3);
\draw[thick, red, knot] (.8,-.5) arc(0:-180:.8 and .3);
\draw[thick, red, dotted] (.8,-.5) arc(0:180:.8 and .3);
\draw[thick, red, knot] (.8,.3) arc (0:-180:.8 and .3);
\draw[thick, red, dotted] (.8,.3) arc (0:180:.8 and .3);
\node[red] at (.5,.3) {$\scriptstyle \Phi a$};
\node[red] at (.5,-1) {$\scriptstyle \Phi b$};
\roundNbox{fill=white}{(0,0)}{.3}{0}{0}{$f$}
}
=
\frac{1}{D_\cB}
\sum_{c\in\Irr(\cB)}
d_c\cdot
\tikzmath{
\draw[thick] (.8,-1) -- (.8,1);
\draw[thick] (-.8,-1) -- (-.8,1);
\draw[thick] (0,1) ellipse (.8 and .3);
\draw[thick] (.8,-1) arc(0:-180:.8 and .3);
\draw[thick, dotted] (.8,-1) arc(0:180:.8 and .3);
\draw (0,-1.3) --node[left]{$\scriptstyle X$} (0,-.7) -- (0,-.3);
\draw (0,.7) --node[left]{$\scriptstyle X$} (0,.3);
\draw[thick, red, knot] (.8,-.5) arc(0:-180:.8 and .3);
\draw[thick, red, dotted] (.8,-.5) arc(0:180:.8 and .3);
\draw[thick, red, knot] (.3,0) arc(90:-90:.25cm) -- (-.3,-.5) arc(270:90:.25cm);
\node[red] at (.5,.2) {$\scriptstyle \Phi a$};
\node[red] at (.5,-1) {$\scriptstyle \Phi c$};
\roundNbox{fill=white}{(0,0)}{.3}{0}{0}{$f$}
}
\,.
$$
That is, under the isomorphism $\Tube(\cX)\cong \End_{Z(\cX)}(\Tr(X))$ from \eqref{eq:IsoTube(X)->End(Tr(X))},
$fp^{\textcolor{red}{\cB}}=0$ for all $f\in \cI$ and $1-p^{\textcolor{red}{\cB}}$ acts as the identity on $\Dome^{\textcolor{red}{\cB}}(\cX)$.
Again, the same argument for the Levin-Wen plaquette operator shows that for every $a\in \cB$ and $x\in\cX$,  $\Tr^{\textcolor{red}{\cB}}(x)\in Z(\cX)$ centralizes $a\in\cB$:
$$
\frac{1}{D_\cB}
\sum_{
b\in\Irr(\cB)
}
d_b\cdot
\tikzmath{
\draw[thick, red, snake] (-.8,-.8) -- (-.6,-.6);
\draw[thick, red, snake] (1.2,1.2) --node[above]{$\scriptstyle a$} (.6,.6);
\draw[thick] (.5,-.6) -- (.5,1);
\draw[thick] (-.5,-.6) -- (-.5,1);
\draw (0,-.8) --node[right]{$\scriptstyle x$} (0,.4) -- (0,.8);
\draw[thick, red, knot] (.5,.2) arc (0:-180:.5 and .2);
\draw[thick, red, dotted] (.5,.2) arc (0:180:.5 and .2);
\node[red] at (-.25,.2) {$\scriptstyle \Phi b$};
\draw[thick] (0,1) ellipse (.5 and .2);
\draw[thick] (.5,-.6) arc(0:-180:.5 and .2);
\draw[thick,dotted] (.5,-.6) arc(0:180:.5 and .2);
}
=
\frac{1}{D_\cB}
\sum_{
b\in\Irr(\cB)
}
d_b\cdot
\tikzmath{
\draw[thick] (.5,-.6) -- (.5,1);
\draw[thick] (-.5,-.6) -- (-.5,1);
\draw (0,-.8) --node[right]{$\scriptstyle x$} (0,.4) -- (0,.8);
\draw[thick, red, knot] (.5,.2) arc (0:-180:.5 and .2);
\draw[thick, red, dotted] (.5,.2) arc (0:180:.5 and .2);
\node[red] at (-.25,.2) {$\scriptstyle \Phi b$};
\draw[thick] (0,1) ellipse (.5 and .2);
\draw[thick] (.5,-.6) arc(0:-180:.5 and .2);
\draw[thick,dotted] (.5,-.6) arc(0:180:.5 and .2);
\draw[thick, red, snake, knot] (1.2,1.2) --node[above]{$\scriptstyle a$} (.6,.6) -- (-.8,-.8);
}
.
$$
It follows that
$
\Dome^{\textcolor{red}{\cB}}(\cX)\cong\End_{Z(\cX)}(\Tr^{\textcolor{red}{\cB}}(X))
$
as $\rmC^*$-algebras.

%%%%%%%%%%%%%%%%%%%%%%%%%%%%%%%%%%%%%%%%%%%%%%%%%%%%
\subsection{Defects for enriched string nets and twist defects}

We now consider a $\cB$-\emph{enriched extension} $\cX\subset \cY$ in the sense of \cite{MR4498161}.
That is, $\cX\subset\cY$ is a full fusion subcategory, both $\cX,\cY$ equipped with central functors from $\cB$, and there is an action coherence morphism satisfying \eqref{eq:EnrichedExtension}.

We introduce a 1D $\cY$-defect line on the 2D boundary of our $\cB$-Walker-Wang model, drawn in orange.
$$
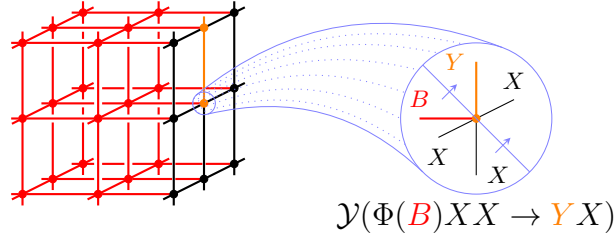
\begin{tikzpicture}
	\pgfmathsetmacro{\lattice}{1};	
	\pgfmathsetmacro{\xoffset}{.4};	
	\pgfmathsetmacro{\yoffset}{.2};	
	\pgfmathsetmacro{\extra}{.15};	
	\pgfmathsetmacro{\zextra}{.35};	
	\pgfmathsetmacro{\zoom}{6};	
	\coordinate (z1) at ($ -.6*(\zoom,0) + (0,\lattice) $);
	\coordinate (z2) at ($ (\zoom,\lattice) $);
	\coordinate (aaa) at (0,0);
	\coordinate (baa) at ($ (aaa) + (\lattice,0) $);
	\coordinate (caa) at ($ (aaa) + 2*(\lattice,0) $);
	\coordinate (aba) at ($ (aaa) + (0,\lattice) $);
	\coordinate (bba) at ($ (aaa) + (\lattice,0) + (0,\lattice) $);
	\coordinate (cba) at ($ (aaa) + 2*(\lattice,0) + (0,\lattice) $);
	\coordinate (aca) at ($ (aaa) + 2*(0,\lattice) $);
	\coordinate (bca) at ($ (aaa) + (\lattice,0) + 2*(0,\lattice) $);
	\coordinate (cca) at ($ (aaa) + 2*(\lattice,0) + 2*(0,\lattice) $);
	\coordinate (aab) at ($ (aaa) + (\xoffset,\yoffset) $);
	\coordinate (bab) at ($ (aaa) + (\lattice,0) + (\xoffset,\yoffset) $);
	\coordinate (cab) at ($ (aaa) + 2*(\lattice,0) + (\xoffset,\yoffset) $);
	\coordinate (abb) at ($ (aaa) + (0,\lattice) + (\xoffset,\yoffset) $);
	\coordinate (bbb) at ($ (aaa) + (\lattice,0) + (0,\lattice) + (\xoffset,\yoffset) $);
	\coordinate (cbb) at ($ (aaa) + 2*(\lattice,0) + (0,\lattice) + (\xoffset,\yoffset) $);
	\coordinate (acb) at ($ (aaa) + 2*(0,\lattice) + (\xoffset,\yoffset) $);
	\coordinate (bcb) at ($ (aaa) + (\lattice,0) + 2*(0,\lattice) + (\xoffset,\yoffset) $);
	\coordinate (ccb) at ($ (aaa) + 2*(\lattice,0) + 2*(0,\lattice) + (\xoffset,\yoffset) $);
	\coordinate (aac) at ($ (aaa) + 2*(\xoffset,\yoffset) $);
	\coordinate (bac) at ($ (aaa) + (\lattice,0) + 2*(\xoffset,\yoffset) $);
	\coordinate (cac) at ($ (aaa) + 2*(\lattice,0) + 2*(\xoffset,\yoffset) $);
	\coordinate (abc) at ($ (aaa) + (0,\lattice) + 2*(\xoffset,\yoffset) $);
	\coordinate (bbc) at ($ (aaa) + (\lattice,0) + (0,\lattice) + 2*(\xoffset,\yoffset) $);
	\coordinate (cbc) at ($ (aaa) + 2*(\lattice,0) + (0,\lattice) + 2*(\xoffset,\yoffset) $);
	\coordinate (acc) at ($ (aaa) + 2*(0,\lattice) + 2*(\xoffset,\yoffset) $);
	\coordinate (bcc) at ($ (aaa) + (\lattice,0) + 2*(0,\lattice) + 2*(\xoffset,\yoffset) $);
	\coordinate (ccc) at ($ (aaa) + 2*(\lattice,0) + 2*(0,\lattice) + 2*(\xoffset,\yoffset) $);
	\draw[thick, red] ($ (aac) - \extra*(\lattice,0) $) -- (cac);
	\draw[thick, red] ($ (abc) - \extra*(\lattice,0) $) -- (cbc);
	\draw[thick, red] ($ (acc) - \extra*(\lattice,0) $) -- (ccc);
	\draw[thick, red] ($ (aac) - \extra*(0,\lattice) $) -- ($ (acc) + \extra*(0,\lattice) $);
	\draw[thick, red] ($ (bac) - \extra*(0,\lattice) $) -- ($ (bcc) + \extra*(0,\lattice) $);
	\draw[thick] ($ (cac) - \extra*(0,\lattice) $) -- ($ (ccc) + \extra*(0,\lattice) $);
	\draw[thick, red, knot] ($ (aab) - \extra*(\lattice,0) $) -- (cab);
	\draw[thick, red, knot] ($ (abb) - \extra*(\lattice,0) $) -- (cbb);
	\draw[thick, red, knot] ($ (acb) - \extra*(\lattice,0) $) -- (ccb);
	\draw[thick, red, knot] ($ (aab) - \extra*(0,\lattice) $) -- ($ (acb) + \extra*(0,\lattice) $);
	\draw[thick, red, knot] ($ (bab) - \extra*(0,\lattice) $) -- ($ (bcb) + \extra*(0,\lattice) $);
	\draw[very thick, white] ($ (cab) - \extra*(0,\lattice) $) -- ($ (ccb) + \extra*(0,\lattice) $);
	\draw[thick] ($ (cab) - \extra*(0,\lattice) $) -- (cbb);
	\draw[thick, orange] (cbb) -- ($ (ccb) + \extra*(0,\lattice) $);
	\draw[thick, red, knot] ($ (aaa) - \extra*(\lattice,0) $) -- (caa);
	\draw[thick, red, knot] ($ (aba) - \extra*(\lattice,0) $) -- (cba);
	\draw[thick, red, knot] ($ (aca) - \extra*(\lattice,0) $) -- (cca);
	\draw[thick, red, knot] ($ (aaa) - \extra*(0,\lattice) $) -- ($ (aca) + \extra*(0,\lattice) $);
	\draw[thick, red, knot] ($ (baa) - \extra*(0,\lattice) $) -- ($ (bca) + \extra*(0,\lattice) $);
	\draw[very thick, white] ($ (caa) - \extra*(0,\lattice) $) -- ($ (cca) + \extra*(0,\lattice) $);
	\draw[thick] ($ (caa) - \extra*(0,\lattice) $) -- ($ (cca) + \extra*(0,\lattice) $);
	\draw[thick, red] ($ (aaa) - \zextra*(\xoffset,\yoffset) $) -- ($ (aac) + \zextra*(\xoffset,\yoffset) $);
	\draw[thick, red] ($ (aba) - \zextra*(\xoffset,\yoffset) $) -- ($ (abc) + \zextra*(\xoffset,\yoffset) $);
	\draw[thick, red] ($ (aca) - \zextra*(\xoffset,\yoffset) $) -- ($ (acc) + \zextra*(\xoffset,\yoffset) $);
	\draw[thick, red] ($ (baa) - \zextra*(\xoffset,\yoffset) $) -- ($ (bac) + \zextra*(\xoffset,\yoffset) $);
	\draw[thick, red] ($ (bba) - \zextra*(\xoffset,\yoffset) $) -- ($ (bbc) + \zextra*(\xoffset,\yoffset) $);
	\draw[thick, red] ($ (bca) - \zextra*(\xoffset,\yoffset) $) -- ($ (bcc) + \zextra*(\xoffset,\yoffset) $);
	\draw[thick] ($ (caa) - \zextra*(\xoffset,\yoffset) $) -- ($ (cac) + \zextra*(\xoffset,\yoffset) $);
	\draw[thick] ($ (cba) - \zextra*(\xoffset,\yoffset) $) -- ($ (cbc) + \zextra*(\xoffset,\yoffset) $);
	\draw[thick] ($ (cca) - \zextra*(\xoffset,\yoffset) $) -- ($ (ccc) + \zextra*(\xoffset,\yoffset) $);
	\filldraw[red] (aaa) circle (.05cm);
	\filldraw[red] (aba) circle (.05cm);
	\filldraw[red] (aca) circle (.05cm);
	\filldraw[red] (aab) circle (.05cm);
	\filldraw[red] (abb) circle (.05cm);
	\filldraw[red] (acb) circle (.05cm);
	\filldraw[red] (aac) circle (.05cm);
	\filldraw[red] (abc) circle (.05cm);
	\filldraw[red] (acc) circle (.05cm);
	\filldraw[red] (baa) circle (.05cm);
	\filldraw[red] (bab) circle (.05cm);
	\filldraw[red] (bac) circle (.05cm);
	\filldraw[red] (bba) circle (.05cm);
	\filldraw[red] (bbb) circle (.05cm);
	\filldraw[red] (bbc) circle (.05cm);
	\filldraw[red] (bca) circle (.05cm);
	\filldraw[red] (bcb) circle (.05cm);
	\filldraw[red] (bcc) circle (.05cm);
	\filldraw (caa) circle (.05cm);
	\filldraw (cba) circle (.05cm);
	\filldraw (cca) circle (.05cm);
	\filldraw (cab) circle (.05cm);
	\filldraw[orange] (cbb) circle (.05cm);
	\filldraw[orange] (ccb) circle (.05cm);
	\filldraw (cac) circle (.05cm);
	\filldraw (cbc) circle (.05cm);
	\filldraw (ccc) circle (.05cm);
	\draw[blue!50, very thin] (cbb) circle (.15cm);
		\foreach \x/\y/\s in {155/100/24, 175/120/23, 185/140/22, 200/160/23, 220/180/25, 240/200/26} 
		{\draw[dotted, blue!50] ($(cbb)+(\x:\extra)$) to[bend left=\s] ($(z2) + (\y:\lattice)$);}
		\foreach \x/\y/\s in {135/70/30,290/225/28}
		{\draw[blue!50, very thin] ($ (cbb) + (\x:\extra) $) to[bend left=\s] ($ (z2) + (\y:\lattice) $);}
			\draw[blue!50, very thin] (z2) circle (\lattice);
			\draw[thick, red] ($ (z2) - .75*(\lattice,0) $) node[above] {\scriptsize{$B$}} -- (z2);
			\draw[thick, orange] (z2) -- ($ (z2) + .75*(0,\lattice) $) node[left] {\scriptsize{$Y$}};
			\draw ($ (z2) - .75*(0,\lattice) $) node[right] {\scriptsize{$X$}} -- (z2);	\draw ($ (z2) - 1.25*(\xoffset,\yoffset) $) node[below] {\scriptsize{$X$}} -- ($ (z2) + 1.25*(\xoffset,\yoffset) $) node[above] {\scriptsize{$X$}};
			\filldraw[orange] (z2) circle (.05cm);
			\node at ($ (z2) -1.3*(0,\lattice) $) {$\cY(\Phi(\textcolor{red}{B})XX \to \textcolor{orange}{Y} X)$};
			\draw[blue!50, very thin] ($ (z2) + (135:\lattice) $) -- ($ (z2) +(-45:\lattice) $);
			\foreach \x in {135,-45}
			{\draw[blue!50, very thin, -stealth] ($ (z2) + (\x:.5*\lattice) - (.1,.1)$) to ($ (z2) + (\x:.5*\lattice) + (.1,.1)$);}
\end{tikzpicture}
$$
Similar to the unenriched case, we change the local Hilbert space along the $\cY$-defect line, which again carry skein module inner products.
By \cite{2305.14068}, the topological boundary excitations which live at the vertex correspond to the \emph{enriched relative center}
$$
\Hom_{\cX-\cX}^{\textcolor{red}{\cB}}(\cX\to \cY)\cong Z^{\textcolor{red}{\cB}}_\cX(\cY).
$$
Similar to the $\Dome^{\textcolor{red}{\cB}}(\cX)$-action to distinguish topological boundary excitations in $Z^{\textcolor{red}{\cB}}(\cX)$, in the presence of the $\cY$-defect, we have an action of the \emph{relative dome algebra} $\Dome^{\textcolor{red}{\cB}}_\cX(\cY)$, which can be defined as either the corner $\Tube_\cX(\cY)p^{\textcolor{red}{\cB}}_Y$ (using that $p^{\textcolor{red}{\cB}}_Y$ is central depicted on the left hand side below)
of the relative tube algebra $\Tube_\cX(\cY)$,
or equivalently as the quotient by a certain ideal $\cI$ depicted on the right hand side below
(recall Warning \ref{warn:Perspective}!).
$$
p^{\textcolor{red}{\cB}}_Y:=
\frac{1}{D_\cB}
\sum_{
b\in\Irr(\cB)
}
d_b\cdot
\tikzmath{
\draw[thick] (.5,-.6) -- (.5,1);
\draw[thick] (-.5,-.6) -- (-.5,1);
\draw[thick, orange] (0,-.8) --node[left]{$\scriptstyle Y$} (0,.4) -- (0,.8);
\draw[thick, red, knot] (.5,.2) arc (0:-180:.5 and .2);
\draw[thick, red, dotted] (.5,.2) arc (0:180:.5 and .2);
\node[red] at (.2,.2) {$\scriptstyle \Phi b$};
\draw[thick] (0,1) ellipse (.5 and .2);
\draw[thick] (.5,-.6) arc(0:-180:.5 and .2);
\draw[thick,dotted] (.5,-.6) arc(0:180:.5 and .2);
}
\qquad\qquad\qquad
\cI_\cX=
\left\langle
\tikzmath{
\draw[thick] (.5,-.6) -- (.5,1);
\draw[thick] (-.5,-.6) -- (-.5,1);
\draw[thick] (0,1) ellipse (.5 and .2);
\draw[thick] (.5,-.6) arc(0:-180:.5 and .2);
\draw[thick,dotted] (.5,-.6) arc(0:180:.5 and .2);
\draw[thick, orange] (0,-.8) --node[left]{$\scriptstyle Y$} (0,-.3);
\draw[thick, orange] (0,.8) --node[left]{$\scriptstyle Y$} (0,.3);
\draw[thick, red] (.5,.2) arc (0:-180:.5 and .2);
\draw[thick, red, dotted] (.5,.2) arc (0:180:.5 and .2);
\roundNbox{fill=white}{(0,0)}{.3}{0}{0}{$f$}
}
-
\tikzmath{
\draw[thick] (.7,-.8) -- (.7,.8);
\draw[thick] (-.7,-.8) -- (-.7,.81);
\draw[thick] (0,.8) ellipse (.7 and .2);
\draw[thick] (.7,-.8) arc(0:-180:.7 and .2);
\draw[thick, dotted] (.7,-.8) arc(0:180:.7 and .2);
\draw[thick, orange] (0,-1) --node[left]{$\scriptstyle Y$} (0,-.6) -- (0,-.25);
\draw[thick, orange] (0,.6) --node[left]{$\scriptstyle Y$} (0,.3);
\draw[thick, red, knot] (.3,0) arc(90:-90:.25cm) -- (-.3,-.5) arc(270:90:.25cm);
\roundNbox{fill=white}{(0,0)}{.3}{0}{0}{$f$}
}
\right\rangle
$$
By \cite[Thm.~3.9]{2305.14068}, 
$\Mod(\Dome^{\textcolor{red}{\cB}}_\cX(\cY))\cong Z^{\textcolor{red}{\cB}}_\cX(\cY)$.

\begin{rem}
The construction of this section can be adapted for any $\cB$-enriched $\cX-\cX$ bimodule category $\cM$, where the twist defects would be given by $Z_\cX^{\textcolor{red}{\cB}}(\cM)\cong\Hom_{\cX-\cX}^{\textcolor{red}{\cB}}(\cX\to \cM)$.
Again as in Remark \ref{rem:ConstructionForBimodule}, we will not have a fusion operation on twist defects without a monoidal structure on $\cM$.
Again, we must require that $\cM$ comes equipped with a unitary trace for the existence of skein module inner products.
$$
\begin{tikzpicture}
	\pgfmathsetmacro{\lattice}{1};	
	\pgfmathsetmacro{\xoffset}{.4};	
	\pgfmathsetmacro{\yoffset}{.2};	
	\pgfmathsetmacro{\extra}{.15};	
	\pgfmathsetmacro{\zextra}{.35};	
	\pgfmathsetmacro{\zoom}{6};	
	\coordinate (z1) at ($ -.6*(\zoom,0) + (0,\lattice) $);
	\coordinate (z2) at ($ (\zoom,\lattice) $);
	\coordinate (aaa) at (0,0);
	\coordinate (baa) at ($ (aaa) + (\lattice,0) $);
	\coordinate (caa) at ($ (aaa) + 2*(\lattice,0) $);
	\coordinate (aba) at ($ (aaa) + (0,\lattice) $);
	\coordinate (bba) at ($ (aaa) + (\lattice,0) + (0,\lattice) $);
	\coordinate (cba) at ($ (aaa) + 2*(\lattice,0) + (0,\lattice) $);
	\coordinate (aca) at ($ (aaa) + 2*(0,\lattice) $);
	\coordinate (bca) at ($ (aaa) + (\lattice,0) + 2*(0,\lattice) $);
	\coordinate (cca) at ($ (aaa) + 2*(\lattice,0) + 2*(0,\lattice) $);
	\coordinate (aab) at ($ (aaa) + (\xoffset,\yoffset) $);
	\coordinate (bab) at ($ (aaa) + (\lattice,0) + (\xoffset,\yoffset) $);
	\coordinate (cab) at ($ (aaa) + 2*(\lattice,0) + (\xoffset,\yoffset) $);
	\coordinate (abb) at ($ (aaa) + (0,\lattice) + (\xoffset,\yoffset) $);
	\coordinate (bbb) at ($ (aaa) + (\lattice,0) + (0,\lattice) + (\xoffset,\yoffset) $);
	\coordinate (cbb) at ($ (aaa) + 2*(\lattice,0) + (0,\lattice) + (\xoffset,\yoffset) $);
	\coordinate (acb) at ($ (aaa) + 2*(0,\lattice) + (\xoffset,\yoffset) $);
	\coordinate (bcb) at ($ (aaa) + (\lattice,0) + 2*(0,\lattice) + (\xoffset,\yoffset) $);
	\coordinate (ccb) at ($ (aaa) + 2*(\lattice,0) + 2*(0,\lattice) + (\xoffset,\yoffset) $);
	\coordinate (aac) at ($ (aaa) + 2*(\xoffset,\yoffset) $);
	\coordinate (bac) at ($ (aaa) + (\lattice,0) + 2*(\xoffset,\yoffset) $);
	\coordinate (cac) at ($ (aaa) + 2*(\lattice,0) + 2*(\xoffset,\yoffset) $);
	\coordinate (abc) at ($ (aaa) + (0,\lattice) + 2*(\xoffset,\yoffset) $);
	\coordinate (bbc) at ($ (aaa) + (\lattice,0) + (0,\lattice) + 2*(\xoffset,\yoffset) $);
	\coordinate (cbc) at ($ (aaa) + 2*(\lattice,0) + (0,\lattice) + 2*(\xoffset,\yoffset) $);
	\coordinate (acc) at ($ (aaa) + 2*(0,\lattice) + 2*(\xoffset,\yoffset) $);
	\coordinate (bcc) at ($ (aaa) + (\lattice,0) + 2*(0,\lattice) + 2*(\xoffset,\yoffset) $);
	\coordinate (ccc) at ($ (aaa) + 2*(\lattice,0) + 2*(0,\lattice) + 2*(\xoffset,\yoffset) $);
	\draw[thick, red] ($ (aac) - \extra*(\lattice,0) $) -- (cac);
	\draw[thick, red] ($ (abc) - \extra*(\lattice,0) $) -- (cbc);
	\draw[thick, red] ($ (acc) - \extra*(\lattice,0) $) -- (ccc);
	\draw[thick, red] ($ (aac) - \extra*(0,\lattice) $) -- ($ (acc) + \extra*(0,\lattice) $);
	\draw[thick, red] ($ (bac) - \extra*(0,\lattice) $) -- ($ (bcc) + \extra*(0,\lattice) $);
	\draw[thick] ($ (cac) - \extra*(0,\lattice) $) -- ($ (ccc) + \extra*(0,\lattice) $);
	\draw[thick, red, knot] ($ (aab) - \extra*(\lattice,0) $) -- (cab);
	\draw[thick, red, knot] ($ (abb) - \extra*(\lattice,0) $) -- (cbb);
	\draw[thick, red, knot] ($ (acb) - \extra*(\lattice,0) $) -- (ccb);
	\draw[thick, red, knot] ($ (aab) - \extra*(0,\lattice) $) -- ($ (acb) + \extra*(0,\lattice) $);
	\draw[thick, red, knot] ($ (bab) - \extra*(0,\lattice) $) -- ($ (bcb) + \extra*(0,\lattice) $);
%	\draw[very thick, white] ($ (cab) - \extra*(0,\lattice) $) -- ($ (ccb) + \extra*(0,\lattice) $);
	\draw[thick] ($ (cab) - \extra*(0,\lattice) $) -- (cab);
	\draw[thick, orange] ($ (cbb) - \zextra*(0,\lattice) $) -- ($ (ccb) + \extra*(0,\lattice) $);
	\draw[thick, red, knot] ($ (aaa) - \extra*(\lattice,0) $) -- (caa);
	\draw[thick, red, knot] ($ (aba) - \extra*(\lattice,0) $) -- (cba);
	\draw[thick, red, knot] ($ (aca) - \extra*(\lattice,0) $) -- (cca);
	\draw[thick, red, knot] ($ (aaa) - \extra*(0,\lattice) $) -- ($ (aca) + \extra*(0,\lattice) $);
	\draw[thick, red, knot] ($ (baa) - \extra*(0,\lattice) $) -- ($ (bca) + \extra*(0,\lattice) $);
	\draw[very thick, white] ($ (caa) - \extra*(0,\lattice) $) -- ($ (cca) + \extra*(0,\lattice) $);
	\draw[thick] ($ (caa) - \extra*(0,\lattice) $) -- ($ (cca) + \extra*(0,\lattice) $);
	\draw[thick, red] ($ (aaa) - \zextra*(\xoffset,\yoffset) $) -- ($ (aac) + \zextra*(\xoffset,\yoffset) $);
	\draw[thick, red] ($ (aba) - \zextra*(\xoffset,\yoffset) $) -- ($ (abc) + \zextra*(\xoffset,\yoffset) $);
	\draw[thick, red] ($ (aca) - \zextra*(\xoffset,\yoffset) $) -- ($ (acc) + \zextra*(\xoffset,\yoffset) $);
	\draw[thick, red] ($ (baa) - \zextra*(\xoffset,\yoffset) $) -- ($ (bac) + \zextra*(\xoffset,\yoffset) $);
	\draw[thick, red] ($ (bba) - \zextra*(\xoffset,\yoffset) $) -- ($ (bbc) + \zextra*(\xoffset,\yoffset) $);
	\draw[thick, red] ($ (bca) - \zextra*(\xoffset,\yoffset) $) -- ($ (bcc) + \zextra*(\xoffset,\yoffset) $);
	\draw[thick] ($ (caa) - \zextra*(\xoffset,\yoffset) $) -- ($ (cac) + \zextra*(\xoffset,\yoffset) $);
	\draw[thick] ($ (cba) - \zextra*(\xoffset,\yoffset) $) -- ($ (cbc) + \zextra*(\xoffset,\yoffset) $);
	\draw[thick] ($ (cca) - \zextra*(\xoffset,\yoffset) $) -- ($ (ccc) + \zextra*(\xoffset,\yoffset) $);
	\filldraw[red] (aaa) circle (.05cm);
	\filldraw[red] (aba) circle (.05cm);
	\filldraw[red] (aca) circle (.05cm);
	\filldraw[red] (aab) circle (.05cm);
	\filldraw[red] (abb) circle (.05cm);
	\filldraw[red] (acb) circle (.05cm);
	\filldraw[red] (aac) circle (.05cm);
	\filldraw[red] (abc) circle (.05cm);
	\filldraw[red] (acc) circle (.05cm);
	\filldraw[red] (baa) circle (.05cm);
	\filldraw[red] (bab) circle (.05cm);
	\filldraw[red] (bac) circle (.05cm);
	\filldraw[red] (bba) circle (.05cm);
	\filldraw[red] (bbb) circle (.05cm);
	\filldraw[red] (bbc) circle (.05cm);
	\filldraw[red] (bca) circle (.05cm);
	\filldraw[red] (bcb) circle (.05cm);
	\filldraw[red] (bcc) circle (.05cm);
	\filldraw (caa) circle (.05cm);
	\filldraw (cba) circle (.05cm);
	\filldraw (cca) circle (.05cm);
	\filldraw (cab) circle (.05cm);
	\filldraw[orange] (cbb) circle (.05cm);
	\filldraw[orange] (ccb) circle (.05cm);
	\filldraw (cac) circle (.05cm);
	\filldraw (cbc) circle (.05cm);
	\filldraw (ccc) circle (.05cm);
	\draw[blue!50, very thin] (cbb) circle (.15cm);
		\foreach \x/\y/\s in {155/100/24, 175/120/23, 185/140/22, 200/160/23, 220/180/25, 240/200/26} 
		{\draw[dotted, blue!50] ($(cbb)+(\x:\extra)$) to[bend left=\s] ($(z2) + (\y:\lattice)$);}
		\foreach \x/\y/\s in {135/70/30,290/225/28}
		{\draw[blue!50, very thin] ($ (cbb) + (\x:\extra) $) to[bend left=\s] ($ (z2) + (\y:\lattice) $);}
			\draw[blue!50, very thin] (z2) circle (\lattice);
			\draw[thick, red] ($ (z2) - .75*(\lattice,0) $) node[above] {\scriptsize{$B$}} -- (z2);
			\draw[thick, orange] (z2) -- ($ (z2) + .75*(0,\lattice) $) node[left] {\scriptsize{$M$}};
			\draw[thick, orange] ($ (z2) - .75*(0,\lattice) $) node[right] {\scriptsize{$M$}} -- (z2);	
            \draw ($ (z2) - 1.25*(\xoffset,\yoffset) $) node[below] {\scriptsize{$X$}} -- ($ (z2) + 1.25*(\xoffset,\yoffset) $) node[above] {\scriptsize{$X$}};
			\filldraw[orange] (z2) circle (.05cm);
			\node at ($ (z2) -1.3*(0,\lattice) $) {$\cM(\Phi(\textcolor{red}{B})X\rhd \textcolor{orange}{M} \to \textcolor{orange}{M} \lhd X)$};
			\draw[blue!50, very thin] ($ (z2) + (135:\lattice) $) -- ($ (z2) +(-45:\lattice) $);
			\foreach \x in {135,-45}
			{\draw[blue!50, very thin, -stealth] ($ (z2) + (\x:.5*\lattice) - (.1,.1)$) to ($ (z2) + (\x:.5*\lattice) + (.1,.1)$);}
\end{tikzpicture}
$$
For general $\cM$, there is no obvious canonical choice of Hilbert space we can put at the end of the defect line, so we always have a boundary point excitation at the defect site.
\end{rem}

\begin{rem}
As in \S\ref{sec:FusionOfTwistDefects} above, we can also describe the fusion of enriched twist defects using an enriched version of the Day convolution product.
Since the functor
$\Tr^{\textcolor{red}{\cB}}_\cX: \cY\to Z^{\textcolor{red}{\cB}}_\cX(\cY)$
is unitary adjoint to the monoidal forgetful functor 
$
Z^{\textcolor{red}{\cB}}_\cX(\cY) \to Z_\cX(\cY) \to \cY,
$
it has a canonical Frobenius structure, and so we get a similar formula for the Day convolution product.
On the enriched lattice, this can be implemented as follows.
$$
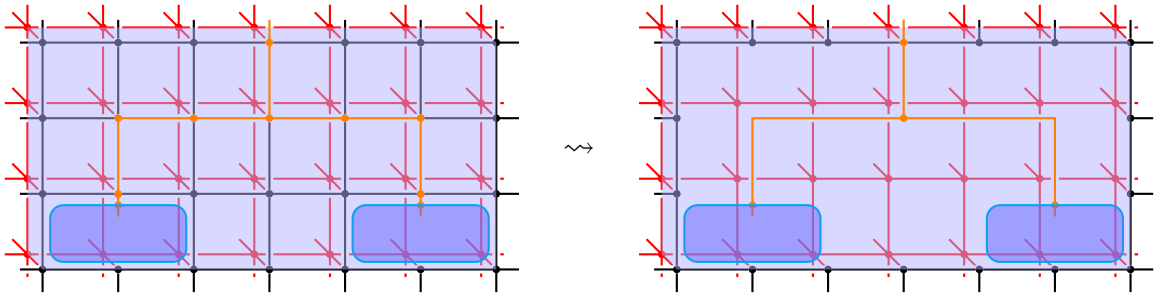
\tikzmath{
\draw[thick, red, knot, xshift=-.2cm, yshift=.2cm] (.7,.7) grid (7.3,4.3);
\foreach \x in {1,2,...,7}{
\foreach \y in {1,2,3,4}{
\draw[thick, red] (\x,\y) -- ($ (\x,\y) + (-.4,.4) $);
}}
\foreach \x in {.8,1.8,...,6.8}{
\foreach \y in {1.2,2.2,3.2,4.2}{
\fill[red] (\x,\y) circle (.05cm);
}}
\draw[thick, knot] (2,3) -- (2,4.3);
\draw[thick, knot] (2,.7) -- (2,1);
\draw[thick, knot] (4,.7) -- (4,3);
\draw[thick, knot] (6,3) -- (6,4.3);
\draw[thick, knot] (6,.7) -- (6,1);
\draw[thick, knot, xscale=2, xshift=-.5cm] (.85,.7) grid (4.15,4.3);
\foreach \x in {1,2,...,7}{
\foreach \y in {1,2,3,4}{
\fill (\x,\y) circle (.05cm);
}}
\draw[thick, white, knot] (2,2) -- (2,3);
\draw[thick, white, knot] (2,3) -- (6,3);
\draw[thick, white, knot] (4,3) -- (4,4.3);
\draw[thick, white, knot] (6,2) -- (6,3);
\fill[fill=blue!30, opacity=.5] (1,1)  -- (.8,1.2) -- (.8,4.2) -- (6.8,4.2) -- (7,4) -- (7,1) -- (1,1);
\fill[orange] (2,1.85) circle (.05cm);
\fill[orange] (2,2) circle (.05cm);
\fill[orange] (2,3) circle (.05cm);
\fill[orange] (3,3) circle (.05cm);
\fill[orange] (4,3) circle (.05cm);
\fill[orange] (4,4) circle (.05cm);
\fill[orange] (5,3) circle (.05cm);
\fill[orange] (6,1.85) circle (.05cm);
\fill[orange] (6,2) circle (.05cm);
\fill[orange] (6,3) circle (.05cm);
\draw[thick, orange] (2,2) -- (2,3);
\draw[thick,orange] (2,2) -- (2,1.7);
\draw[thick, orange] (2,3) -- (6,3);
\draw[thick, orange] (4,3) -- (4,4.3);
\draw[thick, orange] (6,2) -- (6,3);
\draw[thick, orange] (6,2) -- (6,1.7);
\draw[rounded corners=5pt, thick, cyan] (1.1,1.1) rectangle (2.9,1.85);
\fill[fill=blue!60, opacity=.5, rounded corners=5pt] (1.1,1.1) rectangle (2.9,1.85);
\draw[rounded corners=5pt, thick, cyan] (5.1,1.1) rectangle (6.9,1.85);
\fill[fill=blue!60, opacity=.5, rounded corners=5pt] (5.1,1.1) rectangle (6.9,1.85);
}
\quad
\rightsquigarrow
\quad
\tikzmath{
\draw[thick, red, knot, xshift=-.2cm, yshift=.2cm] (.7,.7) grid (7.3,4.3);
\foreach \y in {1,2,3,4}{
\draw[thick, red] ($ (1,\y) + (-.2,.2) $) -- (1,\y);
\draw[thick, red] ($ (7,\y) + (-.2,.2) $) -- (7,\y);
\foreach \x in {1,2,...,7}{
\draw[thick, red] ($ (\x,\y) + (-.2,.2) $) -- ($ (\x,\y) + (-.4,.4) $);
}}
\foreach \x in {.8,1.8,...,6.8}{
\foreach \y in {1.2,2.2,3.2,4.2}{
\fill[red] (\x,\y) circle (.05cm);
}}
\foreach \x in {2,3,...,6}{
\draw[thick, red] ($ (\x,1) + (-.2,.2) $) -- (\x,1);
\draw[thick, red] ($ (\x,4) + (-.2,.2) $) -- (\x,4);
}
\draw[thick, knot] (1,.7) -- (1,4.3);
\draw[thick, knot] (7,.7) -- (7,4.3);
\draw[thick, knot] (.7,1) -- (7.3,1);
\draw[thick, knot] (.7,4) -- (7.3,4);
\foreach \x in {1,7}{
\foreach \y in {2,3}{
\draw[thick, knot] (\x,\y) -- ($ (\x,\y) + {\x-4}*(.1,0) $);
}}
\foreach \x in {2,3,...,6}{
\foreach \y in {1,4}{
\draw[thick, knot] (\x,\y) -- ($ (\x,\y) + {\y-2.5}*(0,.2) $);
\fill (\x,\y) circle (.05cm);
}}
\foreach \x in {1,7}{
\foreach \y in {1,2,3,4}{
\fill (\x,\y) circle (.05cm);
}}
\draw[thick, white, knot] (2,2) -- (2,3) -- (6,3) -- (6,2);
\draw[thick, white, knot] (4,3) -- (4,4.3);
\fill[fill=blue!30, opacity=.5] (1,1)  -- (.8,1.2) -- (.8,4.2) -- (6.8,4.2) -- (7,4) -- (7,1) -- (1,1);
\draw[thick, orange] (2,2) -- (2,3) -- (6,3) -- (6,2);
\draw[thick, orange] (2,1.7) -- (2,2);
\draw[thick, orange] (6,2) -- (6,1.7);
\draw[thick, orange] (4,3) -- (4,4.3);
\fill[orange] (2,1.85) circle (.05cm);
\fill[orange] (4,3) circle (.05cm);
\fill[orange] (4,4) circle (.05cm);
\fill[orange] (6,1.85) circle (.05cm);
\draw[rounded corners=5pt, thick, cyan] (1.1,1.1) rectangle (2.9,1.85);
\fill[fill=blue!60, opacity=.5, rounded corners=5pt] (1.1,1.1) rectangle (2.9,1.85);
\draw[rounded corners=5pt, thick, cyan] (5.1,1.1) rectangle (6.9,1.85);
\fill[fill=blue!60, opacity=.5, rounded corners=5pt] (5.1,1.1) rectangle (6.9,1.85);
}
$$
Thinking higher categorically, $\cY \in \Bim^{\textcolor{red}{\cB}}(\cX)$ is a separable rigid $E_1$-algebra which corresponds to the $E_1$-algebra $Z_\cX^{\textcolor{red}{\cB}}(\cY)\in \Mod(Z^{\textcolor{red}{\cB}}(\cX))$ under the equivalence
$\Bim^{\textcolor{red}{\cB}}(\cX)\cong \Mod(Z^{\textcolor{red}{\cB}}(\cX))$ from \cite[Prop.~II.13]{MR4640433}.
\end{rem}

%%%%%%%%%%%%%%%%%%%%%%%%%%%%%%%%%%%%%%%%%%%%%%%%%%%%
\section{Condensation lattice model for string nets and NI-SETOs}
\label{sec:Condensation}

We now show how to implement the above NI-SETO story using anyon condensation \cite{MR3246855}.
We generalize \cite{PhysRevB.94.235136} for the $G$-graded setting in two important ways: passing to non-invertible SETOs, as well as using enriched models to study chiral theories.

The article \cite{MR4642306} gives a lattice model for condensing an \'etale (commutative, separable) connected ($\dim\cY(1\to A)=1$) algebra $A\in Z(\cY)$ in the $\cY$-Levin-Wen model, which was originally suggested by Corey Jones.
We give a rapid review using the skein theoretic approach to the Levin-Wen model from \cite{2305.14068}.
Below, we use the convention that algebras in $Z(\cY)$ come out of the page as in Warning \ref{warn:Perspective}.

The local Hilbert space assigned to a vertex is now
$$
\cH_v
:=
\cY(Y\otimes Y\to F(A)\otimes Y\otimes Y)
=
\left\{\tikzmath{
\draw (-.7,0) -- (-.3,0);
\draw (.7,0) -- (.3,0);
\draw (0,-.7) -- (0,-.3);
\draw (0,.7) -- (0,.3);
\draw[thick, blue] (0,0) -- (-.3,.7);
\draw[thin, cyan] (-.7,.7) -- (.7,-.7);
\draw[->, cyan, thin] (-.6,.4) -- (-.4,.6);
\draw[<-, cyan, thin] (.6,-.4) -- (.4,-.6);
\roundNbox{fill=white}{(0,0)}{.3}{0}{0}{$f$}
}
\right\}
$$
where the blue strand denotes the image of $A\in Z(\cY)$ in $\cY$ under the forgetful functor $F: Z(\cY)\to \cY$.
We endow $\cH_v$ with the skein module inner product, but we only include normalizations for the edges labeled by simple objects in $Y$ as we want convenient formulas for our condensation projectors.
The Hamiltonian again has $A_\ell$ terms gluing along each edge $\ell$, and each $B_p^y$ term for $y\in\Irr(\cY)$ is updated to include the half-braiding for $A\in Z(\cY)$ with $\cY$:
\begin{align}\label{eq:Bxp}
\tikzmath{
\foreach \x in {0,1}{
\foreach \y in {0,1}{
\draw[thick, blue] ($ 1.5*(\x,\y) $) -- ($ 1.5*(\x,\y) + (-.45,.6) $);
\draw ($ (-.6,0) + 1.5*(\x,\y) $) -- ($ (-.3,0) + 1.5*(\x,\y) $);
\draw ($ (.6,0) + 1.5*(\x,\y) $) -- ($ (.3,0) + 1.5*(\x,\y) $);
\draw ($ (0,-.6) + 1.5*(\x,\y) $) -- ($ (0,-.3) + 1.5*(\x,\y) $);
\draw ($ (0,.6) + 1.5*(\x,\y) $) -- ($ (0,.3) + 1.5*(\x,\y) $);
\roundNbox{fill=white}{($ 1.5*(\x,\y)$)}{.3}{0}{0}{$f_{\x\y}$}
}}
}
\longmapsto
\tikzmath{
\draw[thick, cyan, mid<] (.3,.9) arc (180:90:.3cm);
\draw[thick, cyan, far>] (.3,.6) arc (-180:-90:.3cm);
\draw[thick, cyan, mid<] (.9,1.2) arc (90:0:.3cm);
\draw[thick, cyan, far>] (.9,.3) arc (-90:0:.3cm);
\foreach \x in {0,1}{
\foreach \y in {0,1}{
\draw[thick, blue, knot] ($ 1.5*(\x,\y) $) -- ($ 1.5*(\x,\y) + (-.45,.6) $);
\draw ($ (-.6,0) + 1.5*(\x,\y) $) -- ($ (-.3,0) + 1.5*(\x,\y) $);
\draw ($ (.6,0) + 1.5*(\x,\y) $) -- ($ (.3,0) + 1.5*(\x,\y) $);
\draw ($ (0,-.6) + 1.5*(\x,\y) $) -- ($ (0,-.3) + 1.5*(\x,\y) $);
\draw ($ (0,.6) + 1.5*(\x,\y) $) -- ($ (0,.3) + 1.5*(\x,\y) $);
\roundNbox{fill=white}{($ 1.5*(\x,\y)$)}{.3}{0}{0}{$f_{\x\y}$}
}}
}
\end{align}
There are two additional terms for the Hamiltonian that we may impose; a $C_u$ term at each vertex $u$ projects to the copy of $1\subset A$,
\begin{equation}
\label{eq:UnitTerm}
C_u := i\circ i^\dag = 
\tikzmath{
\draw[thick, blue] (0,.2) -- (0,.5);
\draw[thick, blue] (0,-.2) -- (0,-.5);
\filldraw[blue] (0,.2) circle (.05cm);
\filldraw[blue] (0,-.2) circle (.05cm);
}
\qquad\qquad\qquad
\tikzmath{
\draw[thick, blue] (0,.2) -- (0,.5);
\filldraw[blue] (0,.2) circle (.05cm);
}
:=
i:1\to A,
\end{equation}
and the $D_{v,w}$ term for adjacent vertices $v,w$ implements the multiplication and its adjoint:
\begin{equation}
\label{eq:MultiplicationTerm}
D_{v,w}:= m^\dag \circ m = 
\tikzmath{
\draw[thick, blue] (-.3,-.5) arc(180:0:.3cm);
\draw[thick, blue] (-.3,.5) arc(-180:0:.3cm);
\draw[thick, blue] (0,-.2) -- (0,.2);
\filldraw[blue] (0,.2) circle (.05cm);
\filldraw[blue] (0,-.2) circle (.05cm);
}
\qquad\qquad\qquad
\tikzmath{
\draw[thick, blue] (-.3,-.5) arc(180:0:.3cm);
\draw[thick, blue] (0,-.2) -- (0,.1);
\filldraw[blue] (0,-.2) circle (.05cm);
}
:=
m:A\otimes A\to A.
\end{equation}
Each of $C_u, D_{v,w}$ are orthogonal projections which each commute with all the $A_\ell, B_p$ terms and with other $C_{u'}$ and $D_{v',w'}$ respectively (use that $A$ is commutative!), but the $C_u$ and $D_{u,v}$ will not commute when $u\in\{v,w\}$.
The Hamiltonian 
$$
H_{LW} = -\sum A_\ell - \sum B_p - \sum C_u
$$
implements the $\cY$-Levin-Wen Hamiltonian,
and the Hamiltonian 
$$
H_{\rm cond} = -\sum A_\ell - \sum B_p - \sum D_{v,w}
$$
implements the $A$-condensed phase, producing $Z(\cY)_A^{\rm loc}$ topological order.
One can implement the phase transition by tuning a parameter between the $C_u$ and $D_{v,w}$ terms:
$$
H = 
 -\sum A_\ell - \sum B_p - K\left((1-t)\sum C_u + t\sum \sum D_{v,w}\right).
$$
When $K$ is sufficiently large, tuning the parameter $t$ from zero to one forces the system through the condensation phase change from $Z(\cY)$ to $Z(\cY)_A^{\rm loc}$.

%%%%%%%%%%%%%%%%%%%%%%%%%%%%%%%%%%%%%%%%%%%%%%%%%
\subsection{Example: \texorpdfstring{$G$}{G}-graded extensions as condensations}
\label{sec:GGradedAnomalyFree}

When $\cX\subset \cY$ is a $G$-graded extension, $Z(\cY)$ contains a canonical copy of $\Rep(G)$ whose forgetful image in $\cY$ lands in $\langle 1_\cX\rangle\cong \Hilb\subset \cX$ \cite{MR2587410}.
(This was called a $\Rep(G)$-\emph{fibered enrichment} in \cite{MR4498161}.)
We can recover $Z(\cX)$ topological order by condensing the bosons in the condensable algebra
\[
A=\mathbb{C}^G\cong\bigoplus\limits_{\rho\in\Irr(\Rep(G))}\dim(\rho)\cdot\rho.
\]

The string net condensation model above, with degrees of freedom placed on edges rather than vertices, recovers the SETO story from \cite{PhysRevB.94.235136}.
The black edges are labeled by simple objects in $\cY$ and the red links by basis vectors $|g\rangle$ spanning $F(A)\cong\mathbb{C}^G \cdot 1_{\cY}$.
$$
\tikzmath{
\levinHexGrid[thick, blue]{5}{0}{2}{3}{.5}{black}
}
\quad
\mathcal{H}_\ell=\mathbb{C}^{\Irr(\cY)}
=
\bigoplus_{c\in \Irr(\cY)} \bbC |c\rangle, \quad
\mathcal{H}_{{\color{blue}{p}}} = F(A) \cong \mathbb{C}^G \cdot |1_{\cY} \rangle \cong \bigoplus_{g \in G} \mathbb{C} |g\rangle.
$$
The $A_v$ operator in the Hamiltonian is the usual one enforcing the $\cY$ fusion. The plaquette operator is defined as 
\begin{equation}\label{eq:Bpx}
B_p^y=
 \tikzmath{
 \coordinate (center) at ($ .7*(-.333,.333) +.3*(-.5,0)$);
 \coordinate (pointD) at (canvas polar cs:angle=-30,radius=.866*.5);
 \coordinate (pointE) at (canvas polar cs:angle=-30,radius=.288*.5);
  \levinHex[blue, thick]{0}{0}{.5}{black}
  \draw[cyan, thick] (0,0) circle (0.35cm);
  \draw[white, line width=1mm] ($(center)+(pointD)$) -- +($.5*(-.333, .333)$);
  \draw[thick, blue] ($(center)+(pointD)$) -- +($.5*(-.333, .333)$);
  \node[cyan] at (210:.18) {$\scriptstyle{y}$};
 }, \quad
y \in \cY.
\end{equation}
This $B^y_p$ operator uses the half-braiding of the $A$ when the lightblue $y$-link under-crosses the darkblue $A$-leg.
It contains the usual $\Tilde{B}^y_p$ action on the original $Z(\cY)$ string-net, as well as an additional action on the blue leg in state $g \in G$:
if $y$ is in the $h$-graded component of the extension $\cY$,
\[
B^y_p \; |\text{black edges} \rangle |{\color{blue}{g}} \rangle = \Tilde{B}^y_p  \; |\text{black edges} \rangle |{\color{blue}{hg}} \rangle.
\]
The $D_{p,q}$ operator implementing the condensation is defined as follows \cite{MR4642306}:
\begin{equation}
 \label{eq:DTerm}
 D_{p,q}=
 \tikzmath{
 \levinHexGrid[]{0}{0}{1}{2}{1}{black}
 \coordinate (a) at (.93,.15);
 \coordinate (b) at (1.8,.15);
 \coordinate (c) at (2.3,.5);
 \coordinate (m) at (2.1,.7);
 \coordinate (md) at (1.9,.9);
 \coordinate (d) at (1.9,1.2);
 \coordinate (e) at (1.7,.9);
 \coordinate (f) at (1.9,.35);
 \coordinate (g) at (.87,.25);
 \coordinate (h) at (.5,.45);
 \fill[white] (a) circle (.07cm);
 \fill[white] (g) circle (.07cm);
 \draw[thick, blue] (.7,-.5) .. controls +(90:.3) and +(180:.2) .. (a) -- (b) .. controls +(0:.15) and +(-90:.5) .. (m);
 \draw[thick, blue] (c) .. controls +(135:.2) and +(-45:.2) .. (m);
 \draw[thick, blue] (m) -- (md);
 \draw[thick, blue] (md) -- (d);
 \draw[thick, blue] (md) .. controls +(135:.15) and +(90:.15) .. (e) .. controls +(-90:.2) and +(90:.2) .. (f) arc (0:-90:.1cm) -- (g) .. controls +(180:.2) and +(-90:.2) .. (h);
 \fill[blue] (m) circle (.05cm);
 \fill[blue] (md) circle (.05cm);
 \node at (0,0) {$\scriptstyle p$};
 \node at (1.5,.866) {$\scriptstyle q$};
 }
\end{equation}
Because $A \in Z(\cY)$ forgets to $1_\cY \in \cY$, the $D_{p,q}$ operator only acts on the blue legs of the $p$, $q$ plaquettes and the black link between them, not on other black parallel links.
The action of $D_{p,q}$ on two adjacent blue legs in states $g,h\in G$ separated by a black link labeled by $y\in\mathcal{Y}_k$ is then given by
\begin{equation}
 \label{eq:diagonalGroupD}
 D_{p,q}\left|g,y,h\right\rangle=\delta_{kg=h}\left|g,y,h\right\rangle,
\end{equation}
as in \cite{PhysRevB.94.235136}.

%%%%%%%%%%%%%%%%%%%%%%%%%%%%%%%%%%%%%%%%%%%%%%%%%
\subsection{Excitations and the \texorpdfstring{$A$}{A}-tube algebra}
\label{sec:ExcitationsOfATube}

Again, excitations correspond to violations of a single $A_\ell$ term and two adjacent $B_p,B_q$ terms.
Since $A$ is commutative with trivial twist $(\theta_A=\id_A)$, the site hosting an excitation supports an action of the local $A$-\emph{tube algebra} $\Tube_A^{\loc}(\cY)$, whose category of right modules is $Z(\cY)_A^{\rm loc}$.\footnote{The algebra $\Tube_A(\cY)$ and its quotient $\Tube_A^{\loc}(\cY)$ were previously studied in \cite{MR4642306}, where they were instead called $\widetilde{\Tube_A(\cY)}$ and $\Tube_A(\cY)$ respectively. We adjust the notation to better match the representation categories $\Mod(\Tube_A(\cY))\cong Z(\cY)_A$ and $\Mod(\Tube_A^{\loc}(\cY))\cong Z(\cY)_A^{\loc}$, which are both fundamental to the present story.}
\begin{equation}
\label{eq:CondensationExcitation}
\tikzmath{
\foreach \x in {0,1,2}{
\foreach \y in {0,1}{
\draw ($ (-.7,0) + 1.4*(\x,\y) $) -- ($ (-.3,0) + 1.4*(\x,\y) $);
\draw ($ (.7,0) + 1.4*(\x,\y) $) -- ($ (.3,0) + 1.4*(\x,\y) $);
\draw ($ (0,-.7) + 1.4*(\x,\y) $) -- ($ (0,-.3) + 1.4*(\x,\y) $);
\draw ($ (0,.7) + 1.4*(\x,\y) $) -- ($ (0,.3) + 1.4*(\x,\y) $);
}}
\fill[white] ($ 1.4*(.9,.23) $) rectangle ($ 1.4*(1.1,.65) $);
\fill[blue!30, rounded corners=3pt] (.7,.1) -- (.5,.1) -- (.1,.5) -- (.1,.9) -- (.5,1.3) -- (.9,1.3) -- (1.3,.9) -- (1.5,.9) -- (1.9,1.3) -- (2.3,1.3) -- (2.7,.9) -- (2.7,.5) -- (2.3,.1) -- (1.9,.1) -- (1.5,.5) -- (1.3,.5) -- (.9,.1) -- (.7,.1);
\node at (1.4,.7) {\scriptsize{excitation}};
\foreach \x in {0,1,2}{
\foreach \y in {0,1}{
\draw[thick, blue, knot] ($ 1.4*(\x,\y) $) -- ($ 1.4*(\x,\y) + (-.45,.7) $);
\roundNbox{fill=white}{($ 1.4*(\x,\y)$)}{.3}{0}{0}{$f_{\x\y}$}
}}
}
\quad\rightsquigarrow\quad
\tikzmath{
\foreach \x in {0,1,2}{
\foreach \y in {0,1}{
\draw ($ (-.7,0) + 1.4*(\x,\y) $) -- ($ (-.3,0) + 1.4*(\x,\y) $);
\draw ($ (.7,0) + 1.4*(\x,\y) $) -- ($ (.3,0) + 1.4*(\x,\y) $);
\draw ($ (0,-.7) + 1.4*(\x,\y) $) -- ($ (0,-.3) + 1.4*(\x,\y) $);
\draw ($ (0,.7) + 1.4*(\x,\y) $) -- ($ (0,.3) + 1.4*(\x,\y) $);
}}
\fill[white] ($ 1.4*(.9,.23) $) rectangle ($ 1.4*(1.1,.65) $);
\filldraw[thick,cyan,fill=blue!30, rounded corners=3pt] (.7,.1) -- (.5,.1) -- (.1,.5) -- (.1,.9) -- (.5,1.3) -- (.9,1.3) -- (1.3,.9) -- (1.5,.9) -- (1.9,1.3) -- (2.3,1.3) -- (2.7,.9) -- (2.7,.5) -- (2.3,.1) -- (1.9,.1) -- (1.5,.5) -- (1.3,.5) -- (.9,.1) -- (.7,.1);
\draw (1.4,.9) -- (1.4,.7);
\foreach \x in {0,1,2}{
\foreach \y in {0,1}{
\draw[thick, blue, knot] ($ 1.4*(\x,\y) $) -- ($ 1.4*(\x,\y) + (-.45,.7) $);
\roundNbox{fill=white}{($ 1.4*(\x,\y)$)}{.3}{0}{0}{$f_{\x\y}$}
}}
\draw[thick, blue, knot] (1.4,.9) to[out=135,in=-90] ($ (1.4,1.4) + (-.45,.7) $);
\fill[blue] (1.4,.9) circle (.05cm);
\fill[blue] ($ (1.4,1.4) + (-.45,.7) $) circle (.05cm);
\draw[thick, blue] ($ (1.4,1.4) + (-.45,.7) $) -- ($ (1.4,1.4) + (-.45,1) $);
}
\end{equation}
Formally, we define $\Tube_A^{\loc}(\cY)$ as a quotient of
$$
\Tube_A(\cY):=
\bigoplus_{x\in\Irr(\cY)}
\cY(x\otimes Y\to F(A)\otimes Y\otimes x)
=
\operatorname{span}\set{
\tikzmath{
\draw[thick, blue] (0,0) -- (-.45,.7);
\draw (0,.3) --node[right]{$\scriptstyle Y$} (0,.7);
\draw (0,-.3) --node[right]{$\scriptstyle Y$} (0,-.7);
\draw (.3,0) --node[above]{$\scriptstyle x$} (.7,0);
\draw (-.3,0) --node[above]{$\scriptstyle x$} (-.7,0);
\roundNbox{fill=white}{(0,0)}{.3}{0}{0}{$f$}
} \;
}
{x\in\Irr(\cY)},
$$
which is a finite dimensional $\rmC^*$-algebra with multiplication and adjoint given by
\begin{align}\label{eq:condensationmultiplication}
\tikzmath{
\draw[thick, blue] (0,0) -- (-.45,.7);
\draw (0,.3) --node[right]{$\scriptstyle Y$} (0,.7);
\draw (0,-.3) --node[right]{$\scriptstyle Y$} (0,-.7);
\draw (.3,0) --node[above]{$\scriptstyle x$} (.7,0);
\draw (-.3,0) --node[above]{$\scriptstyle x$} (-.7,0);
\roundNbox{fill=white}{(0,0)}{.3}{0}{0}{$f$}
}
\cdot
\tikzmath{
\draw[thick, blue] (0,0) -- (-.45,.7);
\draw (0,.3) --node[right]{$\scriptstyle Y$} (0,.7);
\draw (0,-.3) --node[right]{$\scriptstyle Y$} (0,-.7);
\draw (.3,0) --node[above]{$\scriptstyle y$} (.7,0);
\draw (-.3,0) --node[above]{$\scriptstyle y$} (-.7,0);
\roundNbox{fill=white}{(0,0)}{.3}{0}{0}{$g$}
}
:=
\sum_{z\in\Irr(\cY)}
\tikzmath{
\draw[thick, blue] (0,1) -- (-.45,1.7);
\draw (0,1.3) --node[right]{$\scriptstyle Y$} (0,1.7);
\draw (0,.3) --node[right]{$\scriptstyle Y$} (0,.7);
\draw (0,-.3) --node[right]{$\scriptstyle Y$} (0,-.7);
\draw (.3,0) to[out=0,in=-135] node[right]{$\scriptstyle y$} (.7,.5);
\draw (-.3,0) to[out=180,in=-45] node[left]{$\scriptstyle y$} (-.7,.5) ;
\draw (.3,1)  to[out=0,in=135] node[right]{$\scriptstyle x$} (.7,.5) -- (1,.5) node[right]{$\scriptstyle z$};
\draw (-.3,1) to[out=180,in=45] node[left]{$\scriptstyle x$} (-.7,.5) -- (-1,.5) node[left]{$\scriptstyle z$};
\roundNbox{fill=white}{(0,1)}{.3}{0}{0}{$f$}
\draw[thick, blue, knot] (0,0) to[out=135,in=-90] (-.45,1.7);
\roundNbox{fill=white}{(0,0)}{.3}{0}{0}{$g$}
\filldraw[blue] (-.45,1.7) circle (.05cm);
\draw[thick, blue] (-.45,1.7) -- (-.45,2);
\filldraw[fill=yellow] (-.7,.5) circle (.05cm);
\filldraw[fill=yellow] (.7,.5) circle (.05cm);
}
\qquad\qquad
\left(\hspace*{-.1cm}
\tikzmath{
\draw[thick, blue] (-.4,.2) to[out=135,in=-90] node[left]{$\scriptstyle A$} (-.7,.7);
\draw (-.2,.3) --node[left,xshift=.1cm]{$\scriptstyle Y$} (-.2,.7);
\draw (.2,.3) --node[right,xshift=-.1cm]{$\scriptstyle x$} (.2,.7);
\draw (-.2,-.3) --node[left,xshift=.1cm]{$\scriptstyle x$} (-.2,-.7);
\draw (.2,-.3) --node[right,xshift=-.1cm]{$\scriptstyle Y$} (.2,-.7);
\roundNbox{fill=white}{(0,0)}{.3}{.1}{.1}{$f$}
}
\right)^\dag
:=
\tikzmath{
\draw[thick, blue] (-.4,-.2)  -- (-.5,-.3);
\draw (-.2,.3) arc(0:180:.2cm) -- (-.6,-.3) --node[left,xshift=.1cm]{$\scriptstyle \overline{x}$} (-.6,-.7);
\draw (.2,.3) --node[right,xshift=-.1cm]{$\scriptstyle Y$} (.2,.7);
\draw (-.2,-.3) --node[left,xshift=.1cm]{$\scriptstyle Y$} (-.2,-.7);
\draw (.2,-.3) arc(-180:0:.2cm) -- (.6,.3) --node[right,xshift=-.1cm]{$\scriptstyle \overline{x}$} (.6,.7);
\draw[thick, blue, knot] (-.5,-.3) to[out=-135,in=-90] node[left]{$\scriptstyle A$} (-.8,.7);
\roundNbox{fill=white}{(0,0)}{.3}{.1}{.1}{$f^\dag$}
}\,.
\end{align}
The shaded dots in the multiplication denote summing over an appropriate ONB of the trivalent vertex space $\cY(x\otimes y\to z)$ and its adjoint (here, we suppress scalars for simplicity), which is independent of the choice,
and the blue $A$-cup is $\coev_A=m^\dag\circ i$.
In the graphical calculus of strings on tubes, one can denote this multiplication and adjoint as 
$$
\tikzmath{
\draw[thick] (.5,-.6) -- (.5,1);
\draw[thick] (-.5,-.6) -- (-.5,1);
\draw[thick] (0,1) ellipse (.5 and .2);
\draw[thick] (.5,-.6) arc(0:-180:.5 and .2);
\draw[thick,dotted] (.5,-.6) arc(0:180:.5 and .2);
\draw[] (0,-.8) --node[left]{$\scriptstyle Y$} (0,-.3);
\draw[] (0,.8) --node[left]{$\scriptstyle Y$} (0,.3);
\draw[thick, cyan] (.5,.2) arc (0:-180:.5 and .2);
\draw[thick, cyan, dotted] (.5,.2) arc (0:180:.5 and .2);
\node[cyan] at (.4,-.1) {$\scriptstyle x$};
\node[cyan] at (-.4,-.1) {$\scriptstyle x$};
\draw[thick, blue, knot] (0,0) -- (-.3,.3) to[out=135, in=-90] (-.8,1.2);
\roundNbox{fill=white}{(0,0)}{.3}{0}{0}{$f$}
}
\,\,
\cdot
\tikzmath{
\draw[thick] (.5,-.6) -- (.5,1);
\draw[thick] (-.5,-.6) -- (-.5,1);
\draw[thick] (0,1) ellipse (.5 and .2);
\draw[thick] (.5,-.6) arc(0:-180:.5 and .2);
\draw[thick,dotted] (.5,-.6) arc(0:180:.5 and .2);
\draw[] (0,-.8) --node[left]{$\scriptstyle Y$} (0,-.3);
\draw[] (0,.8) --node[left]{$\scriptstyle Y$} (0,.3);
\draw[thick, purple] (.5,.2) arc (0:-180:.5 and .2);
\draw[thick, purple, dotted] (.5,.2) arc (0:180:.5 and .2);
\node[purple] at (.4,-.1) {$\scriptstyle y$};
\node[purple] at (-.4,-.1) {$\scriptstyle y$};
\draw[thick, blue, knot] (0,0) -- (-.3,.3) to[out=135, in=-90] (-.8,1.2);
\roundNbox{fill=white}{(0,0)}{.3}{0}{0}{$g$}
}
=
\tikzmath{
\draw[thick] (.5,-.6) -- (.5,2);
\draw[thick] (-.5,-.6) -- (-.5,2);
\draw[thick] (0,2) ellipse (.5 and .2);
\draw[thick] (.5,-.6) arc(0:-180:.5 and .2);
\draw[thick,dotted] (.5,-.6) arc(0:180:.5 and .2);
\draw[] (0,-.8) --node[left]{$\scriptstyle Y$} (0,-.3);
\draw[] (0,.8) --node[left]{$\scriptstyle Y$} (0,.3);
\draw[] (0,1.3) --node[left]{$\scriptstyle Y$} (0,1.8);
\draw[thick, purple] (.5,.2) arc (0:-180:.5 and .2);
\draw[thick, purple, dotted] (.5,.2) arc (0:180:.5 and .2);
\draw[thick, cyan] (.5,1.2) arc (0:-180:.5 and .2);
\draw[thick, cyan, dotted] (.5,1.2) arc (0:180:.5 and .2);
\draw[thick, blue, knot] (0,0) -- (-.3,.3) to[out=135, in=-90] (-.8,1.2) -- (-.8,1.8);
\draw[thick, blue, knot] (0,1) -- (-.3,1.3) to[out=135, in=-45] (-.8,1.8) -- (-.8,2.2);
\filldraw[blue] (-.8,1.8) circle (.05cm);
\roundNbox{fill=white}{(0,0)}{.3}{0}{0}{$g$}
\roundNbox{fill=white}{(0,1)}{.3}{0}{0}{$f$}
}
=
\sum_{z\in\Irr(\cY)}
\tikzmath{
\draw[thick] (1,-.6) -- (1,2);
\draw[thick] (-1,-.6) -- (-1,2);
\draw[thick] (0,2) ellipse (1 and .2);
\draw[thick] (1,-.6) arc(0:-180:1 and .2);
\draw[thick,dotted] (1,-.6) arc(0:180:1 and .2);
\draw[] (0,-.8) --node[left]{$\scriptstyle Y$} (0,-.3);
\draw[] (0,.8) --node[left]{$\scriptstyle Y$} (0,.3);
\draw[] (0,1.3) --node[left]{$\scriptstyle Y$} (0,1.8);
\draw[thick, purple] (.3,0) to[out=0,in=-135] node[right]{$\scriptstyle y$} (.7,.5);
\draw[thick, purple] (-.3,0) to[out=180,in=-45] node[left]{$\scriptstyle y$} (-.7,.5) ;
\draw[thick, cyan] (.3,1)  to[out=0,in=135] node[right]{$\scriptstyle x$} (.7,.5);
\draw[thick, cyan] (-.3,1) to[out=180,in=45] node[left]{$\scriptstyle x$} (-.7,.5);
\draw[thick, green] (-.7,.5) arc (260:240:1);
\draw[thick, green] (.7,.5) arc (-80:-60:1);
\draw[thick, green, dotted] (1,.6) arc (0:180:1 and .2);
\node[green] at (.9,.4) {$\scriptstyle z$};
\node[green] at (-.9,.4) {$\scriptstyle z$};
\filldraw[fill=yellow] (-.7,.5) circle (.05cm);
\filldraw[fill=yellow] (.7,.5) circle (.05cm);
\draw[thick, blue, knot] (0,0) -- (-.3,.3) to[out=135, in=-90] (-1.2,1.2) -- (-1.2,1.8);
\draw[thick, blue, knot] (0,1) -- (-.3,1.3) to[out=135, in=-45] (-1.2,1.8) -- (-1.2,2.2);
\filldraw[blue] (-1.2,1.8) circle (.05cm);
\roundNbox{fill=white}{(0,0)}{.3}{0}{0}{$g$}
\roundNbox{fill=white}{(0,1)}{.3}{0}{0}{$f$}
}
\qquad\text{and}\qquad
\tikzmath{
\draw[thick] (.5,-.6) -- (.5,1);
\draw[thick] (-.5,-.6) -- (-.5,1);
\draw[thick] (0,1) ellipse (.5 and .2);
\draw[thick] (.5,-.6) arc(0:-180:.5 and .2);
\draw[thick,dotted] (.5,-.6) arc(0:180:.5 and .2);
\draw[] (0,-.8) --node[left]{$\scriptstyle Y$} (0,-.3);
\draw[] (0,.8) --node[left]{$\scriptstyle Y$} (0,.3);
\draw[thick, cyan] (.5,.2) arc (0:-180:.5 and .2);
\draw[thick, cyan, dotted] (.5,.2) arc (0:180:.5 and .2);
\node[cyan] at (-.4,-.1) {$\scriptstyle x$};
\node[cyan] at (.4,-.1) {$\scriptstyle x$};
\draw[thick, blue, knot] (0,0) -- (-.3,.3) to[out=135, in=-90] (-.8,1.2);
\roundNbox{fill=white}{(0,0)}{.3}{0}{0}{$f$}
}
^\dag
\,\,
=
\tikzmath{
\draw[thick] (.5,-.6) -- (.5,1);
\draw[thick] (-.5,-.6) -- (-.5,1);
\draw[thick] (0,1) ellipse (.5 and .2);
\draw[thick] (.5,-.6) arc(0:-180:.5 and .2);
\draw[thick,dotted] (.5,-.6) arc(0:180:.5 and .2);
\draw[] (0,-.8) --node[left]{$\scriptstyle Y$} (0,-.3);
\draw[] (0,.8) --node[left]{$\scriptstyle Y$} (0,.3);
\draw[thick, cyan] (.5,.2) arc (0:-180:.5 and .2);
\draw[thick, cyan, dotted] (.5,.2) arc (0:180:.5 and .2);
\node[cyan] at (-.4,-.1) {$\scriptstyle \overline{x}$};
\node[cyan] at (.4,-.1) {$\scriptstyle \overline{x}$};
\draw[thick, blue, knot] (0,0) -- (-.3,-.3) to[out=-135, in=-90] (-.8,1.2);
\roundNbox{fill=white}{(0,0)}{.3}{0}{0}{$f^\dag$}
}
$$
where one uses the (adjoint of the) unit of the unitary adjunction $\Forget\dashv^\dag \Tr$ to pull the $A$-string off the tube as in \cite[\S{4.3}]{MR3578212}.

Before quotienting, the representation category of $\Tube_A(\cY)$ algebra is $Z(\cY)_A$, the fusion category of $A$-modules in $Z(\cY)$.\footnote{Since $A\in Z(\cY)$ is commutative, there are canonical equivalences between the categories of left and of right $A$-modules.}
To see this, one bootstraps the equivalence $Z(\cY)\cong \Mod(\Tube(\cY))$ described in \cite[\S{2.3}]{2305.14068}.
Given a $K\in Z(\cY)_A$, the underlying object in $Z(\cY)$ can be identified with a right $\Tube(\cY)$-module via
$$
K
\longmapsto
Z(\cY)(\Tr(Y) \to K)\cong \cY(Y\to F(K)),
$$
where $\Tube(\cY)\cong \End_{Z(\cY)}(\Tr(Y))$ acts via precomposition.
Now since $K$ has a right $A$-action, 
this $\Tube(\cY)$ induces a right $\Tube_A(\cY)$-action by post-composing with the following $A$-module action.
\begin{equation}
\label{eq:TubeAY-RightAction}
\tikzmath{
\draw[thick] (.5,-.6) -- (.5,0);
\draw[thick] (-.5,-.6) -- (-.5,0);
\draw[thick] (.5,-.6) arc(0:-180:.5 and .2);
\draw[thick,dotted] (.5,-.6) arc(0:180:.5 and .2);
\fill[lightgray] (.5,0) arc(0:-180:.5 and .2) -- (-.5,.2) arc(180:0:.5cm) -- (.5,0);
\draw[thick] (.5,0) arc(0:-180:.5 and .2);
\draw[thick] (.5,0) -- (.50,.2) arc(0:180:.5cm) -- (-.5,0);
\draw[thick,dotted] (.5,0) arc(0:180:.5 and .2);
\draw[thick, orange, snake] (0,.7) --node[right]{$\scriptstyle K$} (0,1.1);
\node at (0,.4) {$\eta$};
\draw[thick] (0,-.8) -- (0,-.2);
}
\lhd
\tikzmath{
\draw[thick] (.5,-.6) -- (.5,1);
\draw[thick] (-.5,-.6) -- (-.5,1);
\draw[thick] (0,1) ellipse (.5 and .2);
\draw[thick] (.5,-.6) arc(0:-180:.5 and .2);
\draw[thick,dotted] (.5,-.6) arc(0:180:.5 and .2);
\draw[thick] (0,-.8) --node[left]{$\scriptstyle Y$} (0,-.3);
\draw[thick] (0,.8) --node[left]{$\scriptstyle Y$} (0,.3);
\draw[thick, cyan] (.5,.2) arc (0:-180:.5 and .2);
\draw[thick, cyan, dotted] (.5,.2) arc (0:180:.5 and .2);
\node[cyan] at (.4,-.1) {$\scriptstyle x$};
\node[cyan] at (-.4,-.1) {$\scriptstyle x$};
\draw[thick, blue, knot] (0,0) -- (-.3,.3) to[out=135, in=-90] (-.8,1.2);
\roundNbox{fill=white}{(0,0)}{.3}{0}{0}{$f$}
}
:=
\tikzmath{
\fill[lightgray] (.5,1) arc(0:-180:.5 and .2) -- (-.5,1.2) arc(180:0:.5cm) -- (.5,1);
\draw[thick] (.5,1) arc(0:-180:.5 and .2);
\draw[thick] (.5,1) -- (.50,1.2) arc(0:180:.5cm) -- (-.5,1);
\draw[thick,dotted] (.5,1) arc(0:180:.5 and .2);
\draw[thick, orange, snake] (0,1.7) --node[right]{$\scriptstyle K$} (0,2.5);
\node at (0,1.4) {$\eta$};
\draw[thick] (.5,-.6) -- (.5,1);
\draw[thick] (-.5,-.6) -- (-.5,1);
\draw[thick] (.5,-.6) arc(0:-180:.5 and .2);
\draw[thick,dotted] (.5,-.6) arc(0:180:.5 and .2);
\draw[thick] (0,-.8) --node[left]{$\scriptstyle Y$} (0,-.3);
\draw[thick] (0,.8) --node[left]{$\scriptstyle Y$} (0,.3);
\draw[thick, cyan] (.5,.2) arc (0:-180:.5 and .2);
\draw[thick, cyan, dotted] (.5,.2) arc (0:180:.5 and .2);
\node[cyan] at (.4,-.1) {$\scriptstyle x$};
\node[cyan] at (-.4,-.1) {$\scriptstyle x$};
\draw[thick, blue, knot] (0,0) -- (-.3,.3) to[out=135, in=-90] (-.8,1.1) to[out=90, in=-135] (0,2);
\fill[blue] (0,2) circle (.05cm);
\roundNbox{fill=white}{(0,0)}{.3}{0}{0}{$f$}
}
\end{equation}
This action is easily seen to be a $*$-action via the identity
$$
\langle \xi | \eta \lhd f \rangle
=
\frac{1}{d_K}
\tr^\cY_{K}
\left(
\tikzmath{
\draw[thick] (-.3,.7) -- node[left]{$\scriptstyle Y$} (-.3,.3);
\draw[thick] (.3,-.3) -- node[right]{$\scriptstyle Y$} (.3,-.7);
\draw[thick, cyan] (-.3,-.3) node[right, yshift=-.15cm]{$\scriptstyle x$} arc (0:-180:.3cm) -- node[left]{$\scriptstyle \overline{x}$} (-.9,1) arc (180:0:.6cm) -- node[right]{$\scriptstyle x$} (.3,.3);
\draw[thick, orange, snake, knot] (-.3,2.3) -- node[right]{$\scriptstyle K$} (-.3,1.6) --(-.3,1.3);
\draw[thick, orange, snake] (.3,-1.7) -- node[right]{$\scriptstyle K$} (.3,-1.3);
\draw[thick, blue, knot] (-.3,0) -- (-.6,.3) to[out=135, in=-90] (-1.1,1.1) to[out=90, in=-135] (-.3,2);
\fill[red] (-.3,2) circle (.05cm);
\roundNbox{fill=white}{(0,0)}{.3}{.3}{.3}{$f$}
\roundNbox{fill=white}{(-.3,1)}{.3}{0}{0}{$\eta$}
\roundNbox{fill=white}{(.3,-1)}{.3}{0}{0}{$\xi^\dag$}
}
\right)
=
\sum_{y\in \Irr(\cY)}\sum_{i=1}^{n_y}
\frac{1}{d_K}\tr^\cY_y\left(
\tikzmath{
\draw[thick, cyan] (.2,-2.4) -- (.2,-.45) .. controls ++(90:.3cm) and ++(-90:.3cm) .. (-.2,.45) --node[left]{$\scriptstyle x$} (-.2,1.4);
\draw[thick] (.3,.95) --node[right]{$\scriptstyle Y$} (.3,1.4);
\draw[thick] (-.3,-1.95) --node[left]{$\scriptstyle Y$} (-.3,-2.4);
\draw (0,2) --node[right]{$\scriptstyle y$} (0,2.4);
\draw (0,-3.05) --node[right]{$\scriptstyle y$} (0,-3.45); % a little lower due to taller ticket
\draw[thick, orange, snake, knot] (-.3,-1.45) -- (-.3,-.45) node[left, yshift=.2cm]{$\scriptstyle K$} .. controls ++(90:.3cm) and ++(-90:.3cm) .. (.3,.45);
\draw[thick, blue, knot] (-.3,-2.7) -- (-.6,-2.4) to[out=135, in=-90] (-.8,-1.6) to[out=90, in=-135] (-.3,-.7);
\fill[red] (-.3,-.7) circle (.05cm);
\roundNbox{fill=white}{(0,1.7)}{.3}{.2}{.2}{$g_{y,i}$}
\roundNbox{fill=white}{(.3,.7)}{.25}{0}{0}{$\xi$}
\roundNbox{fill=white}{(-.3,-1.7)}{.25}{0}{0}{$\eta$}
\roundNbox{fill=white}{(0,-2.7)}{.35}{.2}{.2}{$h^\dag_{y,i}$} 
}
\right)
$$
where on the right hand side, we have factored $f=\sum_{y\in\Irr(\cY)}\sum_{i=1}^{n_y} h_{y,i}^\dag g_{y,i}$ through $\Irr(\cY)$ as in \cite[Eq.~(8)]{2305.14068} using semisimplicity and used the tracial property of $\tr^\cY$.

One can use the above formula to give matrix elements for both the half-braiding for $K$ with $x$ as in \cite{2305.14068} and the module action $\rho: K\otimes A\to K$ as in \cite[\S{3.4.3}]{MR4642306} by respectively considering $f$ of the forms
$$
\tikzmath{
\draw[] (-.2,.7) -- node[left]{$\scriptstyle Y$} (-.2,.3);
\draw[] (.2,-.3) -- node[right]{$\scriptstyle Y$} (.2,-.7);
\draw[thick, cyan] (-.2,-.3) --node[left]{$\scriptstyle x$} (-.2,-.7);
\draw[thick, cyan] (.2,.3) --node[right]{$\scriptstyle x$} (.2,.7);
\draw[thick, blue, knot] (0,0) to[out=135, in=-90] (-.8,.7);
\roundNbox{fill=white}{(0,0)}{.3}{.1}{.1}{$f$}
}
\qquad\qquad
\text{ or }
\qquad\qquad
\tikzmath{
\draw[] (0,.7) -- node[right]{$\scriptstyle Y$} (0,.3);
\draw[] (0,-.3) -- node[right]{$\scriptstyle Y$} (0,-.7);
\draw[thick, blue, knot] (0,0) to[out=135, in=-90] (-.7,.7);
\roundNbox{fill=white}{(0,0)}{.3}{.0}{0}{$f$}
}\,.
$$

Conversely, given a right $\Tube_A(\cY)$-module $\cK$,
we may restrict it to $\Tube(\cY)$, which sits inside $\Tube_A(\cY)$ as a subalgebra by
\begin{equation}
\label{eq:IncludeTubeIntoATube}
\tikzmath{
\draw[thick] (.5,-.6) -- (.5,1);
\draw[thick] (-.5,-.6) -- (-.5,1);
\draw[thick] (0,1) ellipse (.5 and .2);
\draw[thick] (.5,-.6) arc(0:-180:.5 and .2);
\draw[thick,dotted] (.5,-.6) arc(0:180:.5 and .2);
\draw[] (0,-.8) --node[left]{$\scriptstyle Y$} (0,-.3);
\draw[] (0,.8) --node[left]{$\scriptstyle Y$} (0,.3);
\draw[thick, cyan] (.5,.2) arc (0:-180:.5 and .2);
\draw[thick, cyan, dotted] (.5,.2) arc (0:180:.5 and .2);
\node[cyan] at (.4,-.1) {$\scriptstyle x$};
\node[cyan] at (-.4,-.1) {$\scriptstyle x$};
\roundNbox{fill=white}{(0,0)}{.3}{0}{0}{$f$}
}
\longmapsto
\tikzmath{
\draw[thick] (.5,-.6) -- (.5,1);
\draw[thick] (-.5,-.6) -- (-.5,1);
\draw[thick] (0,1) ellipse (.5 and .2);
\draw[thick] (.5,-.6) arc(0:-180:.5 and .2);
\draw[thick,dotted] (.5,-.6) arc(0:180:.5 and .2);
\draw[] (0,-.8) --node[left]{$\scriptstyle Y$} (0,-.3);
\draw[] (0,.8) --node[left]{$\scriptstyle Y$} (0,.3);
\draw[thick, cyan] (.5,.2) arc (0:-180:.5 and .2);
\draw[thick, cyan, dotted] (.5,.2) arc (0:180:.5 and .2);
\node[cyan] at (.4,-.1) {$\scriptstyle x$};
\node[cyan] at (-.4,-.1) {$\scriptstyle x$};
\draw[thick, blue] (-.8,.6) -- (-.8,1.2);
\fill[blue] (-.8,.6) circle (.05cm);
\roundNbox{fill=white}{(0,0)}{.3}{0}{0}{$f$}
}\,.
\end{equation}
The projections $p_y:=\Tr(\id_y)\otimes i_A$ for $y\in\Irr(\cY)$ endow $\cK$ with an $\Irr(\cY)$-grading: $\cK_y := p_y\cK$.
We get an object in $K\in\cY$ by the formula 
$$
K:=\bigoplus_{y\in\cY} \cK_y\otimes y.
$$
We can then endow $K$ with both a half-braiding with all $y\in \Irr(\cY)$ and with a module action $K\otimes A\to K$ via the matrix element formula discussed above.

To get only local modules, we quotient out by the ideal
$$
\cI_A(\cY):=
\left\langle
\tikzmath{
\draw[thick, blue] (0,0) -- (-.45,.7);
\draw (0,.3) --node[right]{$\scriptstyle Y$} (0,.9);
\draw (0,-.3) --node[right]{$\scriptstyle Y$} (0,-.7);
\draw (.3,0) --node[above]{$\scriptstyle y$} (.7,0);
\draw (-.3,0) --node[above]{$\scriptstyle x$} (-.7,0);
\draw (1.3,0) --node[above]{$\scriptstyle x$} (1.7,0);
\draw[thick, blue,knot] (1,0) to[out=135,in=0] (-.45,.7);
\roundNbox{fill=white}{(0,0)}{.3}{0}{0}{$f$}
\roundNbox{fill=white}{(1,0)}{.3}{0}{0}{$h$}
\filldraw[blue] (-.45,.7) circle (.05cm);
\draw[thick, blue] (-.45,.7) -- (-.45,1);
}
-
\tikzmath{
\draw[thick, blue] (0,0) -- (-.45,.7);
\draw (0,.3) --node[right]{$\scriptstyle Y$} (0,.9);
\draw (0,-.3) --node[right]{$\scriptstyle Y$} (0,-.7);
\draw (.3,0) --node[above]{$\scriptstyle y$} (.7,0);
\draw (-.3,0) --node[above]{$\scriptstyle x$} (-.7,0);
\draw (-1.3,0) --node[above]{$\scriptstyle y$} (-1.7,0);
\draw[thick, blue,knot] (-1,0) -- (-1.3,.3) to[out=135,in=-135] (-.45,.7);
\roundNbox{fill=white}{(0,0)}{.3}{0}{0}{$f$}
\roundNbox{fill=white}{(-1,0)}{.3}{0}{0}{$h$}
\filldraw[blue] (-.45,.7) circle (.05cm);
\draw[thick, blue] (-.45,.7) -- (-.45,1);
}
\right\rangle
=
\left\langle\,
\tikzmath{
\draw[thick] (1,-.6) -- (1,.8);
\draw[thick] (-1,-.6) -- (-1,.8);
\draw[thick] (0,.8) ellipse (1 and .2);
\draw[thick] (1,-.6) arc(0:-180:1 and .2);
\draw[thick,dotted] (1,-.6) arc(0:180:1 and .2);
\draw[] (-.5,-.78) --node[right]{$\scriptstyle Y$} (-.5,-.3);
\draw[] (-.5,.64) --node[right]{$\scriptstyle Y$} (-.5,.3);
\draw[thick, purple] (-.5,0) --node[below,yshift=.1cm]{$\scriptstyle y$} (.5,0);
\draw[thick, cyan] (-.7,0) arc (260:240:1);
\draw[thick, cyan] (.7,0) arc (-80:-60:1);
\draw[thick, cyan, dotted] (1,.1) arc (0:180:1 and .2);
\node[cyan] at (.9,-.1) {$\scriptstyle x$};
\node[cyan] at (-.9,-.1) {$\scriptstyle x$};
\draw[thick, blue, knot] (-.5,0) -- (-.8,.3) to[out=135, in=-90] (-1.2,.8) -- (-1.2,1.1);
\draw[thick, blue, knot] (.5,0) -- (.2,.3) to[out=135, in=-45] (-1.2,.8);
\filldraw[blue] (-1.2,.8) circle (.05cm);
\roundNbox{fill=white}{(-.5,0)}{.3}{0}{0}{$f$}
\roundNbox{fill=white}{(.5,0)}{.3}{0}{0}{$h$}
}
-
\tikzmath{
\draw[thick] (1,-.6) -- (1,.8);
\draw[thick] (-1,-.6) -- (-1,.8);
\draw[thick] (0,.8) ellipse (1 and .2);
\draw[thick] (1,-.6) arc(0:-180:1 and .2);
\draw[thick,dotted] (1,-.6) arc(0:180:1 and .2);
\draw[] (.5,-.78) --node[left]{$\scriptstyle Y$} (.5,-.3);
\draw[] (.5,.64) --node[right]{$\scriptstyle Y$} (.5,.3);
\draw[thick, cyan] (-.5,0) --node[below,yshift=.1cm]{$\scriptstyle x$} (.5,0);
\draw[thick, purple] (-.7,0) arc (260:240:1);
\draw[thick, purple] (.7,0) arc (-80:-60:1);
\draw[thick, purple, dotted] (1,.1) arc (0:180:1 and .2);
\node[purple] at (.9,-.1) {$\scriptstyle y$};
\node[purple] at (-.9,-.1) {$\scriptstyle y$};
\draw[thick, blue, knot] (-.5,0) -- (-.8,.3) to[out=135, in=-90] (-1.2,.8) -- (-1.2,1.1);
\draw[thick, blue, knot] (.5,0) -- (.2,.3) to[out=135, in=-45] (-1.2,.8);
\filldraw[blue] (-1.2,.8) circle (.05cm);
\roundNbox{fill=white}{(-.5,0)}{.3}{0}{0}{$h$}
\roundNbox{fill=white}{(.5,0)}{.3}{0}{0}{$f$}
}
\,\right\rangle
$$
so that $\Tube_A^{\loc}(\cY):=\Tube_A(\cY)/\mathcal{I}_A(\cY)$.
Observe that the ideal $\cI_A(\cY)$ acts as zero in \eqref{eq:CondensationExcitation} as the $D_{p,q}$ terms in the Hamiltonian are applied on the six edges along the boundary of the 2 plaquettes where the excitation lives, showing that the $\Tube_A(\cY)$ action really descends to $\Tube_A^{\loc}(\cY)$.
Alternatively, 
$\Tube_A^{\loc}(\cY)=\Tube_A(\cY)p_A$ for the central idempotent
$$
p_A:=
\tikzmath{
\draw[thick] (.5,-.6) -- (.5,1);
\draw[thick] (-.5,-.6) -- (-.5,1);
\draw[thick] (0,1) ellipse (.5 and .2);
\draw[thick] (.5,-.6) arc(0:-180:.5 and .2);
\draw[thick,dotted] (.5,-.6) arc(0:180:.5 and .2);
\draw[] (0,-.8) --node[left]{$\scriptstyle Y$} (0,-.3) -- (0,.3) --node[left]{$\scriptstyle Y$} (0,.8);
\draw[thick, blue, knot] (.5,.2) arc (0:-180:.5 and .2);
\draw[thick, blue, dotted] (.5,.2) arc (0:180:.5 and .2);
\draw[thick, blue, knot] (-.25,.02) to[out=135, in=-90] (-.8,1.2);
\fill[blue] (-.25,.02) circle (.05cm);
}\,.
$$

By \cite{MR4642306},
$\Mod(\Tube_A^{\rm loc}(\cY))\cong Z(\cY)_A^{\rm loc}$.
To see this, a module $\cK$ for $\Tube_A^{\loc}(\cY)$ gives a module for $\Tube_A(\cY)$ by precomposing with the canonical surjection 
$\Tube_A(\cY) \to \Tube_A(\cY)/\cI_A(\cY)=\Tube_A^{\loc}(\cY)$,
and thus we get an object $K\in Z(\cY)_A$.
However, since $\cI_A(\cY)$ must act as zero on $\cK$ and $A$ has trivial twist as it is condensable, we must have the relation
$$
\tikzmath{
\fill[lightgray]  (1,.8) -- (1,1) arc(0:180:1cm and .5cm) -- (-1,.8) arc(-180:0:1 and .2);
\draw[thick] (1,.8) -- (1,1) arc(0:180:1cm and .5cm) -- (-1,.8);
\draw[thick] (1,-.6) -- (1,.8);
\draw[thick] (-1,-.6) -- (-1,.8);
\draw[thick] (1,.8) arc(0:-180:1 and .2);
\draw[thick,dotted] (1,.8) arc(0:180:1 and .2);
\draw[thick] (1,-.6) arc(0:-180:1 and .2);
\draw[thick,dotted] (1,-.6) arc(0:180:1 and .2);
\node at (0,1.2) {$\eta$};
\draw[thick, orange, snake] (0,1.5) --node[right]{$\scriptstyle K$} (0,2) -- (0,2.4);
\draw[] (-.5,-.78) --node[right]{$\scriptstyle Y$} (-.5,-.3);
\draw[] (-.5,.64) --node[right]{$\scriptstyle Y$} (-.5,.3);
\draw[thick, purple] (-.5,0) --node[below,yshift=.1cm]{$\scriptstyle y$} (.5,0);
\draw[thick, cyan] (-.7,0) arc (260:240:1);
\draw[thick, cyan] (.7,0) arc (-80:-60:1);
\draw[thick, cyan, dotted] (1,.1) arc (0:180:1 and .2);
\node[cyan] at (.9,-.1) {$\scriptstyle x$};
\node[cyan] at (-.9,-.1) {$\scriptstyle x$};
\draw[thick, blue, knot] (-.5,0) -- (-.8,.3) to[out=135, in=-90] (-1.2,.8) to[out=90,in=-135] (0,2);
\draw[thick, blue, knot] (.5,0) -- (.2,.3) to[out=135, in=-45] (-1.2,.8);
\filldraw[blue] (-1.2,.8) circle (.05cm);
\filldraw[blue] (0,2) circle (.05cm);
\roundNbox{fill=white}{(-.5,0)}{.3}{0}{0}{$f$}
\roundNbox{fill=white}{(.5,0)}{.3}{0}{0}{$h$}
}
\underset{(\cI_A(\cY))}{=}
\tikzmath{
\fill[lightgray]  (1,.8) -- (1,1) arc(0:180:1cm and .5cm) -- (-1,.8) arc(-180:0:1 and .2);
\draw[thick] (1,.8) -- (1,1) arc(0:180:1cm and .5cm) -- (-1,.8);
\draw[thick] (1,-.6) -- (1,.8);
\draw[thick] (-1,-.6) -- (-1,.8);
\draw[thick] (1,.8) arc(0:-180:1 and .2);
\draw[thick,dotted] (1,.8) arc(0:180:1 and .2);
\draw[thick] (1,-.6) arc(0:-180:1 and .2);
\draw[thick,dotted] (1,-.6) arc(0:180:1 and .2);
\node at (0,1.2) {$\eta$};
\draw[thick, orange, snake] (0,1.5) --node[right]{$\scriptstyle K$} (0,2) -- (0,2.4);
\draw[] (.5,-.78) --node[left]{$\scriptstyle Y$} (.5,-.3);
\draw[] (.5,.64) --node[right]{$\scriptstyle Y$} (.5,.3);
\draw[thick, cyan] (-.5,0) --node[below,yshift=.1cm]{$\scriptstyle x$} (.5,0);
\draw[thick, purple] (-.7,0) arc (260:240:1);
\draw[thick, purple] (.7,0) arc (-80:-60:1);
\draw[thick, purple, dotted] (1,.1) arc (0:180:1 and .2);
\node[purple] at (.9,-.1) {$\scriptstyle y$};
\node[purple] at (-.9,-.1) {$\scriptstyle y$};
\draw[thick, blue, knot] (-.5,0) -- (-.8,.3) to[out=135, in=-90] (-1.2,.8) to[out=90,in=-135] (0,2);
\draw[thick, blue, knot] (.5,0) -- (.2,.3) to[out=135, in=-45] (-1.2,.8);
\filldraw[blue] (-1.2,.8) circle (.05cm);
\filldraw[blue] (0,2) circle (.05cm);
\roundNbox{fill=white}{(-.5,0)}{.3}{0}{0}{$h$}
\roundNbox{fill=white}{(.5,0)}{.3}{0}{0}{$f$}
}
\underset{\left(
\tikzmath{\draw[thin, densely dotted, <-] (-40:.3 and .15) arc(-40:220:.3 and .15) node[right, yshift=.05cm,xshift=-.1cm]{\tiny $h$};}
\right)
}{=}
\tikzmath{
\fill[lightgray]  (1,.8) -- (1,1) arc(0:180:1cm and .5cm) -- (-1,.8) arc(-180:0:1 and .2);
\draw[thick] (1,.8) -- (1,1) arc(0:180:1cm and .5cm) -- (-1,.8);
\draw[thick] (1,-.6) -- (1,.8);
\draw[thick] (-1,-.6) -- (-1,.8);
\draw[thick] (1,.8) arc(0:-180:1 and .2);
\draw[thick,dotted] (1,.8) arc(0:180:1 and .2);
\draw[thick] (1,-.6) arc(0:-180:1 and .2);
\draw[thick,dotted] (1,-.6) arc(0:180:1 and .2);
\node at (0,1.2) {$\eta$};
\draw[] (-.5,-.78) --node[right]{$\scriptstyle Y$} (-.5,-.3);
\draw[] (-.5,.64) --node[right]{$\scriptstyle Y$} (-.5,.3);
\draw[thick, purple] (-.5,0) --node[below,yshift=.1cm]{$\scriptstyle y$} (.5,0);
\draw[thick, cyan] (-.7,0) arc (260:240:1);
\draw[thick, cyan] (.7,0) arc (-80:-60:1);
\draw[thick, cyan, dotted] (1,.1) arc (0:180:1 and .2);
\node[cyan] at (.9,-.1) {$\scriptstyle x$};
\node[cyan] at (-.9,-.1) {$\scriptstyle x$};
\draw[thick, blue, knot] (.5,0) -- (.2,.3) to[out=135, in=-90] (.2,1.5) to[out=90,in=-45] (-.2,1.8);
\draw[thick, orange, snake, knot] (0,1.53) -- (0,2) --node[right]{$\scriptstyle K$} (0,2.4);
\draw[thick, blue, knot] (-.5,0) -- (-.8,.3) to[out=135, in=-90] (-1.2,.8) to[out=90,in=-135] (0,2);
\filldraw[blue] (-.22,1.82) circle (.05cm);
\filldraw[blue] (0,2) circle (.05cm);
\roundNbox{fill=white}{(-.5,0)}{.3}{0}{0}{$f$}
\roundNbox{fill=white}{(.5,0)}{.3}{0}{0}{$h$}
}
$$
By a Yoneda Lemma style argument, we infer from the above relation that $K\in Z(\cY)_A^{\loc}$.

We can summarize the above conceptually by another equivalent definition of $\Tube_A(\cY)$, parallel to \eqref{eq:tubeViaTrace}, by observing that the adjunction $F\dashv \Tr$ yields the following natural isomorphisms:
\begin{align}
\Tube_A(\cY) &:=
 \bigoplus_{y\in\Irr(\cY)}\cY(y\otimes Y\to F(A)\otimes Y\otimes y) 
 \notag
 \\&\cong
 \cY(F\Tr(Y)\to F(A)\otimes Y)
 \notag
 \\&
 \cong Z(\cY)(\Tr(Y)\to \Tr(F(A)Y))
  \label{eq:tubeAViaTrace}
\\
 &\cong Z(\cY)(\Tr(Y)\to A\Tr(Y)) && \text{(\cite[Lem.~4.11]{MR3578212})}
  \notag
 \\
 &\cong Z(\cY)_A(\Tr(Y)A\to \Tr(Y)A)
 && \text{(\cite[Fig.~4]{MR1936496})}
 \notag
\end{align}
One checks that the algebra structure on $\Tube_A(\cY)$ is the one obtained by transporting the algebra structure on $\End_{Z(\cY)_A}(\Tr(Y)A)$ across the equivalences \eqref{eq:tubeAViaTrace}.
Finally, the fully faithful monoidal inclusion $Z(\cY)_A^{\loc}\hookrightarrow Z(\cY)_A$ has a $2$-sided adjoint $\ell$, which restricts to the identity functor on the full subcategory $Z(\cY)_A^{\loc}$ and sends non-local simple modules to $0$.
One computes that the corner $\Tube_A^{\loc}(\cY)\subseteq\Tube_A(\cY)$ corresponds to the corner
$\End_{Z(\cY)_A^{\loc}}(\operatorname{Loc}(\Tr(Y)A))\subseteq \End_{Z(\cY)}(\Tr(Y)A)$ under the equivalences \eqref{eq:tubeAViaTrace}, where $\operatorname{Loc}$ is the functor $Z(\cY)_A\to Z(\cY)_A^{\loc}$ which takes the local part, sending nonlocal simple modules to $0$.

We end this section with the following facts which we record for later use.

\begin{lem}
\label{lem:FreeModuleFunctorViaTubes}
The inclusion map $i_A\otimes-:\Tube(\cY)\hookrightarrow \Tube_A(\cY)$ implements the free module functor.
That is, if $q_z\in \Tube(\cY)$ is a minimal projection corresponding to the simple object $z\in Z(\cY)$, then $i_A\otimes q_z$ is a projection in $\Tube_A(\cY)$, and the number of copies of the simple module $M\in Z(\cY)_A$ in $Az$ is given by
\[
\dim\left((i_A\otimes q_z)\rhd\Tube_A(\cY)\lhd p_M\right)
\]
where $p_M\in \Tube_A(\cY)$ is a minimal projection corresponding to $M$.
\end{lem}
\begin{rem}
In the graphical calculus of strings on tubes, 
$
i_A\otimes q_z=\tikzmath{
\draw[thick] (.5,-.6) -- (.5,1);
\draw[thick] (-.5,-.6) -- (-.5,1);
\draw[thick] (0,1) ellipse (.5 and .2);
\draw[thick] (.5,-.6) arc(0:-180:.5 and .2);
\draw[thick,dotted] (.5,-.6) arc(0:180:.5 and .2);
\draw[] (0,-.8) --node[left]{$\scriptstyle Y$} (0,-.3);
\draw[] (0,-.3) -- (0,.3);
\draw[] (0,.8) --node[left]{$\scriptstyle Y$} (0,.3); 
\draw[thick, blue] (-.8,.6) -- (-.8,1.2);
\fill[blue] (-.8,.6) circle (.05cm);
\draw[thick, cyan] (.5,.2) arc (0:-180:.5 and .2);
\draw[thick, cyan, dotted] (.5,.2) arc (0:180:.5 and .2);
\node[cyan] at (.4,-.1) {$\scriptstyle y$};
\node[cyan] at (-.4,-.1) {$\scriptstyle y$};
\roundNbox{fill=white}{(0,0)}{.3}{0}{0}{$q_z$}
}
$\,.
\end{rem}
\begin{proof}
By \cite[\S2.3 and 3.4.3]{MR4642306}, the canonical equivalence $\Mod(\Tube_A(\cY))\cong Z(\cY)_A$ sends $(i_A\otimes q_z)\Tube_A(\cY)$ to $Az$, and thus we have an equivalence of right $\Tube_A(\cY)$ modules
\begin{align*}
(i_A\otimes q_z)\rhd\Tube_A(\cY)
&\underset{\text{\eqref{eq:tubeAViaTrace}}}{\cong}
(1_A\otimes q_z) Z(\cY)(\Tr(Y)\to A\Tr(Y))
\\&\cong
Z(\cY)(\Tr(Y)\to Az)
\\&\cong
Z(\cY)_A(A\Tr(Y)\to Az).
\end{align*}
Now we precompose with $p_M\in \Tube_A(\cY)\cong Z(\cY)_A(A\Tr(Y)\to A\Tr(Y))$ to see
\[
(i_A\otimes q_z)\rhd\Tube_A(\cY)\lhd p_M
\cong
Z(\cY)_A(M\to Az)
=
\Hom_{\Mod(\Tube_A(\cY))}(M\to Az)
.
\qedhere
\]
\end{proof}

We get the following immediate corollaries which will be helpful in our computations in \S\ref{sssec:chiralExample} below.

\begin{cor}
    \label{cor:freeModuleProjRanks}
    For any simple $M\in Z(\cY)_A$,
    \[\operatorname{rank}(P_M(i_A\otimes q_z))=\dim(Z(\cY)_A(M\to Az))\]
\end{cor}
\begin{proof}
    Both are equal to $\dim\left((i_A\otimes q_z)\rhd\Tube_A(\cY)\lhd p_M\right)$.
\end{proof}

\begin{cor}
\label{cor:freeModuleRank1}
If $z\in \Irr(Z(\cY))$ such that $Az\in \Irr(Z(\cY)_A)$ is simple,
then for a rank one idempotent $q_z\in P_z\Tube(\cY)$, $i_A\otimes q_z\in \Tube_A(\cY)$ has rank one.
\end{cor}

%%%%%%%%%%%%%%%%%%%%%%%%%%%%%%%%%%%%%%%%%%%%%%%%%
\subsection{Defects}
\label{sec:CondensationDefects}

As in \S\ref{sec:GeneralizedTwistDefects} above, we can modify our $A$-condensation Hamiltonian along a 1D defect line as follows.
If $\ell$ is an edge in the 1D defect connecting vertices $u,v$, we replace the $D_{u,v}$ term \eqref{eq:MultiplicationTerm} with unit terms $C_u,C_v$ \eqref{eq:UnitTerm}.
Passing into the ground state space of these $C_u$ terms effectively removes the blue $A$-edges from the 1D defect line.
$$
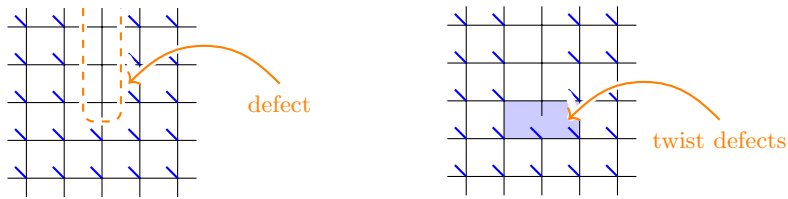
\tikzmath{
\draw[step=.5] (-.25,-.25) grid (2.25,2.25);
\foreach \x in {0,1,3,4}{
\foreach \y in {0,1,2,3,4}{
\draw[thick, blue] ($ .5*(\x,\y) $) -- ($ .5*(\x,\y) + (-.15,.15)$);
}}
\draw[thick, blue] ($ .5*(2,0) $) -- ($ .5*(2,0) + (-.15,.15)$);
\draw[thick, blue] ($ .5*(2,1) $) -- ($ .5*(2,1) + (-.15,.15)$);
\draw[thick, orange, dashed, knot, rounded corners] (.75,2.25) -- (.75,.75) -- (1.25,.75) -- (1.25,2.25);
\draw[thick, orange, knot, ->] (3.35,1.25) node[below]{\scriptsize defect} to[out=135,in=0] (2.35,1.75) to[out=180,in=45] (1.35,1.25); 
%\draw[thick, orange] (1,2.25) -- (1,1);
%\filldraw[orange] (1,1) circle (.05cm);
}
\qquad\qquad
\tikzmath{
\fill[blue!20] (.5,.5) rectangle (1.5,1);
\draw[step=.5] (-.25,-.25) grid (2.25,2.25);
\fill[blue!20] (.98,.52) rectangle (1.02,.8);
\foreach \x in {0,1,3,4}{
\foreach \y in {0,1,2,3,4}{
\draw[thick, blue] ($ .5*(\x,\y) $) -- ($ .5*(\x,\y) + (-.15,.15)$);
}}
\draw[thick, blue] ($ .5*(2,0) $) -- ($ .5*(2,0) + (-.15,.15)$);
\draw[thick, blue] ($ .5*(2,1) $) -- ($ .5*(2,1) + (-.15,.15)$);
%\draw[thick, orange, dashed, knot, rounded corners] (.75,2.25) -- (.75,.75) -- (1.25,.75) -- (1.25,2.25);
\draw[thick, orange, knot, ->] (3.35,.75) node[below]{\scriptsize twist defects} to[out=135,in=0] (2.35,1.25) to[out=180,in=45] (1.35,.75); 
}
$$
We can measure twist defects via an action of $\Tube_A(\cY)$ (not $\Tube_A^{\loc}(\cY)$!) by modifying \eqref{eq:CondensationExcitation} as follows:
$$
\tikzmath{
\foreach \x in {0,1,2}{
\foreach \y in {0,1}{
\draw ($ (-.7,0) + 1.4*(\x,\y) $) -- ($ (-.3,0) + 1.4*(\x,\y) $);
\draw ($ (.7,0) + 1.4*(\x,\y) $) -- ($ (.3,0) + 1.4*(\x,\y) $);
\draw ($ (0,-.7) + 1.4*(\x,\y) $) -- ($ (0,-.3) + 1.4*(\x,\y) $);
\draw ($ (0,.7) + 1.4*(\x,\y) $) -- ($ (0,.3) + 1.4*(\x,\y) $);
}}
\fill[white] ($ 1.4*(.9,.23) $) rectangle ($ 1.4*(1.1,.65) $);
\filldraw[thick,cyan,fill=blue!30, rounded corners=3pt] (.7,.1) -- (.5,.1) -- (.1,.5) -- (.1,.9) -- (.5,1.3) -- (.9,1.3) -- (1.3,.9) -- (1.5,.9) -- (1.9,1.3) -- (2.3,1.3) -- (2.7,.9) -- (2.7,.5) -- (2.3,.1) -- (1.9,.1) -- (1.5,.5) -- (1.3,.5) -- (.9,.1) -- (.7,.1);
\draw (1.4,.9) -- (1.4,.7);
\foreach \x in {0,2}{
\foreach \y in {0,1}{
\draw[thick, blue, knot] ($ 1.4*(\x,\y) $) -- ($ 1.4*(\x,\y) + (-.45,.7) $);
\roundNbox{fill=white}{($ 1.4*(\x,\y)$)}{.3}{0}{0}{$f_{\x\y}$}
}}
\draw[thick, blue, knot] (1.4,.9) to[out=135,in=-90] ($ (0,1.4) + (-.45,.7) $);
\fill[blue] (1.4,.9) circle (.05cm);
\fill[blue] ($ (0,1.4) + (-.45,.7) $) circle (.05cm);
\draw[thick, blue] ($ (0,1.4) + (-.45,.7) $) -- ($ (0,1.4) + (-.45,1) $);
\draw[thick, blue, knot] ($ 1.4*(1,0) $) -- ($ 1.4*(1,0) + (-.45,.7) $);
\roundNbox{fill=white}{($ 1.4*(1,0)$)}{.3}{0}{0}{$f_{10}$}
\roundNbox{fill=white}{($ 1.4*(1,1)$)}{.3}{0}{0}{$f_{11}$}
}
$$
The twist defects in the condensation model correspond to the relative center $Z_{\cY_A}(\cY)$:
\begin{align*}
\Rep(\Tube_A(\cY))
&\cong
Z(\cY)_A
&&\text{(\cite[\S{3.4.3}]{MR4642306})}
\\&\cong
\End_{\cY-\cY_A}(\cY_A)
&&
\text{(\cite[Ex.~III.2]{MR4640433})}
\\&\cong
\Hom_{\cY_A-\cY_A}(\cY_A\to \cY_A\boxtimes_\cY\cY_A) 
&& 
\text{(dualizability in $\UmFC$)}
\\&=:
Z_{\cY_A}({}_A\cY_A).
\end{align*}
(The second equivalence depends on results from \cite[\S{3}]{MR3039775}.)
Fusion for twist defects, which is the monoidal product on $\End_{\cY-\cY_A}(\cY_A)\cong Z(\cY)_A$,
is defined analogously to \S\ref{sec:FusionOfTwistDefects}.

Arguing as in \S\ref{ssec:indRes}, we may compute the non-invertible symmetry action of sweeping a defect across a anyon in the condensed theory by computing induction and restriction.
The correct functors were described in the discussion at the start of \cite[\S{V}]{MR4640433}: 
\begin{equation}
\label{eq:condensationIndRes}
Z(\cY_A)
=
Z(\cY)_A^{\loc}
\hookrightarrow
Z(\cY)_A
\twoheadrightarrow
{}_AZ(\cY)_A
\to
Z(\cY)_A
=
Z_{\cY_A}(\cY),
\end{equation}
where the right dominant tensor functor ${}_AZ(\cY)_A\to Z(\cY)_A$ is the one which forgets the left $A$ action, and the left full inclusion $Z(\cY)_A\to {}_A Z(\cY)_A$ is adjoint to the functor which forgets the right $A$-action.
Here, the category ${}_AZ(\cY)_A$ of $A-A$ bimodules arises as $(Z(\cY)_A)\boxtimes_{Z(\cY)}(Z(\cY)_A)$.

As in \S\ref{subsec:domainwallwrap}, each of these categories is realized as the representations of one of our tube algebras, and each functor between representation categories is given by tensoring with a bimodule obtained from the dual inclusions and projections between them, analogous to \eqref{eq:fullInclusionRelativeTensor}:
\[
\Tube_{A}^{\loc}(\cY)\overset{p_A}{\twoheadleftarrow}\Tube_A(\cY)\hookrightarrow\Tube_{AA}(\cY)\hookleftarrow\Tube_A(\cY)\]
In particular, ${}_AZ(\cY)_A$ is the category of representations of
\[
\Tube_A(\cY)\otimes_{\Tube(\cY)}\Tube_A(\cY)\cong\Tube_{AA}(\cY):=
\operatorname{span}\set{
\tikzmath{
\draw[thick] (.5,-.6) -- (.5,1);
\draw[thick] (-.5,-.6) -- (-.5,1);
\draw[thick] (0,1) ellipse (.5 and .2);
\draw[thick] (.5,-.6) arc(0:-180:.5 and .2);
\draw[thick,dotted] (.5,-.6) arc(0:180:.5 and .2);
\draw[] (0,-.8) --node[left]{$\scriptstyle Y$} (0,-.3);
\draw[] (0,.8) --node[left]{$\scriptstyle Y$} (0,.3);
\draw[thick, cyan] (.5,.2) arc (0:-180:.5 and .2);
\draw[thick, cyan, dotted] (.5,.2) arc (0:180:.5 and .2);
\node[cyan] at (-.4,-.1) {$\scriptstyle y$};
\node[cyan] at (.4,-.1) {$\scriptstyle y$};
\draw[thick, blue, knot] (0,0) -- (-.3,.3) to[out=135, in=-90] (-.8,1.2);
\draw[thick, blue, knot] (0,0) -- (-.3,-.2) to[out=-135, in=90] (-.8,-.8);
\roundNbox{fill=white}{(0,0)}{.3}{0}{0}{$f$}
}
}
{y\in\Irr(\cY)},
\]
which has the structure of both a $\Tube_A(\cY)$-bimodule category
as well as a
finite dimensional $\rmC^*$-algebra in its own right, where the multiplication\footnote{While the vector space $\Tube_{AA}(\cY)$ is the underlying space of the \emph{V.~Jones basic construction} \cite{MR0696688} of the inclusion of algebras $\Tube(\cY)\subset \Tube_A(\cY)$ from \eqref{eq:IncludeTubeIntoATube}, this multiplication is not V.~Jones' multiplication.} 
and $*$-structure are given by 
\[
\tikzmath{
\draw[thick] (.5,-.6) -- (.5,1);
\draw[thick] (-.5,-.6) -- (-.5,1);
\draw[thick] (0,1) ellipse (.5 and .2);
\draw[thick] (.5,-.6) arc(0:-180:.5 and .2);
\draw[thick,dotted] (.5,-.6) arc(0:180:.5 and .2);
\draw[] (0,-.8) --node[left]{$\scriptstyle Y$} (0,-.3);
\draw[] (0,.8) --node[left]{$\scriptstyle Y$} (0,.3);
\draw[thick, cyan] (.5,.2) arc (0:-180:.5 and .2);
\draw[thick, cyan, dotted] (.5,.2) arc (0:180:.5 and .2);
\node[cyan] at (.4,-.1) {$\scriptstyle x$};
\node[cyan] at (-.4,-.1) {$\scriptstyle x$};
\draw[thick, blue, knot] (0,0) -- (-.3,.3) to[out=135, in=-90] (-.8,1.2);
\draw[thick, blue, knot] (0,0) -- (-.3,-.2) to[out=-135, in=90] (-.8,-.8);
\roundNbox{fill=white}{(0,0)}{.3}{0}{0}{$f$}
}
\,\,
\cdot
\tikzmath{
\draw[thick] (.5,-.6) -- (.5,1);
\draw[thick] (-.5,-.6) -- (-.5,1);
\draw[thick] (0,1) ellipse (.5 and .2);
\draw[thick] (.5,-.6) arc(0:-180:.5 and .2);
\draw[thick,dotted] (.5,-.6) arc(0:180:.5 and .2);
\draw[] (0,-.8) --node[left]{$\scriptstyle Y$} (0,-.3);
\draw[] (0,.8) --node[left]{$\scriptstyle Y$} (0,.3);
\draw[thick, purple] (.5,.2) arc (0:-180:.5 and .2);
\draw[thick, purple, dotted] (.5,.2) arc (0:180:.5 and .2);
\node[purple] at (.4,-.1) {$\scriptstyle y$};
\node[purple] at (-.4,-.1) {$\scriptstyle y$};
\draw[thick, blue, knot] (0,0) -- (-.3,.3) to[out=135, in=-90] (-.8,1.2);
\draw[thick, blue, knot] (0,0) -- (-.3,-.2) to[out=-135, in=90] (-.8,-.8);
\roundNbox{fill=white}{(0,0)}{.3}{0}{0}{$g$}
}
=
\tikzmath{
\draw[thick] (.5,-.6) -- (.5,2);
\draw[thick] (-.5,-.6) -- (-.5,2);
\draw[thick] (0,2) ellipse (.5 and .2);
\draw[thick] (.5,-.6) arc(0:-180:.5 and .2);
\draw[thick,dotted] (.5,-.6) arc(0:180:.5 and .2);
\draw[] (0,-.8) --node[left]{$\scriptstyle Y$} (0,-.3);
\draw[] (0,.8) --node[left]{$\scriptstyle Y$} (0,.3);
\draw[] (0,1.3) --node[left]{$\scriptstyle Y$} (0,1.8);
\draw[thick, purple] (.5,.2) arc (0:-180:.5 and .2);
\draw[thick, purple, dotted] (.5,.2) arc (0:180:.5 and .2);
\draw[thick, cyan] (.5,1.2) arc (0:-180:.5 and .2);
\draw[thick, cyan, dotted] (.5,1.2) arc (0:180:.5 and .2);
\draw[thick, blue, knot] (0,1) -- (-.3,.8) to[out=-135, in=90] (-.8,.2) -- (-.8,-.8);
\draw[thick, blue, knot] (0,0) -- (-.3,-.2) to[out=-135, in=45] (-.8,-.6);
\filldraw[blue] (-.8,-.6) circle (.05cm);
\draw[thick, blue, knot] (0,0) -- (-.3,.3) to[out=135, in=-90] (-.8,1.2) -- (-.8,1.8);
\draw[thick, blue, knot] (0,1) -- (-.3,1.3) to[out=135, in=-45] (-.8,1.8) -- (-.8,2.2);
\filldraw[blue] (-.8,1.8) circle (.05cm);
\roundNbox{fill=white}{(0,0)}{.3}{0}{0}{$g$}
\roundNbox{fill=white}{(0,1)}{.3}{0}{0}{$f$}
}
=
\sum_{z\in\Irr(\cY)}
\tikzmath{
\draw[thick] (1,-.6) -- (1,2);
\draw[thick] (-1,-.6) -- (-1,2);
\draw[thick] (0,2) ellipse (1 and .2);
\draw[thick] (1,-.6) arc(0:-180:1 and .2);
\draw[thick,dotted] (1,-.6) arc(0:180:1 and .2);
\draw[] (0,-.8) --node[left]{$\scriptstyle Y$} (0,-.3);
\draw[] (0,.8) --node[left]{$\scriptstyle Y$} (0,.3);
\draw[] (0,1.3) --node[left]{$\scriptstyle Y$} (0,1.8);
\draw[thick, purple] (.3,0) to[out=0,in=-135] node[right]{$\scriptstyle y$} (.7,.5);
\draw[thick, purple] (-.3,0) to[out=180,in=-45] node[left]{$\scriptstyle y$} (-.7,.5) ;
\draw[thick, cyan] (.3,1)  to[out=0,in=135] node[right]{$\scriptstyle x$} (.7,.5);
\draw[thick, cyan] (-.3,1) to[out=180,in=45] node[left]{$\scriptstyle x$} (-.7,.5);
\draw[thick, green] (-.7,.5) arc (260:240:1);
\draw[thick, green] (.7,.5) arc (-80:-60:1);
\draw[thick, green, dotted] (1,.6) arc (0:180:1 and .2);
\node[green] at (.9,.4) {$\scriptstyle z$};
\node[green] at (-.9,.4) {$\scriptstyle z$};
\filldraw[fill=yellow] (-.7,.5) circle (.05cm);
\filldraw[fill=yellow] (.7,.5) circle (.05cm);
\draw[thick, blue, knot] (0,1) -- (-.3,.8) to[out=-135, in=90] (-1.2,.2) -- (-1.2,-.8);
\draw[thick, blue, knot] (0,0) -- (-.3,-.2) to[out=-135, in=45] (-1.2,-.6);
\filldraw[blue] (-1.2,-.6) circle (.05cm);
\draw[thick, blue, knot] (0,0) -- (-.3,.3) to[out=135, in=-90] (-1.2,1.2) -- (-1.2,1.8);
\draw[thick, blue, knot] (0,1) -- (-.3,1.3) to[out=135, in=-45] (-1.2,1.8) -- (-1.2,2.2);
\filldraw[blue] (-1.2,1.8) circle (.05cm);
\roundNbox{fill=white}{(0,0)}{.3}{0}{0}{$g$}
\roundNbox{fill=white}{(0,1)}{.3}{0}{0}{$f$}
}
\qquad\qquad
\tikzmath{
\draw[thick] (.5,-.6) -- (.5,1);
\draw[thick] (-.5,-.6) -- (-.5,1);
\draw[thick] (0,1) ellipse (.5 and .2);
\draw[thick] (.5,-.6) arc(0:-180:.5 and .2);
\draw[thick,dotted] (.5,-.6) arc(0:180:.5 and .2);
\draw[] (0,-.8) --node[left]{$\scriptstyle Y$} (0,-.3);
\draw[] (0,.8) --node[left]{$\scriptstyle Y$} (0,.3);
\draw[thick, cyan] (.5,.2) arc (0:-180:.5 and .2);
\draw[thick, cyan, dotted] (.5,.2) arc (0:180:.5 and .2);
\node[cyan] at (-.4,-.1) {$\scriptstyle x$};
\node[cyan] at (.4,-.1) {$\scriptstyle x$};
\draw[thick, blue, knot] (0,0) -- (-.3,.3) to[out=135, in=-90] (-.8,1.2);
\draw[thick, blue, knot] (0,0) -- (-.3,-.2) to[out=-135, in=90] (-.8,-.8);
\roundNbox{fill=white}{(0,0)}{.3}{0}{0}{$f$}
}
^\dag
\,\,
=
\tikzmath{
\draw[thick] (.5,-.6) -- (.5,1);
\draw[thick] (-.5,-.6) -- (-.5,1);
\draw[thick] (0,1) ellipse (.5 and .2);
\draw[thick] (.5,-.6) arc(0:-180:.5 and .2);
\draw[thick,dotted] (.5,-.6) arc(0:180:.5 and .2);
\draw[] (0,-.8) --node[left]{$\scriptstyle Y$} (0,-.3);
\draw[] (0,.8) --node[left]{$\scriptstyle Y$} (0,.3);
\draw[thick, cyan] (.5,.2) arc (0:-180:.5 and .2);
\draw[thick, cyan, dotted] (.5,.2) arc (0:180:.5 and .2);
\node[cyan] at (-.4,-.1) {$\scriptstyle \overline{x}$};
\node[cyan] at (.4,-.1) {$\scriptstyle \overline{x}$};
\draw[thick, blue, knot] (0,0) -- (-.3,.3) to[out=135, in=90] (-.8,-.8);
\draw[thick, blue, knot] (0,0) -- (-.3,-.3) to[out=-135, in=-90] (-.8,1.2);
\roundNbox{fill=white}{(0,0)}{.3}{0}{0}{$f^\dag$}
}\,.
\]
The left and right actions of $\Tube_A(\cY)$ are given by the inclusions
\[
\Tube_A(\cY)\cong\Tube_{A1}(\cY)\hookrightarrow\Tube_{AA}(\cY)\hookleftarrow\Tube_{1A}(\cY)\cong\Tube_A(\cY).
\]
Graphically, these inclusions resemble \eqref{eq:IncludeTubeIntoATube}, using the unit to add the missing blue string.
The right (left) action of $\Tube_A(\cY)$ therefore involves stacking the $\Tube_A(\cY)$ element on the top (bottom) and joining its \textcolor{blue}{$A$} string with the top (bottom) \textcolor{blue}{$A$} string of the $\Tube_{AA}(\cY)$ element.

We can use these stacking operations to compute the symmetry action, analogous to \eqref{eq:fullInclusionTubeStack}.
An anyon type $a\in\Irr(Z(\cY_A))$ corresponds to a minimal central idempotent $P_a\in\Tube_A^{\loc}(\cY)$, and a twist defect of type $b\in\Irr(Z(\cY)_A)$ corresponds to a minimal central idempotent $P_b\in\Tube_A(\cY)$.
As before, we may choose rank $1$ projectors $p_a\in\Tube_A^{\loc}(\cY)P_a$ and $p_b\in P_b\Tube_A(\cY)$.
We may then compute 
the multiplicities for the symmetry action as we did in \eqref{eq:symMultRank1}:
\[
N_{a,b}=\dim(p_b\Tube_{AA}(\cY) p_a).
\]
Alternatively, we have
\[N_{a,b}=\frac{\dim(P_b\Tube_{AA}(\cY)P_a)}{\sqrt{\dim(P_b(\Tube_A(\cY))}\sqrt{\dim(\Tube_A^{\loc}(\cY))P_a)}}\]
as in \eqref{eq:symMultCentral}.

The following corollary is similar to Corollary \ref{cor:freeModuleRank1} using Lemma \ref{lem:FreeModuleFunctorViaTubes}.

\begin{cor}
\label{cor:freeModuleMult}
If $x,y\in \Irr(Z(\cY))$ such that the free modules $Ax, Ay\in Z(\cY)_A$ are again simple in $\Irr(Z(\cY)_A)$, then 
$$
p_y\Tube_{AA}(\cY) p_x \cong  Z(\cY)(Ax\to Ay).
$$
\end{cor}

%%%%%%%%%%%%%%%%%%%%%%%%%%%%%%%%%%%%%%%
\subsubsection{Individual symmetries in the condensation model}
As discussed in \S\ref{sssec:individualFull}, we view the simple $\cY_A-\cY_A$ bimodule summands of ${}_A\cY_A$ as playing the role of individual symmetry sectors in our overall non-invertible symmetry action.
By \cite[Thm.~III.22]{MR4640433}, indecomposable $\cY_A-\cY_A$ bimodule summands of ${}_A\cY_A$ are in bijection with minimal idempotents in the abelian algebra $(Z(\cY)(A\to A),\ast)$, where $\ast$ is the convolution operation of \cite[Defn.~III.14]{MR4640433}:
\begin{equation}
\label{eq:ConvolutionMultiplication}
\tikzmath{
\draw[thick, blue] (0,.3) -- (0,1);
\draw[thick, blue] (0,-.3) -- (0,-1);
\roundNbox{}{(0,0)}{.3}{0}{0}{$f$}
}
*
\tikzmath{
\draw[thick, blue] (0,.3) -- (0,1);
\draw[thick, blue] (0,-.3) -- (0,-1);
\roundNbox{}{(0,0)}{.3}{0}{0}{$g$}
}
:=
\tikzmath{
\draw[thick, blue] (0,.7) -- (0,1);
\draw[thick, blue] (0,-.7) -- (0,-1);
\draw[thick, blue] (-.4,.3) arc(180:0:.4cm);
\draw[thick, blue] (-.4,-.3) arc(-180:0:.4cm);
\roundNbox{}{(-.4,0)}{.3}{0}{0}{$f$}
\roundNbox{}{(.4,0)}{.3}{0}{0}{$g$}
}
\end{equation}

In particular, as discussed in \cite[\S{V}]{MR4640433}, the identity morphism $\id_A\in Z(\cY)(A\to A)$ becomes a \textit{minimal} central idempotent for the $\ast$ product, corresponding to the trivial summand $\cY_A\subseteq{}_A\cY_A$.

As in \S\ref{sssec:individualFull}, decompose 
\[{}_A\cY_A\cong\bigoplus_{j=0}^{n-1}\cM_j\]
where each $\cM_j$ is an indecomposable $\cY_A-\cY_A$ bimodule category, and let $\omega_j\in Z(\cY)(A\to A)$ be the central projector corresponding to $\cM_j$, where $\omega_0=\id_A$ so that $\cM_0\cong\cY_A$.
One may check that this decomposition corresponds to the decomposition of $\Tube_{AA}(\cY)$ as a direct sum of $\Tube_A(\cY)-\Tube_A(\cY)$ bimodules
\[
\Tube_{AA}(\cY)=\bigoplus_{j=0}^{n-1}\Tube_{\cM_j}(\cY)
\]
where
\[
\Tube_{\cM_j}(\cY):=
\operatorname{span}
\set{
\tikzmath{
\draw[thick] (.5,-.6) -- (.5,1);
\draw[thick] (-.5,-.6) -- (-.5,1);
\draw[thick] (0,1) ellipse (.5 and .2);
\draw[thick] (.5,-.6) arc(0:-180:.5 and .2);
\draw[thick,dotted] (.5,-.6) arc(0:180:.5 and .2);
\draw[] (0,-.8) --node[left]{$\scriptstyle Y$} (0,-.3);
\draw[] (0,.8) --node[left]{$\scriptstyle Y$} (0,.3);
\draw[thick, cyan] (.5,.2) arc (0:-180:.5 and .2);
\draw[thick, cyan, dotted] (.5,.2) arc (0:180:.5 and .2);
\node[cyan] at (-.4,-.1) {$\scriptstyle x$};
\node[cyan] at (.4,-.1) {$\scriptstyle x$};
\draw[thick, blue, knot] (0,0) -- (-.3,.3) to[out=135, in=90] (-1,.3);
\draw[thick, blue, knot] (0,0) -- (-.3,-.3) to[out=-135, in=-90] (-1,-.3);
\draw[thick, blue] (-.75,.47) -- (-.75,1.2);
\draw[thick, blue] (-.75,-.47) -- (-.75,-.8);
\roundNbox{fill=white}{(-1,-0)}{.3}{0}{0}{$\omega_j$}
\roundNbox{fill=white}{(0,0)}{.3}{0}{0}{$f$}
}
}{f\in\Tube_{AA}(\cY)}.
\]
Similar to \eqref{eq:indivSymMultRank1}, we may calculate how wrapping the domain wall $\cM_j$ about an anyon $a\in Z(\cY)_A^{\loc}$ decomposes into anyons 
in $Z_{\cY_A}(\cM_0)=Z(\cY_A)\cong Z(\cY)_A^{\loc}$
and defects
in $Z_{\cY_A}(\cM_k)$ for $k>0$.
For rank $1$ projectors $p_a\in\Tube_A^{\loc}(\cY)P_A$ and $p_b\in P_b\Tube_A(\cY)$, we compute that 
\[a\mapsto\bigoplus_k N^j_{a,b_k}b_k
\qquad\text{ with }
\qquad
N^j_{a,b_k}=\dim(p_{b_k}\Tube_{\cM_j}(\cY)p_a).
\]
Alternatively, we may use the central projections and compute
\begin{align}\label{eq:dimensioncount_condensed}
N^j_{a,b_k}=\frac{\dim(P_{b_k}\Tube_{\cM_j}(\cY)P_a)}{\sqrt{\dim(P_{b_k}\Tube_A(\cY))}\sqrt{\dim(\Tube_A^{\loc}(\cY)P_a)}}.
\end{align}

\begin{rem}
 \label{rem:cheatingConvolution}
 In the remaining examples, we will not need to explicitly diagonalize the algebra $(Z(\cY)(A\to A),\ast)$, because these examples will only have a single nontrivial symmetry sector, {i.e.} ${}_A\cY_A\cong\cY_A\oplus\cM$ where $\cM$ is an indecomposable $\cY_A-\cY_A$ bimodule category.
 Since we know the action of the trivial symmetry and can compute the full symmetry action, we can work out the action of the nontrivial symmetry sector by direct subtraction.
 When working with examples involving multiple nontrivial symmetry sectors, diagonalizing the convolution multiplication and stacking with the bimodules $\Tube_{\cM_j}(\cY)$ will be much harder to avoid.
\end{rem}

%%%%%%%%%%%%%%%%%%%%%%%%%%%%%%%%%%%%%%%%%%%%%%%%%%%%%%%%%%%
\subsubsection{Example: redoing the \texorpdfstring{$S_3$}{S3} example via condensation}

Here we redo the example based on $\cY = \Hilb(S_3)$ and $\cX=\Hilb(\mathbb{Z}_2)$ from \S\ref{subsec:beyondG} in the condensation model. The condensable algebra $A \in D(S_3)$ is then $A = (1,1) + (1,\rho)$, with $\rho$ being the 2d irrep of $S_3$.
Elements in $\Tube_A(\cY)$ are given by
$$
    T_{g,h,v} = 
     \tikzmath[scale=0.8]{
    \draw[thick] (.5,-.6) -- (.5,1);
    \draw[thick] (-.5,-.6) -- (-.5,1);
    \draw[thick] (0,1) ellipse (.5 and .2);
    \draw[thick] (.5,-.6) arc(0:-180:.5 and .2);
    \draw[thick, dotted] (.5,-.6) arc(0:180:.5 and .2);
    \draw[thick, DarkGreen,
  postaction={decorate},decoration={markings, mark=at position 0.7 with {\arrow{>}}  }] (-.5,.2) arc (-180:-90:.5 and .2);
  \draw[thick, DarkGreen,
  postaction={decorate},decoration={markings, mark=at position 0.7 with {\arrow{>}}  }] (0.,0.) arc (-90:0:.5 and .2);
    \node[DarkGreen, left] at (-0.5,.2) {$\scriptstyle g$};
    \draw[thick, violet, postaction={decorate},decoration={markings, mark=at position 0.8 with {\arrow{>}}}] (0,0) -- (0,0.8);
    \node[violet, right] at (0.5,0.6) {$\scriptstyle g^{-1} h g$};
    \draw[thick, gray,postaction={decorate},decoration={markings, mark=at position 0.4 with {\arrow{>}}}] (0,-0.8) -- (0,-0.);
    \node[gray, right] at (0.5,-0.3) {$\scriptstyle h$};
    \draw[thick, DarkGreen, dotted] (.5,.2) arc (0:180:.5 and .2);
    \draw[thick, blue, knot] (0,0) -- (-.3,.3) to[out=135, in=-90] (-.8,1.2) node[left] {$\scriptstyle A$};
    \fill[black] (0,0) circle (0.8mm);
    \node[right] at (0.5,0.2) {$\scriptstyle v$};
} \in \Tube_A(\cY) \quad \text{with } g,h \in S_3 \text{ and } v \in \{ \phi_0, \phi_1, \phi_{2} \}.
$$
Because $F(A) \cong 3 \cdot 1_\cY$, the black dot in the center of the tube is labeled by a basis vector of this three-dimensional space. We choose a basis with $\phi_j$, on which the $S_3$ elements act as
$$
r^j \ket{\phi_k} = \ket{\phi_{j+k}}, \quad s \ket{\phi_k} = \ket{\phi_{-k}}, \quad j, k \in \{0,1,2\} \mod 3.
$$
The multiplication of two tubes is given by
\begin{align*}
    T_{g,h,\phi_j} \cdot T_{g',h',\phi_{k}} = \delta_{h'=g^{-1}hg}\delta_{\phi_j=g'(\phi_k)}T_{gg',h,\phi_k}
\end{align*}
Here we used the half-braiding of $g'$ with $A$. 

For the toric code idempotents, we take the vertical strands to be $F(M)$, where $M$ is a local $A$-module. For the toric code anyons $1$ and $e$, this means the vertical strand is labeled by $1$. We claim that the corresponding minimal central idempotents are given by
\begin{align*}
    P_1 &= \frac{1}{2} \sum_{k} \sum_{g \in \Stab(\phi_k)} T_{g,1,\phi_k} = \frac{1}{2} \left( T_{1,1,\phi_0} + T_{s,1,\phi_0} + T_{1,1,\phi_1} + T_{r^2s,1,\phi_1} + T_{1,1,\phi_2} + T_{rs,1,\phi_2} \right) \\ 
    P_e &= \frac{1}{2} \sum_{k} \sum_{g \in \Stab(\phi_k)} \hspace{-2mm} \sigma(g) T_{g,1,\phi_k} =  \frac{1}{2} \left( T_{1,1,\phi_0} - T_{s,1,\phi_0} + T_{1,1,\phi_1} - T_{r^2s,1,\phi_1} + T_{1,1,\phi_2} - T_{rs,1,\phi_2} \right), 
\end{align*}
where $\sigma(g)$ denotes the sign irrep of $S_3$. 
The restriction to elements $g \in \Stab(\phi_k)$, defined as 
$$
    \Stab(\phi_k):=\{g\in S_3:\; g\phi_k=\phi_k\},
$$
is forced by centrality.

For the anyons $m$ and $f$, we expect the vertical strand to be labeled by group elements $\{s,rs, r^2s\}$. Due to centrality, the top and bottom vertical strand should be labeled by the same object $h$. This requires the label $g$ going around the tube to be in the centralizer of $h$, defined as 
$$
    Z(h) := \{ g \in S_3: \; gh=hg \}.
$$
The $m$ and $f$ central idempotents should then be given by
\begin{align*}
    P_m &= \frac{1}{2} \sum_{k} \sum_{g \in Z(r^ks)} T_{g,r^ks,\phi_{-k}} = \frac{1}{2} \left( T_{1,s,\phi_0} + T_{s,s,\phi_0} + T_{1,rs,\phi_2} + T_{rs,rs,\phi_2} + T_{1,r^2s,\phi_1} + T_{r^2s,r^2s,\phi_1} \right), \\ 
    P_f &= \frac{1}{2} \sum_{k} \sum_{g \in Z(r^ks)} \sigma(g) T_{g,r^ks,\phi_{-k}} \\& = \frac{1}{2} \left( T_{1,s,\phi_0} - T_{s,s,\phi_0} + T_{1,rs,\phi_2} - T_{rs,rs,\phi_2} + T_{1,r^2s,\phi_1} - T_{r^2s,r^2s,\phi_1} \right)
\end{align*}
The same three-dimensional anyon in $D(S_3)$ that produces $m$ and $f$ also produces a two-dimensional defect $\chi$, whose idempotent should be given by
$$
    P_\chi = \frac{1}{2} \sum_{j, k \ne -j} T_{1,r^js, \phi_k}.
$$

Next, we need to find the bimodules implementing the symmetry action, which are linear combinations of basis elements $T_{g,h,\phi_j,\phi_k}$ with two blue $A$ strands.
Their left and right actions on the $\Tube_A(\cY)$ elements are given by
\begin{align*}
    T_{g,h,\phi_j,\phi_k} \cdot T_{g',h',\phi_l} = \delta_{h'=g^{-1}hg} \, \delta_{\phi_k=g'(\phi_l)} \, T_{gg',h,\phi_j,\phi_l} \\ 
    T_{g',h',\phi_l} \cdot T_{g,h,\phi_j,\phi_k} = \delta_{h=g'^{-1}h'g'} \, \delta_{\phi_l=\phi_j} \, T_{g'g,h',g' \phi_j,\phi_k}. 
\end{align*}

We can divide the full $\Tube(\cY)_{AA}$ space into two subspaces:
\[
T_0:=\operatorname{span}\{T_{g,h,\phi_j,\phi_k} \mid \phi_j = g \phi_k \}
\qquad\text{and}\qquad
T_1:=\operatorname{span}\{T_{g,h,\phi_j,\phi_k} \mid \phi_j \ne g \phi_k \}.
\]
which are manifestly invariant under the $\Tube(\cY)_A$ bimodule actions. 

Stacking the $1$ and $e$ idempotents with $T_1$ gives the space
\begin{align*}
    P_b T_1 P_a &= \operatorname{span}\Big\{ \sum_{k,k'} \sum_{\substack{g \in \Stab(\phi_k) \\ g' \in \Stab(\phi_{k'})}} c^a_g c^b_{g'} \, T_{g,1,\phi_k} \cdot T_{m,n,\phi_i,\phi_j} \cdot T_{g',1,\phi_{k'}} \mid \phi_i \ne m \phi_j \Big\} \\ 
    %&= \operatorname{span}\Big\{ \sum_{k,k'} \sum_{\substack{g \in \Stab(\phi_k) \\ g' \in \Stab(\phi_{k'})}} c^a_g c^b_{g'} \, \delta_{\phi_i,\phi_k} \delta_{\phi_j = g'(\phi_{k'})} T_{gmg',1,g\phi_i,\phi_{k'}} \mid \phi_i \ne m \phi_j \Big\} \\ 
    %&= \operatorname{span}\Big\{\sum_{\substack{g \in \Stab(\phi_i) \\ g' \in \Stab(\phi_{j})}} c^a_g c^b_{g'} \, T_{gmg',1,g\phi_i,g'^{-1}\phi_{j}} \mid \phi_i \ne m \phi_j  \Big\} \\
    &= \operatorname{span}\Big\{\sum_{\substack{g \in \Stab(\phi_i)\\ g' \in \Stab(\phi_{j})}} c^a_g c^b_{g'} \, T_{gmg',1,\phi_i,\phi_{j}} \mid \phi_i \ne m \phi_j  \Big\}, \quad c^a_g = \begin{cases}
        1 & a = 1 \\ 
        \sigma(g) & a = e
    \end{cases},
\end{align*}
where $a,b \in \{1,e\}$.
The dimension of this space is $9$, as there are three choices for $i$ and $j$ each (there are also $4$ choices for $m$ satisfying $\phi_i \ne m \phi_j$ at fixed $i,j$, but all of these give the same vector, as the coefficients are independent of $m$).

We also need to compute the dimension of the space $P_a \Tube_A(\cY)$, which is given by
\begin{align*}
    P_a \Tube_A(\cY) &= \operatorname{span}\Big\{ \sum_k \sum_{g \in \Stab(\phi_k)} c^a_g \, T_{g,1,\phi_k} T_{g',h',\phi_l} \Big\} \\
    %&=  \operatorname{span}\Big\{ \sum_k \sum_{g \in \Stab(\phi_k)} c^a_g \delta_{\phi_k=g'(\phi_l)} T_{gg',1,\phi_l} \Big\} \\ 
    &=  \operatorname{span}\Big\{ \sum_{g \in \Stab(g'(\phi_l))} c^a_g  \, T_{gg',1,\phi_l} \Big\} 
\end{align*}
Its dimension should be 9, as there are three choices for $l$, and three independent vectors coming from a choice of $g'$ at each fixed $l$. Then according to \eqref{eq:dimensioncount_condensed}, the multiplicities are 
$$
    N_{b,a} = \frac{\dim(P_b T_1 P_1)}{\dim(P_b \Tube_A(\cY))^{1/2} \dim(P_a \Tube_A(\cY))^{1/2}} = \frac{9}{3^2} = 1.
$$

To check that $m,f\mapsto \chi$, we compute $P_a T_1 P_\chi$ for $a \in \{m,f\}$:
\begin{align*}
    P_\chi T_1 P_a &= \operatorname{span}\left\{ \sum_{k,k', j' \ne -k'} \sum_{g \in Z(r^k s)} c^a_g T_{g,r^k s, \phi_{-k}} T_{m,n,\phi_i, \phi_j} T_{1,r^{j'}s, \phi_{k'}} \mid \phi_i \ne m \phi_j \right\} \\ 
    &= \operatorname{span}\left\{
\sum_{g\in Z(r^{-i}s)} c_g^a\, T_{gm,r^{-i}s,\phi_i,\phi_j} \mid r^{j'}s=(gm)^{-1}r^{-i}s\,(gm) \;
\text{ for a } j'\neq -j \right\}. 
\end{align*}
For each fixed pair $(i,j)$, there are two independent vectors, corresponding to the two allowed values of $j'\neq -j$. Since $i$ and $j$ each have three choices, this makes the space $18$-dimensional.
Then we need to compute the dimension of the space $P_\chi \Tube_\cA(\cY)$, which is spanned by
\begin{align*}
    P_\chi\Tube_\cA(\cY) &= \operatorname{span}\left\{ \sum_{k \ne -j} T_{1,r^j s, \phi_k} T_{g,h,\phi_l}\right\} = \operatorname{span}\left\{ T_{g,r^js,\phi_l} \mid j \ne -k \text{ where } \phi_k=g\phi_l  \right\}
\end{align*}
For each pair $(g,j)$, the two allowed choices for $l$ give two independent vectors; therefore, this space has dimension $36$.
Consequently, for $a\in\{m,f\}$, we have
\[N^1_{\chi,a}=\frac{\dim(P_\chi T_1 P_{a})}{\dim(P_a\Tube_A(\cY))^{1/2}\dim(T_1P_{\chi})^{1/2}}=\frac{18}{\sqrt{9}\sqrt{36}}=1.\]

%%%%%%%%%%%%%%%%%%%%%%%%%%%%%%%%%%%%%%%%%%%%%%%%%%%%%%%
\subsection{Enrichment and condensation together}
\label{sec:chiralexample}

We may also implement condensation for string nets and the twist defect story above in the presence of an enrichment in order to recover chiral models.
Here, our braided-enrichment $\textcolor{red}{\cB}$ will always be assumed to be a UMTC so that the center of a $\textcolor{red}{\cB}$-enriched UFC $\cY$ always factorizes as 
$$
Z(\cY)\cong \textcolor{red}{\cB} \boxtimes Z^{\textcolor{red}{\cB}}(\cY).
$$
Now if $\textcolor{blue}{A} \in Z^{\textcolor{red}{\cB}}(\cY)$ is a condensable algebra, then since it clearly centralizes $\textcolor{red}{\cB}$, one may think of it as acting on the `other side' of the 2D $\cY$ domain wall between $\textcolor{red}{\cB}$ and $Z^{\textcolor{red}{\cB}}(\cY)^{\rev}$.
We will thus think of the blue $\textcolor{blue}{A}$-strings as coming out of the page, while the $\textcolor{red}{\cB}$-enriching strings go into the page.

Recall that $p^{\textcolor{red}{\cB}} \in \Tube(\cY)$
and $p_{\textcolor{blue}{A}}\in \Tube_A(\cY)$
are central projectors, and the corresponding corners are used to detect the surface anyons in the enriched model and condensation model respectively.
Under the inclusion $\Tube(\cY) \hookrightarrow \Tube_A(\cY)$ given by tensoring with the unit map $i: 1_\cY\to A$ from \eqref{eq:IncludeTubeIntoATube}, the image of $p^{\textcolor{red}{\cB}}$ commutes with $p_{\textcolor{blue}{A}}$ since $A\in Z^{\textcolor{red}{\cB}}(\cY)$.
Thus the product $p_{\textcolor{blue}{A}}p^{\textcolor{red}{\cB}}$ is again central in $\Tube_A(\cY)$, and we have a commuting square of finite dimensional algebras related by compression:
\[
\begin{tikzcd}
    \Tube_A(\cY) \arrow[r,twoheadrightarrow,"p_A"] \arrow[d,twoheadrightarrow,"p^{\textcolor{red}{\cB}}"] &
    \Tube_A^{\loc}(\cY) \arrow[d,twoheadrightarrow,"p^{\textcolor{red}{\cB}}"] \\
    \Dome_A^{\textcolor{red}{\cB}}(\cY) \arrow[r,twoheadrightarrow,"p_A"] &
    \Dome_{A,\loc}^{\textcolor{red}{\cB}}(\cY)
\end{tikzcd}
\]

\begin{lem}
\label{lem:RepDomeAB==ZBYA}
    If $\cY$ is a $\cB$-enriched multifusion category and $\textcolor{blue}{A}\in Z^{\textcolor{red}{\cB}}(\cY)$ is a condensable algebra, then the twist defects in the enriched condensation model correspond to
    $\Rep(\Dome^{\textcolor{red}{\cB}}_{\textcolor{blue}{A}}(\cY))\cong Z^{\textcolor{red}{\cB}}(\cY)_{\textcolor{blue}{A}}\cong Z_{\cY_{\textcolor{blue}{A}}}^{\textcolor{red}{\cB}}(\cY)$.
\end{lem}
\begin{proof}
    By definition, $Z(\cY)=Z^{\textcolor{red}{\cB}}(\cY)\boxtimes\textcolor{red}{\cB}$.
    Since $\textcolor{blue}{A}$ centralizes $\textcolor{red}{\cB}$,
    \[Z(\cY)_{\textcolor{blue}{A}}\cong Z^{\textcolor{red}{\cB}}(\cY)_{\textcolor{blue}{A}}\boxtimes\cB\]
    We know already that $\Rep(\Tube_{\textcolor{blue}{A}}(\cY))\cong Z(\cY)_{\textcolor{blue}{A}}$ by \cite[\S4.3.3]{MR4642306},
    and applying the projector $p_{\textcolor{blue}{A}}$ onto the corner $\Dome^{\textcolor{red}{\cB}}_{\textcolor{blue}{A}}(\cY)$ leaves only $Z^{\textcolor{red}{\cB}}(\cY)_{\textcolor{blue}{A}}$.

    The rest follows as in the end of \S\ref{sec:ExcitationsOfATube}:
    \begin{align*}
    \Rep(\Dome_{\textcolor{blue}{A}}^{\textcolor{red}{\cB}}(\cY))
    &\cong
    Z^{\textcolor{red}{\cB}}(\cY)_{\textcolor{blue}{A}}
    \\&\cong
    \End^{\textcolor{red}{\cB}}_{\cY-\cY_{\textcolor{blue}{A}}}(\cY_{\textcolor{blue}{A}})
    &&
    \text{(\cite[Ex.~III.2]{MR4640433})}
    \\&\cong
    \Hom_{\cY_{\textcolor{blue}{A}}-\cY_{\textcolor{blue}{A}}}^{\textcolor{red}{\cB}}(\cY_{\textcolor{blue}{A}}\to \cY) 
    && 
    \text{(dualizability in $\UmFC^{\textcolor{red}{\cB}}$)}
    \\&=:
    Z_{\cY_{\textcolor{blue}{A}}}^{\textcolor{red}{\cB}}(\cY).
    &&\qedhere
\end{align*}
\end{proof}

In the special case that $\cY=\cB^{\rev}$, which is the typical way one realizes chiral models as 2D boundaries of 3D models as in \cite{1104.2632}, we can further simplify $\Dome_A^{\textcolor{red}{\cB}}(\cY)$ as follows.
Compressing by $p^{\textcolor{red}{\cB}}$ identifies each element of $\Dome_A^{\cB}(\cY)$ with a tube that has the trivial string running around the back.
Therefore, the elements of $\Dome^{\textcolor{red}{\cB}}_A(\cY)$ corresponding to the twist defects can be pictured as only a vertical spine and a blue $A$-string:
\[
\Dome_{{\color{blue}{A}}}^{\textcolor{red}{\cB}}(\cY)
\ni
\hspace*{-.3cm}
\tikzmath{
\draw[thick] (.5,-.6) -- (.5,1);
\draw[thick] (-.5,-.6) -- (-.5,1);
\draw[thick] (0,1) ellipse (.5 and .2);
\draw[thick] (.5,-.6) arc(0:-180:.5 and .2);
\draw[thick,dotted] (.5,-.6) arc(0:180:.5 and .2);
\draw[] (0,-.8) --node[left]{$\scriptstyle Y$} (0,-.3);
\draw[] (0,.8) --node[left]{$\scriptstyle Y$} (0,.3);
\draw[thick, blue, knot] (0,0) -- (-.3,.3) to[out=135, in=-90] (-.8,1.2) node[left, yshift=-.3cm] {$\scriptstyle A$};
\roundNbox{fill=white}{(0,0)}{.3}{0}{0}{$f$}
} 
\qquad\longleftrightarrow\quad
\tikzmath{
\draw[] (0,-.8) --node[left]{$\scriptstyle Y$} (0,-.3);
\draw[] (0,.8) --node[left]{$\scriptstyle Y$} (0,.3);
\draw[thick, blue, knot] (0,0) -- (-.3,.3) to[out=135, in=-90] (-.8,.8) node[left, yshift=-.3cm] {$\scriptstyle A$};
\roundNbox{fill=white}{(0,0)}{.3}{0}{0}{$f$}
}
\in \cY(Y\to A\otimes Y).
\]
The usual argument from \cite[Fig.~4]{MR1936496} shows that this algebra is isomorphic to $D_A^\cB:=\End_{\cY_A}(A\otimes Y)$.\footnote{
In fact, the equivalence $\Dome_A^{\cB}(\cY)\cong\End_{Z^{\cB}(\cY)_A}(A\otimes Y)$ holds even when $\cY\not\cong\overline{\cB}$ via \eqref{eq:tubeAViaTrace}.
}
Observe that any $A$-module $M_A\in \cY_A$ admits a non-zero $A$-module map to $A\otimes Y$ for $Y=\bigoplus_{y\in \Irr(\cY)} y$.

To get the algebra $\Dome_{A,\loc}^\cB(\cY)$ whose central idempotents characterize the anyons, we may either quotient out by an ideal (left hand side below) or look at the image of a central idempotent (right hand side below).
(This second description uses that $Z(\cY)\cong \cY\boxtimes \cB$, and so forgetting $A\in \cY$ down to $\cY$ means we get the reverse braiding as in $\cB$.)
$$
\cI=
\left\langle
    \tikzmath{
\draw[] (0,-.8) --node[left]{$\scriptstyle Y$} (0,-.3);
\draw[] (0,.8) --node[left]{$\scriptstyle Y$} (0,.3);
\draw[thick, blue, knot] (0,0) -- (-.3,.3) to[out=135, in=-90] (-.8,.8) node[left, yshift=-.3cm] {$\scriptstyle A$};
\roundNbox{fill=white}{(0,0)}{.3}{0}{0}{$f$}
} 
-
\tikzmath{
\draw[] (0,-.8) --node[left]{$\scriptstyle Y$} (0,-.3);
\draw[] (0,.8) --node[left]{$\scriptstyle Y$} (0,.3);
\draw[thick, blue, knot] 
plot[smooth] coordinates {(0,0) (-0.2,0.4) (0,0.6) (0.2,0.8) (0,1.0) (-0.2,1.2)};
\roundNbox{fill=white}{(0,0)}{.3}{0}{0}{$f$}
\draw[black, knot] (0,0.8) -- (0,1.2);
}
\right\rangle
\qquad\qquad\qquad\qquad
p
=
\tikzmath{
\draw (0,-.7) --node[right]{$\scriptstyle Y$} (0,-.3)-- (0,0);
\draw[thick, blue, knot] (0,0) circle (.2cm);
\draw[knot] (0,.7) --node[right]{$\scriptstyle Y$} (0,.3)--(0,0);
\draw[thick, blue] (135:.2cm) to[out=135,in=-90] (-.5,.7) node[left, yshift=-.2cm]{$\scriptstyle A$};
\filldraw[blue] (135:.2cm) circle (.05cm);
}\,.
$$

\subsection{Symmetry action in enriched condensation model}
\label{ssec:condensationSymmetryAction}

The story in \S\ref{sec:CondensationDefects} easily extends to the $\cB$-enriched setting.
We compute the non-invertible symmetry action of sweeping a defect across a anyon in the condensed theory by computing induction and restriction:
\begin{equation}
\label{eq:EnrichedCondensationIndRes}
Z^{\textcolor{red}{\cB}}(\cY_A)
=
Z^{\textcolor{red}{\cB}}(\cY)_A^{\loc}
\hookrightarrow
Z^{\textcolor{red}{\cB}}(\cY)_A
\twoheadrightarrow
Z^{\textcolor{red}{\cB}}(\cY)
\to
Z^{\textcolor{red}{\cB}}(\cY)_A
=
Z^{\textcolor{red}{\cB}}_{\cY_A}(\cY).
\end{equation}
Each of these categories is realized as the representations of one of our $\Dome$ algebras, and each functor between representation categories is given by tensoring with a bimodule obtained from the dual inclusions and projections between them:
\[\Dome^{\textcolor{red}{\cB}}_{A,\loc}(\cY)\overset{p_A}{\twoheadleftarrow}\Dome^{\textcolor{red}{\cB}}_A(\cY)\hookleftarrow\Dome^{\textcolor{red}{\cB}}(\cY)\hookrightarrow\Dome^{\textcolor{red}{\cB}}_A(\cY)\]
In particular, we may compute the relative tensor product
\[\Dome^{\textcolor{red}{\cB}}_A(\cY)\boxtimes_{\Dome^{\textcolor{red}{\cB}}(\cY)}\Dome^{\textcolor{red}{\cB}}_A(\cY)
\cong
\Dome^{\textcolor{red}{\cB}}_{AA}(\cY)
\]
where the left and right actions of $\Dome_A^{\cB}(\cY)$ are given by the inclusions
\[
\Dome^{\textcolor{red}{\cB}}_{A}(\cY)
\cong
\Dome^{\textcolor{red}{\cB}}_{A1}(\cY)
\hookrightarrow
\Dome^{\textcolor{red}{\cB}}_{AA}(\cY)
\hookleftarrow
\Dome^{\textcolor{red}{\cB}}_{1A}(\cY)
\cong
\Dome^{\textcolor{red}{\cB}}_{A}(\cY)
\]
The functor \eqref{eq:EnrichedCondensationIndRes} is the equivalent to the composition
\[Z^{\textcolor{red}{\cB}}(\cY)_A^{\loc}\hookrightarrow Z^{\textcolor{red}{\cB}}(\cY)_A\cong Z^{\textcolor{red}{\cB}}(\cY)_{A1}\hookrightarrow Z^{\textcolor{red}{\cB}}(\cY)_{AA}\twoheadrightarrow Z^{\textcolor{red}{\cB}}(\cY)_{1A}\cong Z^{\textcolor{red}{\cB}}(\cY)_A\]
Thus, we can use stacking of tube algebras to compute the symmetry action, analogous to \eqref{eq:fullInclusionTubeStack}.
An anyon type $a\in\Irr(Z^{\textcolor{red}{\cB}}(\cY_A))$ corresponds to a minimal central idempotent $P_a\in\Dome^{\textcolor{red}{\cB}}_{A,\loc}(\cY)$, and a twist defect type $b\in\Irr(Z^{\textcolor{red}{\cB}}(\cY)_A)$ corresponds to a minimal central idempotent $P_b\in\Dome^{\textcolor{red}{\cB}}_A(\cY)$.
We then compute 
the space
\[P_b\Dome^{\textcolor{red}{\cB}}_{AA}(\cY) P_a;
\]
if the result is nonzero, then the anyon or twist defect $b$ appears in the wrapped domain wall excitation.

%%%%%%%%%%%%%%%%%%%%%%%%%%%%%%%%%%%%%%%%%%%%%%%%%%%%%%%%%%%
\subsubsection{Example: chiral and condensation example with \texorpdfstring{$\cY=SU(2)_{10}$}{Y=SU(2)10}}
\label{sssec:chiralExample}
For an example with chiral topological order and condensation, we take 
$$
\cY=SU(2)_{10} \quad \text{and} \quad \cX=\cY_A = E_6 \quad \text{for } A = 0 \oplus 6,
$$
using the convention that the simple objects in $SU(2)_{10}$ are labeled by the integers $0,1, \dots,10$. Also, we will attach a $\cB = \overline{SU(2)}_{10}$ bulk to the string-net, so that the anyons live in the chiral theory $Z^\cB(\cX) = \mathsf{Ising}$.
The twist defects live in the relative enriched center $Z^{\cB}_{\cX}(\cY)$.
In this case, $\cY$ can be equipped with a nondegenerate braiding, so we have $Z(\cY)\cong\cY\boxtimes\overline{\cY}$, and setting $\cB=\overline{\cY}$ gives $Z^{\cB}(\cY)\cong\cY$.
Since $Z^{\cB}_{\cY_A}(\cY)\cong Z^{\cB}(\cY)_A$, in this case we have $Z^{\cB}_{\cY_A}(\cY)\cong\cY_A$, meaning that the twist defects are described by $\cY_A\cong E_6$.
In particular, the anyons are given by the modular tensor category $\cY_A^{\loc}\cong\mathsf{Ising}$ \cite[Cor.~4.9]{MR1815993} cf.~\cite[Cor.~3.30]{MR3039775}.

In this particular example, we are in the fortunate case that every simple module admits an isometric embedding into $A\otimes Y$.
Also, since simple objects in $Z(\cY)$ forget to simple objects in $\cY$, we have $\Dome^{\cB}(\cY)\cong\operatorname{span}\{\id_y:y\in\Irr(\cY)\}$.
Each $\id_y$ is therefore a projector which is both central and rank $1$. 
Moreover, all simple twist defects except for the non-abelian anyon $\sigma$ are free modules.
Data about $E_6$ twist defects and the corresponding $A$-modules were computed in \cite[pg.~33]{MR1936496}, and is reproduced in Figure~\ref{fig:cinclusionTable}.
\begin{figure}[!ht]
    \centering
    \[\begin{array}{|c|c|c|}
     \hline
     c\in\Irr(\cY_A) & xA\cong c & U(c)\\
     \hline
     % \rowcolor{blue!10}
     1 & 0A & 0\oplus 6 \\\hline
     % \rowcolor{blue!10}
     \psi & (10)A & 4\oplus 10 \\\hline
     \rowcolor{orange!30} \sigma & - & 3\oplus 7\\\hline
     a & 1A & 1\oplus 5\oplus 7 \\\hline
     x & 2A & 2\oplus 4\oplus 6\oplus 8 \\\hline
     b & 9A & 3\oplus 5\oplus 9\\\hline
    \end{array}\]
    \caption{\label{fig:cinclusionTable} $E_6$ twist defects (including anyons) written as free modules on $SU(2)_{10}$ anyons, and the underlying direct sum of $SU(2)_{10}$ anyons. The anyon which is not a free module is highlighted in orange. }
\end{figure}

Therefore, by Corollary \ref{cor:freeModuleRank1},
we can find rank $1$ projectors $p_b\in P_b\Dome^{\cB}_A(\cY)$ for every $b\in\Irr(Z^{\cB}(\cY)_A)$ except $\sigma$.
\begin{itemize}
\item 
The other two Ising anyons correspond to the local simple free $A$-modules, with rank $1$ idempotents given by:
\[
    p_1 = \tikzmath{
    \draw[thick, dotted] (0,0) node[below] {$\scriptstyle 0$}-- (0,1);
    \draw[thick, blue] (-.4,0.5) -- (-0.4,1) node[left] {$\scriptstyle A$};
    \fill[blue] (-.4,0.5) circle (0.6mm);
    }, \qquad p_\psi = \tikzmath{
    \draw[thick] (0,0) node[below] {$\scriptstyle 10$}-- (0,1);
    \draw[thick, blue] (-.4,0.5) -- (-0.4,1) node[left] {$\scriptstyle A$};
    \fill[blue] (-.4,0.5) circle (0.6mm);
    }
\]
\item
Similarly, the twist defects correspond to the nonlocal simple free $A$-modules, so some rank $1$ idempotents are given by:
\begin{align*}
    p_a &=  \tikzmath{
    \draw[thick] (0,0) node[below] {$\scriptstyle 1$}-- (0,1);
    \draw[thick, blue] (-.4,0.5) -- (-0.4,1) node[left] {$\scriptstyle A$};
    \fill[blue] (-.4,0.5) circle (0.6mm);
    },
    % M_a=1 \oplus 5 \oplus 7, 
    \quad p_x =  \tikzmath{
    \draw[thick] (0,0) node[below] {$\scriptstyle 2$}-- (0,1);
    \draw[thick, blue] (-.4,0.5) -- (-0.4,1) node[left] {$\scriptstyle A$};
    \fill[blue] (-.4,0.5) circle (0.6mm);
    }, 
    % M_x=2 \oplus 4 \oplus 6 \oplus 8,
    \quad p_b =  \tikzmath{
    \draw[thick] (0,0) node[below] {$\scriptstyle 9$}-- (0,1);
    \draw[thick, blue] (-.4,0.5) -- (-0.4,1) node[left] {$\scriptstyle A$};
    \fill[blue] (-.4,0.5) circle (0.6mm);
    }
    % , M_b=3 \oplus 5 \oplus 9
    .
\end{align*}
\end{itemize}
Computing the minimal central idempotent $P_\sigma$ or a rank $1$ idempotent in $P_\sigma\Dome^{\cB}_{A,\loc}(\cY)$ is more difficult, since $\sigma$ is the only non-free module.
However, careful dimension-counting and the use of Corollary~\ref{cor:freeModuleProjRanks} will let us compute the symmetry action coefficients while bypassing this detail.

Here we again have a trivial symmetry generator $\pi_+$ and a nontrivial one $\pi_-$.
Since $\pi_+$ corresponds to the identity symmetry, we know that $N^+_{a,b}=\delta_{a=b}$.
With the above choices of $p_{xA}$ for $xA\in\Irr(\cY_A)\setminus\{\sigma\}$, it is straightforward to see that
\[\dim(p_{yA}\Dome^{\cB}_{AA}(\cY)p_{xA})=\dim(\cY(Ax\to Ay))\]
We can therefore compute all the symmetry action coefficients $N^-_{j,k}$ for $j,k\notin\{\sigma\}$:
\[N^-_{1,x}=N^-_{x,1}=1,\, N^-_{\psi,x}=N^-_{x,\psi}=1,\, N^-_{a,b}=N^-_{b,a}=1,\]
with all other $N^-_{j,k}$ equal to $0$.

Finally, to find the coefficients $N^-_{\sigma,\cdot}$, we take advantage of the fact that $3A\cong b\oplus\sigma$ (or symmetrically, that $7A\cong a\oplus\sigma$).
Consider the projection
\[q:=\tikzmath{
    \draw[thick] (0,0) node[below] {$\scriptstyle 3$}-- (0,1);
    \draw[thick, blue] (-.4,0.5) -- (-0.4,1) node[left] {$\scriptstyle A$};
    \fill[blue] (-.4,0.5) circle (0.6mm);
    }\]
Although the free module $3A$ is not simple, Lemma~\ref{lem:FreeModuleFunctorViaTubes} says that $q$ is still a projection, and that $q=q_{\sigma}+q_b$, where
$q_{\sigma}:=qP_{\sigma}$ and $q_b:=qP_b$ are rank $1$ projections in their respective blocks of $\Tube_A(\cY)$ by Corollary~\ref{cor:freeModuleProjRanks}.
Therefore, recalling that $N^+_{a,\sigma}=N^+_{a,b}=0$, we obtain
$$
    \dim(q \Dome^\cB_{AA}(\cY) p_a) = N^-_{a,\sigma} + N^-_{a,b}.
$$
The left hand side is equal to $\dim(\cY(A1 \to A3)) = 2$. Since we already know that $N^-_{a,b}=1$, it follows that we must have $N^-_{a,\sigma} = N^-_{\sigma,a} = 1$. An analogous argument determines $N^-_{\sigma,b} = N^-_{b,\sigma} = 1$. Thus the non-invertible symmetry maps all Ising anyons to twist defects:
$$
    1 \mapsto x, \qquad \psi \mapsto x, \qquad \sigma \mapsto a \oplus b.
$$

\bibliographystyle{alpha}
{\footnotesize{
\bibliography{bibliography}
}}
\end{document}